\documentclass[letterpaper]{amsart}

\usepackage[letterpaper]{geometry}
\usepackage{amsmath, amsthm, amsfonts, amsbsy, thmtools, amssymb}

\usepackage{thm-restate}
\usepackage{hyperref} 
\usepackage[]{ytableau}
\usepackage[mathscr]{euscript}
\usepackage{stmaryrd}
\usepackage[all, arc, curve, frame]{xy} 
\usepackage[T1]{fontenc}
\usepackage{tikz}
\usetikzlibrary{arrows}

\makeatletter
\def\namedlabel#1#2{\begingroup
	#2%
	\def\@currentlabel{#2}%
	\phantomsection\label{#1}\endgroup
}
\makeatother

\numberwithin{equation}{section}

\SetSymbolFont{stmry}{bold}{U}{stmry}{m}{n}

\DeclareMathOperator{\Sign}{Sign}

\DeclareMathOperator{\Id}{Id}

\DeclareMathOperator{\b|}{\boldsymbol{|}}

\DeclareMathOperator{\TW}{TW}

\title{Dynamical Stochastic Higher Spin Vertex Models} 

\author{Amol Aggarwal} 

\begin{document}

\newtheorem{thm}{Theorem}[section]
\newtheorem{prop}[thm]{Proposition}
\newtheorem{lem}[thm]{Lemma}
\newtheorem{cor}[thm]{Corollary}
\newtheorem{conj}[thm]{Conjecture}
\newtheorem{que}[thm]{Question}
\theoremstyle{remark}
\newtheorem{rem}[thm]{Remark}
\theoremstyle{definition}
\newtheorem{definition}[thm]{Definition}
\newtheorem{example}[thm]{Example}

\maketitle

\begin{abstract}
	
	We introduce a new family of integrable stochastic processes, called \textit{dynamical stochastic higher spin vertex models}, arising from fused representations of Felder's elliptic quantum group $E_{\tau, \eta} (\mathfrak{sl}_2)$. These models simultaneously generalize the stochastic higher spin vertex models, studied by Corwin-Petrov and Borodin-Petrov, and are dynamical in the sense of Borodin's recent stochastic interaction round-a-face models. 
	
	We provide explicit contour integral identities for observables of these models (when run under specific types of initial data) that characterize the distributions of their currents. Through asymptotic analysis of these identities in a special case, we evaluate the scaling limit for the current of a dynamical version of a discrete-time partial exclusion process. In particular, we show that its scaling exponent is $1 / 4$ and that its one-point marginal converges (in a sense of moments) to that of a non-trivial random variable, which we determine explicitly.
\end{abstract}

\tableofcontents

\section{Introduction}

\label{Introduction}

The search for exactly solvable (or \textit{integrable}) systems has long played a prominent role in mathematical physics. Our goal in this paper is to advance further in this direction by introducing a new family of \textit{stochastic} integrable models, which we term the \textit{dynamical stochastic higher spin vertex models}. In our view, the appeals to these models are threefold. 

The first is that they are integrable, in that one can produce explicit (contour integral) identities for a large class of observables for these models. 

The second is that they are \textit{dynamical}. Specifically, associated with each site is a \textit{dynamical parameter} that depends on both the location of the site and the state of the model. At any given site, the Markov transition functions of the model are dependent on the corresponding dynamical parameter. This produces a process that exhibits heavy spatial and temporal inhomogeneities, and one question of interest is how these inhomogeneities affect the large scale properties of the process.

Although dynamical integrable systems have been studied extensively in the statistical physics literature (see, for instance, Chapter 13 of \cite{ESMSM} and references therein) under the name of \textit{interaction round-a-face} (IRF) models, they are typically (with the exception of the recent ones proposed in \cite{ERSF}) not stochastic. Consequently, few results or predictions have been proposed towards the question asked above. As an application of our framework, we show how the integral identities mentioned previously can be analyzed to understand the limiting statistics of a special case of our dynamical model; we will see that this result indicates that the effect of the dynamical parameter on the long-time behavior of the model is quite significant.	

The third appeal is that our models are quite general. In particular, they are a one-parameter deformation of a four-parameter family of stochastic processes introduced by Corwin and Petrov \cite{SHSVML} called the \textit{stochastic higher spin vertex models}. Those models are already very extensive and, for example, are known to exhibit limit degenerations to most models proven to be in the Kardar-Parisi-Zhang (KPZ) universality class. Our models also comprise a one-parameter deformation of a family of stochastic processes recently introduced and studied by Borodin \cite{ERSF}, which he called \textit{stochastic IRF models}. In fact, our analysis is based on a combination of the results from \cite{ERSF} and the \textit{fusion procedure} applied to the elliptic quantum group $E_{\tau, \eta} (\mathfrak{sl}_2)$. 

Before proceeding to a detailed explanation of our model and results, we will attempt to provide a more tangible description of the term ``dynamical.'' Thus, we begin in Section \ref{DynamicalCornerGrowth} with a simple example of a dynamical stochastic process, which is a certain discrete-time corner growth model. Although this is a very specific case of the far more general model to be introduced later (in Section \ref{ModelGeneral}), it already exhibits a number of nonstandard and intricate integrability and scaling phenomena that are not present in non-dynamical stochastic growth models. 

We then define the dynamical stochastic higher spin vertex models in Section \ref{ModelGeneral}. In Section \ref{DegenerationsParticle} we elaborate two degenerations of this model, one to a dynamical variant of the $q$-Hahn boson model (the non-dynamical version of which was introduced by Povolotsky in \cite{IZRCMFSS}), and the other to a dynamical partial exclusion process (which is a generalization of the corner growth model defined in Section \ref{DynamicalCornerGrowth}). Next, in Section \ref{ObservablesAsymptoticsDynamicalExclusion} we describe integral identities for observables of these models and how they can be used to access asymptotics of the dynamical partial exclusion process. We conclude in Section \ref{Outline} with a brief outline of what is used in the proofs of these results.

\subsection{The Dynamical Midpoint Corner Growth Model} 

\label{DynamicalCornerGrowth}

In this section we introduce a discrete-time corner growth model that can be viewed as a prototypical dynamical stochastic process. However, before describing the dynamical version of this model, let us first define its non-dynamical variant, which we call the \textit{midpoint corner growth model}. 

The model we define $\{ \zeta_t (x) \}_{t \in \mathbb{Z}_{\ge 0}, x \in \mathbb{R}}$ will take place on piecewise linear curves, in which the slope of each segment is either $-2$, $0$, or $2$. The initial data for this model is prescribed by $\zeta_0 (x) = 2 |x|$. Associated with each curve will be a set of \textit{vertices} in $\mathbb{R}^2$, given by $V_t = V_t (\zeta_t) = \big\{ \big( k + \frac{t}{2}, \zeta_t \big( k + \frac{t}{2} \big) \big) \big\}_{k \in \mathbb{Z}} \subset (\mathbb{R} \times \mathbb{Z})^{\mathbb{Z}}$. For example, $V_0 = \big\{ \big(k, 2 |k| \big) \big\}_{k \in \mathbb{Z}}$. 

Now, let us explain how to sample $\zeta_{t + 1}$ given $\zeta_t$, for any $t \ge 0$. To that end, it suffices to prescribe the vertex set $V_{t + 1}$ from $V_t$; connecting adjacent vertices in $V_{t + 1}$ then gives rise to $\zeta_{t + 1}$. Suppose that $z_1 = \big (a, \zeta_t (a) \big)$ and $z_2 = \big( a + 1, \zeta_t (a + 1) \big)$ are adjacent vertices in $V_t$. Then, if the slope of the segment connecting $z_1$ and $z_2$ is $-2$ or $2$, we deterministically append its midpoint $z = (z_1 + z_2) / 2$ to $V_{t + 1}$. If the slope of this segment is $0$, then we will either append $z + (0, 1)$ or $z + (0, -1)$ to $V_{t + 1}$, the former with some probability $p \in [0, 1]$ and the latter with probability $1 - p$; stated alternatively, this midpoint can either go ``up'' with probability $p$ or ``down'' with probability $1 - p$. Repeating this procedure for all $a \in \frac{t}{2} + \mathbb{Z}$ produces the vertex set $V_{t + 1}$ and thus the curve $\zeta_{t + 1}$. Observe that $\zeta_t (x) \ge 2 |x|$ for each $t$ and $x$, and in fact $\zeta_t (x) = 2 |x|$ for $x \notin  [-t, t]$. An example of this corner growth model at times $0, 1, 2, 3$ is depicted in Figure \ref{corner}. 

Let us explain that figure in more detail. At time $0$, the curve $\zeta_0 (x) = 2 |x|$, and vertices are placed at lattice points of the form $(x, 2 |x|)$. At time $1$, the midpoints of each of the segments connecting adjacent vertices are taken; this produces a deterministic vertex set $V_1$ and therefore a deterministic curve $\zeta_1 (x)$, equal to $2 |x|$ for $|x| \ge \frac{1}{2}$ and equal to $1$ otherwise.  

At time $t = 2$, the curve $\zeta_2 (x)$ is no longer deterministic. Since the segment connecting $z_1 = \big( -\frac{1}{2}, 1 \big)$ to $z_2 = \big(\frac{1}{2}, 1 \big)$ is horizontal with midpoint $z = (0, 1)$, a choice can be made to either append $z + (0, 1) = (0, 2)$ or $z + (0, -1) = (0, 0)$ to $V_2$; the former event (depicted in the figure) occurs with probability $p$ and the latter with probability $1 - p$. Away from the vertex with $x$-coordinate $0$, the sets $V_0$ and $V_2$ coincide, giving rise to the curve $\zeta_2$ shown in Figure \ref{corner}. At time $3$, there are two choices to be made, corresponding to whether the midpoints of the two horizontal segments go up or down; the situation depicted is when the leftmost midpoint goes down and the rightmost midpoint goes up. 

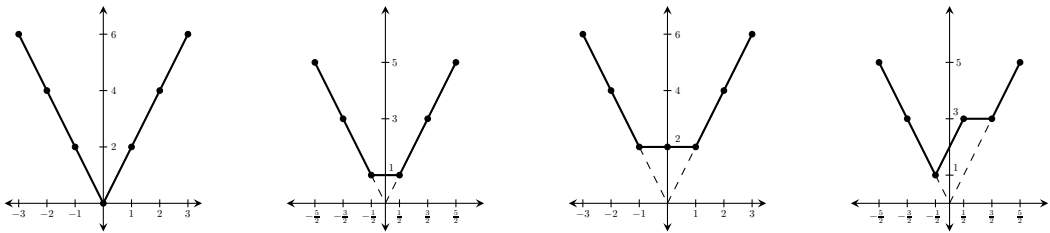
\begin{figure}
	
	\begin{center}

		\begin{tikzpicture}[
		>=stealth,
		scale = .75
		]

		\draw[<->, black] (1.5, -.5) -- (1.5, 3.5);
		\draw[<->, black] (-.25, 0) -- (3.25, 0);
		
		\draw[-, black, thick] (0, 3) -- (.5, 2);
		\draw[-,black, thick] (.5, 2) -- (1, 1);
		\draw[-,black, thick] (1, 1) -- (1.5, 0);
		\draw[-,black, thick] (1.5, 0) -- (2, 1);
		\draw[-,black, thick] (2, 1) -- (2.5, 2);
		\draw[-, black, thick] (2.5, 2) -- (3, 3);

		\draw[-, black] (0, -.07) -- (0, .07) node[scale = .4, below = 7]{$-3$}; 
		\draw[-, black] (.5, -.07) -- (.5, .07) node[scale = .4, below = 7]{$-2$}; 
		\draw[-, black] (1, -.07) -- (1, .07) node[scale = .4, below = 7]{$-1$}; 
		\draw[-, black] (2, -.07) -- (2, .07) node[scale = .4, below = 7]{$1$}; 
		\draw[-, black] (2.5, -.07) -- (2.5, .07) node[scale = .4, below = 7]{$2$}; 
		\draw[-, black] (3, -.07) -- (3, .07) node[scale = .4, below = 7]{$3$}; 
		
		\draw[-, black] (1.43, 1) -- (1.57, 1) node[scale = .4, right]{$2$}; 
		\draw[-, black] (1.43, 2) -- (1.57, 2) node[scale = .4, right]{$4$}; 
		\draw[-, black] (1.43, 3) -- (1.57, 3) node[scale = .4, right]{$6$}; 
		
		\filldraw[fill=black, draw=black] (0, 3) circle [radius=.05]; 
		\filldraw[fill=black, draw=black] (.5, 2) circle [radius=.05];
		\filldraw[fill=black, draw=black] (1, 1) circle [radius=.05];
		\filldraw[fill=black, draw=black] (1.5, 0) circle [radius=.05];
		\filldraw[fill=black, draw=black] (2, 1) circle [radius=.05];
		\filldraw[fill=black, draw=black] (2.5, 2) circle [radius=.05];
		\filldraw[fill=black, draw=black] (3, 3) circle [radius=.05];

		\draw[<->, black] (6.5, -.5) -- (6.5, 3.5);
		\draw[<->, black] (4.75, 0) -- (8.25, 0);

		\draw[-,black, thick] (5.25, 2.5) -- (5.75, 1.5);
		\draw[-,black, thick] (5.75, 1.5) -- (6.25, .5);
		\draw[-,black, thick] (6.25, .5) -- (6.75, .5);
		\draw[-,black, thick] (6.75, .5) -- (7.25, 1.5);
		\draw[-, black, thick] (7.25, 1.5) -- (7.75, 2.5);
		\draw[-, black, dashed] (6.25, .5) -- (6.5, 0); 
		\draw[-, black, dashed] (6.5, 0) -- (6.75, .5); 
		
		\draw[-, black] (6.43, .5) -- (6.57, .5) node[scale = .4, above = 7, right = -4]{$1$}; 
		\draw[-, black] (6.43, 1.5) -- (6.57, 1.5) node[scale = .4, right = -1]{$3$}; 
		\draw[-, black] (6.43, 2.5) -- (6.57, 2.5) node[scale = .4, right = -1]{$5$}; 
		
		\draw[-, black] (5.25, -.07) -- (5.25, .07) node[scale = .4, left = 2, below = 7]{$-\frac{5}{2}$}; 
		\draw[-, black] (5.75, -.07) -- (5.75, .07) node[scale = .4, left = 2, below = 7]{$-\frac{3}{2}$}; 
		\draw[-, black] (6.25, -.07) -- (6.25, .07) node[scale = .4, left = 2, below = 7]{$-\frac{1}{2}$}; 
		\draw[-, black] (6.75, -.07) -- (6.75, .07) node[scale = .4, below = 7]{$\frac{1}{2}$}; 
		\draw[-, black] (7.25, -.07) -- (7.25, .07) node[scale = .4, below = 7]{$\frac{3}{2}$}; 
		\draw[-, black] (7.75, -.07) -- (7.75, .07) node[scale = .4, below = 7]{$\frac{5}{2}$}; 
		
		\filldraw[fill=black, draw=black] (6.25, .5) circle [radius=.05];
		\filldraw[fill=black, draw=black] (6.75, .5) circle [radius=.05];
		\filldraw[fill=black, draw=black] (5.75, 1.5) circle [radius=.05];
		\filldraw[fill=black, draw=black] (7.25, 1.5) circle [radius=.05];
		\filldraw[fill=black, draw=black] (7.75, 2.5) circle [radius=.05];
		\filldraw[fill=black, draw=black] (5.25, 2.5) circle [radius=.05];

		\draw[<->, black] (11.5, -.5) -- (11.5, 3.5);
		\draw[<->, black] (9.75, 0) -- (13.25, 0);
		
		\draw[-, black] (10, -.07) -- (10, .07) node[scale = .4, below = 7]{$-3$}; 
		\draw[-, black] (10.5, -.07) -- (10.5, .07) node[scale = .4, below = 7]{$-2$}; 
		\draw[-, black] (11, -.07) -- (11, .07) node[scale = .4, below = 7]{$-1$}; 
		\draw[-, black] (12, -.07) -- (12, .07) node[scale = .4, below = 7]{$1$}; 
		\draw[-, black] (12.5, -.07) -- (12.5, .07) node[scale = .4, below = 7]{$2$}; 
		\draw[-, black] (13, -.07) -- (13, .07) node[scale = .4, below = 7]{$3$}; 
		
		\draw[-, black] (11.43, 1) -- (11.57, 1) node[scale = .4, above = 8, right]{$2$}; 
		\draw[-, black] (11.43, 2) -- (11.57, 2) node[scale = .4, right]{$4$}; 
		\draw[-, black] (11.43, 3) -- (11.57, 3) node[scale = .4, right]{$6$};

		\draw[-,black, thick] (10, 3) -- (10.5, 2);
		\draw[-,black, thick] (10.5, 2) -- (11, 1);
		\draw[-,black, thick] (11, 1) -- (11.5, 1);
		\draw[-,black, thick] (11.5, 1) -- (12, 1);
		\draw[-,black, thick] (12, 1) -- (12.5, 2);
		\draw[-,black, thick] (12.5, 2) -- (13, 3);
		\draw[-, black, dashed] (10, 3) -- (11.5, 0); 
		\draw[-, black, dashed] (11.5, 0) -- (13, 3);

		\filldraw[fill=black, draw=black] (10, 3) circle [radius=.05];
		\filldraw[fill=black, draw=black] (10.5, 2) circle [radius=.05];
		\filldraw[fill=black, draw=black] (11, 1) circle [radius=.05];
		\filldraw[fill=black, draw=black] (11.5, 1) circle [radius=.05];
		\filldraw[fill=black, draw=black] (12, 1) circle [radius=.05];
		\filldraw[fill=black, draw=black] (12.5, 2) circle [radius=.05];
		\filldraw[fill=black, draw=black] (13, 3) circle [radius=.05];

		\draw[<->, black] (16.5, -.5) -- (16.5, 3.5);
		\draw[<->, black] (14.75, 0) -- (18.25, 0);
		
		\draw[-, black] (16.43, .5) -- (16.57, .5) node[scale = .4, above = 7, right = -4]{$1$}; 
		\draw[-, black] (16.43, 1.5) -- (16.57, 1.5) node[scale = .4, above = 7, right = -4]{$3$}; 
		\draw[-, black] (16.43, 2.5) -- (16.57, 2.5) node[scale = .4, right = -1]{$5$}; 
		
		\draw[-, black] (15.25, -.07) -- (15.25, .07) node[scale = .4, left = 2, below = 7]{$-\frac{5}{2}$}; 
		\draw[-, black] (15.75, -.07) -- (15.75, .07) node[scale = .4, left = 2, below = 7]{$-\frac{3}{2}$}; 
		\draw[-, black] (16.25, -.07) -- (16.25, .07) node[scale = .4, left = 2, below = 7]{$-\frac{1}{2}$}; 
		\draw[-, black] (16.75, -.07) -- (16.75, .07) node[scale = .4, below = 7]{$\frac{1}{2}$}; 
		\draw[-, black] (17.25, -.07) -- (17.25, .07) node[scale = .4, below = 7]{$\frac{3}{2}$}; 
		\draw[-, black] (17.75, -.07) -- (17.75, .07) node[scale = .4, below = 7]{$\frac{5}{2}$};

		\draw[-,black, thick] (15.25, 2.5) -- (15.75, 1.5);
		\draw[-,black, thick] (15.75, 1.5) -- (16.25, .5);
		\draw[-,black, thick] (16.25, .5) -- (16.75, 1.5);
		\draw[-,black, thick] (16.75, 1.5) -- (17.25, 1.5);
		\draw[-,black, thick] (17.25, 1.5) -- (17.75, 2.5);
		\draw[-, black, dashed] (16.25, .5) -- (16.5, 0); 
		\draw[-, black, dashed] (16.5, 0) -- (17.25, 1.5);

		\filldraw[fill=black, draw=black] (16.25, .5) circle [radius=.05];
		\filldraw[fill=black, draw=black] (16.75, 1.5) circle [radius=.05];
		\filldraw[fill=black, draw=black] (15.75, 1.5) circle [radius=.05];
		\filldraw[fill=black, draw=black] (17.25, 1.5) circle [radius=.05];
		\filldraw[fill=black, draw=black] (17.75, 2.5) circle [radius=.05];
		\filldraw[fill=black, draw=black] (15.25, 2.5) circle [radius=.05];

		\end{tikzpicture}
		
	\end{center}
	
	\caption{\label{corner} A sample of the midpoint corner growth model at times $0$, $1$, $2$, and $3$ is shown above, in that order from left to right. The broken lines $\zeta_t$ are solid, and the curve $y = 2 |x|$ is dashed. } 
	
\end{figure}

The question of interest here is how the curve $\zeta_t$ behaves as $t$ tends to $\infty$. The answer will depend on $p$. When $p < \frac{1}{2}$, the curve $\zeta (x)$ ``retracts'' to the origin more often than it expands to infinity and thus one expects $\zeta_t (x)$ to closely approximate $\zeta_0 (x) = 2 |x|$ as $t$ tends to $\infty$. 

When $p = \frac{1}{2}$, it can be shown that $\zeta_T (0)$ is of order $T^{1 / 2}$ as $T$ goes to $\infty$. It can be moreover shown (as will follow from Proposition \ref{momentm} below; see Remark \ref{dynamicalnondynamicalcornergrowth}) that, for any fixed $r > 0$ and $s \in \mathbb{R}$, $T^{-1 / 2} \zeta_{\lfloor r T \rfloor } (s T^{1 / 2})$ has a law of large numbers governed by a normalized heat equation $\partial_r \mathcal{H} = 8 \partial_s^2 \mathcal{H}$. Such \textit{hydrodynamical limits} have been well-studied in the context of continuous-time symmetric exclusion processes; see, for instance, Chapters 4 and 5 of the book \cite{SLIS}, whose methods might be applicable to this setting as well. 

When $p > \frac{1}{2}$, it can be shown that $\zeta_T (0)$ is of order $T$ as $T$ goes to $\infty$. Moreover, $T^{-1} \zeta_{\lfloor r T \rfloor } (s T)$ has a law of large numbers governed by a nonlinear wave equation. Again, this phenomenon has been well-studied in the context of continuous-time asymmetric exclusion processes (or, more generally, \textit{attractive} interacting particle systems) by Rezakhalou \cite{HLAPS}; see also Chapter 8 of the book \cite{SLIS}. In Appendix \ref{AsymmetricCorner} (see Proposition \ref{asymmetriccornerfluctuations}) we will outline a derivation of this hydrodynamical limit and furthermore examine the fluctuations of $\zeta_T (sT)$ around its limit shape. We will see that, in analogy with those of the asymmetric simple exclusion process (ASEP), the fluctuations of $\zeta_T (sT)$ are of order $T^{1 / 3}$ and scale to the Gaussian Unitary Ensemble (GUE) Tracy-Widom distribution, a characteristic feature of models in the KPZ universality class; see the survey \cite{EUC} for more information. 

The midpoint corner growth model described above is not dynamical, since the probabilities of shifting a midpoint $z = (z_1 + z_2) / 2$ (of two adjacent vertices $z_1, z_2 \in V_t$) up or down are independent of $\zeta_t$. Let us propose an altered version of this model, in which these probabilities are some functions $p ( t, \zeta_t, z )$ and $1 - p (t, \zeta_t, z)$, respectively, dependent on the time $t$, the state $\zeta_t$, and the site $z$. 

As one of the simplest possible explicit cases, we can fix a parameter $\gamma > 1$  and set $p (t, \zeta_t, z) = \frac{1}{2} \big( 1 - \big( \gamma + \zeta_t (z) \big)^{-1} \big)$; we call the resulting model the \textit{dynamical midpoint corner growth model}. This becomes a perturbation of the $p = \frac{1}{2}$ non-dynamical case of the corner growth mentioned above (one can also consider integrable dynamical deformations of the $p > \frac{1}{2}$ midpoint corner growth model; see Example \ref{dynamicalasymmetricmidpointcorner} below). However, we will see as a special case of Theorem \ref{momentmdynamic} that $\zeta_T (0)$ is of order $T^{1 / 4}$, far smaller than that of the $p = \frac{1}{2}$ non-dynamical model. Also unusual is that the rescaled height function $T^{-1 / 4} \zeta_{\lfloor r T \rfloor} (s T^{1 / 4})$ does not exhibit a hydrodynamical limit shape; it converges to a random variable that we evaluate explicitly (again as a special case of Theorem \ref{momentmdynamic}; see Remark \ref{dynamicalnondynamicalcornergrowth}). 

Interestingly, these phenomena are similar to those found by Borodin as part 4 of Theorem 11.2 of \cite{ERSF}, in which he showed that the height function for the dynamical analog of the symmetric simple exclusion process (SSEP) also exhibits a scaling exponent of $1 / 4$ and has a random profile. We will discuss dynamical asymptotics in more detail in Section \ref{ObservablesAsymptoticsDynamicalExclusion}.

\subsection{The General Model}

\label{ModelGeneral} 

Having provided an example of a dynamical stochastic process, let us introduce the general model (called the dynamical stochastic higher spin vertex model) to be studied in this paper. Although these models might at first appear to be quite different from the dynamical midpoint corner growth model introduced above, the fact that local Markov transition probabilities depend on the state of the model will remain; in Section \ref{DynamicalPartialExclusion} we will see that the former can in fact be degenerated to the latter. 

Our model will depend on several parameters. In what follows, we fix complex numbers $\delta, q \in \mathbb{C}$, a countably infinite set of positive integers $\mathcal{J} = (J_1, J_2, \ldots ) \subset \mathbb{Z}_{> 0}$, and countably infinite sets of complex variables $U = (u_1, u_2, \ldots ) \subset \mathbb{C}$, $\Xi = (\xi_1, \xi_2, \ldots ) \subset \mathbb{C}$, and $S = (s_1, s_2, \ldots ) \subset \mathbb{C}$. These parameters will be subject to certain restrictions that we explain later. 

Let us begin by describing the state space for our model, which is the set of directed path ensembles with (multiplicative) dynamical parameter. 

\subsubsection{Directed Path Ensembles With a Dynamical Parameter}

\label{PathEnsembles}

A \emph{directed path} is a collection of \emph{vertices}, which are lattice points in the non-negative quadrant $\mathbb{Z}_{\ge 0}^2$, connected by \emph{directed edges} (which we may also refer to as \emph{arrows}). A directed edge can connect a vertex $(i, j)$ to either $(i + 1, j)$ or $(i, j + 1)$, if $(i, j) \in \mathbb{Z}_{> 0}^2$; we also allow directed edges to connect $(k, 0)$ to $(k, 1)$ or $(0, k)$ to $(1, k)$, for any positive integer $k$. Thus, directed edges connect adjacent points, always point either up or to the right, and do not lie on the $x$-axis or $y$-axis. 

A \emph{directed path ensemble} is a collection of paths such that each path contains an edge connecting $(0, k)$ to $(1, k)$ or $(k, 0)$ to $(k, 1)$ for some $k > 0$. Stated alternatively, every path ``emanates'' from either the $x$-axis or the $y$-axis; see Figure \ref{ensemble} for an example.

\begin{figure}
	
	\begin{center} 
		
		\begin{tikzpicture}[
		>=stealth,
		scale = 1
		]

		\draw[->, black, thick] (0, 0) -- (0, 3);
		\draw[->, black, thick] (0, 0) -- (3, 0);
		
		\draw[->,black] (0, .47) -- (.45, .47) node[scale = .65, above = 1, left = 22]{$1$};
		\draw[->, black] (0, .53) -- (.45, .53);
		
		\draw[->,black] (0, .97) -- (.45, .97) node[scale = .65, above = 1, left = 22]{$2$};	
		\draw[->,black] (0, 1.03) -- (.45, 1.03);
		
		\draw[->,black] (0, 1.5) -- (.45, 1.5) node[scale = .65, left = 22]{$3$};
		
		\draw[->,black] (0, 2) -- (.45, 2) node[scale = .65, left = 22]{$4$};
		
		\draw[->,black] (.5, .55) -- (.5, .95);
		
		\draw[->,black] (.55, .5) -- (.95, .5);
		
		\draw[->, black] (.55, .97) -- (.95, .97); 
		\draw[->,black] (.55, 1.03) -- (.95, 1.03);
		
		\draw[->, black] (1.05, 1) -- (1.45, 1); 
		
		\draw[->, black] (1.55, 1) -- (1.95, 1); 
		
		\draw[->, black] (2.05, 1) -- (2.45, 1); 
		
		\draw[->,black] (.55, 1.5) -- (.95, 1.5);
		
		\draw[->,black] (.53, 1.05) -- (.53, 1.45);
		\draw[->,black] (.47, 1.05) -- (.47, 1.45);
		
		\draw[->,black] (.53, 1.55) -- (.53, 1.95);
		\draw[->,black] (.47, 1.55) -- (.47, 1.95);
		
		\draw[->,black] (.55, 2.05) -- (.55, 2.45);
		\draw[->,black] (.5, 2.05) -- (.5, 2.45);
		\draw[->,black] (.45, 2.05) -- (.45, 2.45);

		\draw[->,black] (1.05, 1.5) -- (1.45, 1.5);
		\draw[->,black] (1, 1.55) -- (1, 1.95);
		\draw[->,black] (1, 1.05) -- (1, 1.45);
		\draw[->,black] (1, 2.05) -- (1, 2.45);
		\draw[->,black] (1.05, .5) -- (1.45, .5);	
		
		\draw[->,black] (1.5, .55) -- (1.5, .95);
		\draw[->,black] (1.5, 1.05) -- (1.5, 1.45);
		\draw[->,black] (1.47, 1.55) -- (1.47, 1.95);
		\draw[->,black] (1.53, 1.55) -- (1.53, 1.95);
		\draw[->,black] (1.5, 2.05) -- (1.5, 2.45);
		\draw[->,black] (1.55, 2) -- (1.95, 2);
		
		\draw[->,black] (2.05, 2) -- (2.45, 2);

		\draw[->, black] (2.5, 1.05) -- (2.5, 1.45);
		\draw[->, black] (2.5, 1.55) -- (2.5, 1.95);
		\draw[->, black] (2.47, 2.05) -- (2.47, 2.5);
		\draw[->, black] (2.53, 2.05) -- (2.53, 2.5);

		\filldraw[fill=gray!50!white, draw=black] (.5, .5) circle [radius=.05] node[scale = .65, below = 24]{$1$};
		\filldraw[fill=gray!50!white, draw=black] (.5, 1) circle [radius=.05];
		\filldraw[fill=gray!50!white, draw=black] (.5, 1.5) circle [radius=.05];
		\filldraw[fill=gray!50!white, draw=black] (.5, 2) circle [radius=.05];
		
		\filldraw[fill=gray!50!white, draw=black] (1, .5) circle [radius=.05] node[scale = .65, below = 24]{$2$};
		\filldraw[fill=gray!50!white, draw=black] (1, 1) circle [radius=.05];
		\filldraw[fill=gray!50!white, draw=black] (1, 1.5) circle [radius=.05];
		\filldraw[fill=gray!50!white, draw=black] (1, 2) circle [radius=.05];
		
		\filldraw[fill=gray!50!white, draw=black] (1.5, .5) circle [radius=.05] node[scale = .65, below = 24]{$3$};
		\filldraw[fill=gray!50!white, draw=black] (1.5, 1) circle [radius=.05];
		\filldraw[fill=gray!50!white, draw=black] (1.5, 1.5) circle [radius=.05];
		\filldraw[fill=gray!50!white, draw=black] (1.5, 2) circle [radius=.05];

		\filldraw[fill=gray!50!white, draw=black] (2, .5) circle [radius=.05] node[scale = .65, below = 24]{$4$};
		\filldraw[fill=gray!50!white, draw=black] (2, 1) circle [radius=.05];
		\filldraw[fill=gray!50!white, draw=black] (2, 1.5) circle [radius=.05];
		\filldraw[fill=gray!50!white, draw=black] (2, 2) circle [radius=.05];
		
		\filldraw[fill=gray!50!white, draw=black] (2.5, .5) circle [radius=.05] node[scale = .65, below = 24]{$5$};
		\filldraw[fill=gray!50!white, draw=black] (2.5, 1) circle [radius=.05];
		\filldraw[fill=gray!50!white, draw=black] (2.5, 1.5) circle [radius=.05];
		\filldraw[fill=gray!50!white, draw=black] (2.5, 2) circle [radius=.05];

		\end{tikzpicture}
		
	\end{center}

	\caption{\label{ensemble} Above is a directed path ensemble with $\mathcal{J}$-step initial data, where $\mathcal{J} = (2, 2, 1, 1, \ldots)$. } 
\end{figure}
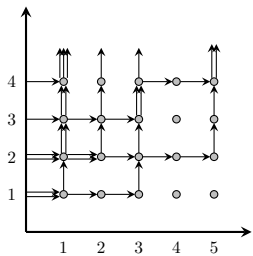

Associated with each $(x, y) \in \mathbb{Z}_{> 0 }^2$ in a path ensemble is an \emph{arrow configuration}, which is a quadruple $(i_1, j_1; i_2, j_2) = (i_1, j_1; i_2, j_2)_{(x, y)} = \big( i_1 (x, y), j_1 (x, y); i_2 (x, y), j_2 (x, y) \big)$ of non-negative integers. Here, $i_1$ denotes the number of directed edges from $(x, y - 1)$ to $(x, y)$; equivalently, $i_1$ denotes the number of vertical incoming arrows at $(x, y)$. Similarly, $j_1$ denotes the number of horizontal incoming arrows; $i_2$ denotes the number of vertical outgoing arrows; and $j_2$ denotes the number of horizontal outgoing arrows. For example, see the right side of Figure \ref{arrows3}. 

Observe that, at any vertex in the positive quadrant, the total number of incoming arrows is equal to the total number of outgoing arrows, that is, $i_1 + j_1 = i_2 + j_2$. This is sometimes referred to as \emph{arrow conservation} (or \emph{spin conservation}). Any assignment of arrow configurations to all vertices of $\mathbb{Z}_{> 0}^2$ satisfying arrow conservation corresponds to a unique directed path ensemble. 	

The models we are interested in will take place on \textit{path ensembles with a (multiplicative) dynamical parameter}. These are path ensembles in which each vertex $(x, y) \in \mathbb{Z}_{> 0} \times \mathbb{Z}_{\ge 0}$ is associated with a complex number called a \textit{(multiplicative) dynamical parameter}, denoted by $\kappa_{x, y}$, that changes from vertex to vertex according to certain rules. Explicitly, we fix $\kappa_{1, 0} = \delta$ and further set
\begin{flalign}
\label{kappamultiplicative}
\kappa_{x, y + 1} = q^{J_{y + 1} - 2j_1 (x, y + 1)} \kappa_{x, y}; \qquad \kappa_{x + 1, y} = q^{2 i_2 (x, y)} s_x^2 \kappa_{x, y}, 
\end{flalign} 

\noindent for any $(x, y) \in \mathbb{Z}_{> 0} \times \mathbb{Z}_{\ge 0}$; see the left side of Figure \ref{arrows3}. Observe the similarity between $s_x^2$ and $q^{J_{y + 1}}$ in \eqref{kappamultiplicative}; $\kappa$ is multiplied by the former as it shifts along the $x$-axis and is multiplied by the latter as it shifts along the $y$-axis. In particular, the rules \eqref{kappamultiplicative} are invariant under the symmetry that reflects the path ensemble across the line $x = y + 1$ and interchanges each $s_x^2$ with $q^{J_x}$. 

The choice $\kappa_{1, 0} = \delta$ and the identity \eqref{kappamultiplicative} together provide a way to associate a complex number with each vertex in the positive quadrant. Under different notation (that requires switching vertices with faces), these directed path ensembles with a dynamical parameter in the special case $(J_1, J_2, \ldots ) = (1, 1, \ldots )$ were the setting of the stochastic IRF models introduced in Section 1 and Section 8 of \cite{ERSF}. 

Assigning values $j_1$ to vertices on the line $(1, y)$ and values $i_1$ to vertices on the line $(x, 1)$ can be viewed as imposing boundary conditions on the vertex model. If $j_1 (1, y) = J_y$ and $i_1 (1, x) = 0$ for all $x, y \in \mathbb{Z}_{> 0}$, then $J_y$ paths enter through each vertex $(0, y)$ on the $y$-axis, and no paths enter through the $x$-axis; see Figure \ref{ensemble}. This particular assignment is called \textit{$\mathcal{J}$-step initial data}; it sets $\kappa_{x, 0} = s^{2 (x - 1)} \delta$ and $\kappa_{1, y} = q^{-\sum_{i = 1}^y J_i} \delta$ for each $x, y > 1$. In general, we will refer to any assignment of $i_1$ to $\mathbb{Z}_{> 0} \times \{ 1 \}$ and $j_1$ to $\{ 1 \} \times \mathbb{Z}_{> 0}$ as \emph{initial data}.

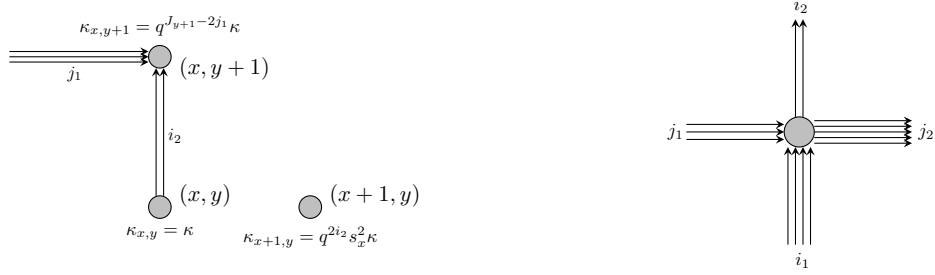
\begin{figure}
	
	\begin{center} 
		
		\begin{tikzpicture}[
		>=stealth,
		scale = 1
		]

		\draw[->,black] (-3.95, 1.7) -- (-3.95, 3) node[scale = .7, left = 1, above]{$i_2$};
		\draw[->,black] (-4.05, 1.7) -- (-4.05, 3);
		
		\draw[->,black] (-3.85, 0) -- (-3.85, 1.3) node[scale = .7, above = 8, left = 65]{$j_1$};
		\draw[->,black] (-3.95, 0) -- (-3.95, 1.3);
		\draw[->,black] (-4.05, 0) -- (-4.05, 1.3);
		\draw[->,black] (-4.15, 0) -- (-4.15, 1.3);
		
		\draw[->,black]  (-5.5, 1.4) -- (-4.2, 1.4) node[scale = .7, right = 10, below = 58]{$i_1$};
		\draw[->,black]  (-5.5, 1.5) -- (-4.2, 1.5);
		\draw[->,black]  (-5.5, 1.6) -- (-4.2, 1.6);
		
		\draw[->,black] (-3.8, 1.35) -- (-2.5, 1.35); 
		\draw[->,black] (-3.8, 1.425) -- (-2.5, 1.425); 
		\draw[->,black] (-3.8, 1.5) -- (-2.5, 1.5) node[scale = .7, right]{$j_2$}; 
		\draw[->,black] (-3.8, 1.575) -- (-2.5, 1.575); 
		\draw[->,black] (-3.8, 1.65) -- (-2.5, 1.65);

		\filldraw[fill=gray!50!white, draw=black] (-4, 1.5) circle [radius=.2];

		\draw[->,black] (-14.5, 2.43) -- (-12.65, 2.43); 
		\draw[->,black] (-14.5, 2.5) -- (-12.65, 2.5) node [scale = .65, below = 10, left = 35]{$j_1$}; 
		\draw[->,black] (-14.5, 2.57) -- (-12.65, 2.57); 
		
		\draw[->,black] (-12.55, .65) -- (-12.55, 2.35) node[scale = .65, right = 11, below = 32]{$i_2$}; 
		\draw[->,black] (-12.45, .65) -- (-12.45, 2.35); 
		
		\filldraw[fill=gray!50!white, draw=black] (-12.5, .5) circle [radius=.15] node[below = 5, scale = .7]{$\kappa_{x, y} = \kappa$} node[scale = .85, above = 5, right = 5]{$(x, y)$};
		
		\filldraw[fill=gray!50!white, draw=black] (-12.5, 2.5) circle [radius=.15] node[above = 5, scale = .7]{$\kappa_{x, y + 1} = q^{J_{y + 1} - 2 j_1} \kappa$} node[scale = .85, below = 5, right = 5]{$(x, y + 1)$};
		
		\filldraw[fill=gray!50!white, draw=black] (-10.5, .5) circle [radius=.15] node[below = 5, scale = .7]{$\kappa_{x + 1, y} = q^{2 i_2} s_x^2 \kappa$} node[scale = .85, above = 5, right = 5]{$(x + 1, y)$};
		\end{tikzpicture}
		
	\end{center}

	\caption{\label{arrows3} Shown to the left are three vertices at positions $(x, y)$, $(x + 1, y)$, and $(x, y + 1)$ for which $i_2 (x, y) = 2$ and $j_1 (x, y + 1) = 3$. From the definition, we have $\kappa_{x, y + 1} = q^{J_{y + 1} - 2j_1 (x, y + 1)} \kappa_{x, y}$ and $\kappa_{x + 1, y} = q^{2 i_2 (x, y)} s_x^2 \kappa_{x, y}$. Shown to the right is a vertex whose arrow configuration is $(i_1, j_1; i_2, j_2) = (4, 3; 2, 5)$.}
\end{figure}

\subsubsection{Definition of the Model}

\label{ProbabilityMeasures}

Having described both the state space and boundary data for our model, let us proceed to explain how it is sampled. The dynamical stochastic higher spin vertex model $\mathcal{P} = \mathcal{P} (U, \Xi, S; q; \delta)$ will be the infinite-volume limit of a a family of probability measures $\mathcal{P}_n = \mathcal{P}_n (U, \Xi, S; q, \delta)$, which are defined on the set of directed path ensembles with multiplicative dynamical parameter whose vertices are all contained in triangles of the form $\mathbb{T}_n = \{ (x, y) \in \mathbb{Z}_{\ge 0}^2: x + y \le n \}$. The first two measures $\mathcal{P}_0$ and $\mathcal{P}_1$ are both supported by the empty ensembles. 

For each positive integer $n \ge 1$, we will define $\mathcal{P}_{n + 1}$ from $\mathcal{P}_n$ through the following Markovian update rules. Use $\mathcal{P}_n$ to sample a directed path ensemble $\mathcal{E}_n$ on $\mathbb{T}_n$. This yields arrow configurations (and hence dynamical parameters) for all vertices in the triangle $\mathbb{T}_{n - 1}$. To extend this to a path ensemble on $\mathbb{T}_{n + 1}$, we must prescribe arrow configurations to all vertices $(x, y)$ on the complement $\mathbb{T}_n \setminus\mathbb{T}_{n - 1}$, which is the diagonal $\mathbb{D}_n = \{ (x, y) \in \mathbb{Z}_{> 0}^2: x + y = n \}$. Since any incoming arrow to $\mathbb{D}_n$ is an outgoing arrow from $\mathbb{D}_{n - 1}$, $\mathcal{E}_n$ and the initial data prescribe the first two coordinates, $i_1$ and $j_1$, of the arrow configuration to each $(x, y) \in \mathbb{D}_n$. Thus, it remains to explain how to assign the second two coordinates ($i_2$ and $j_2$) to any vertex on $\mathbb{D}_n$, given the first two coordinates. 

This is done by producing $(i_2, j_2)_{(x, y)}$ from $(i_1, j_1)_{(x, y)}$ according to the transition probabilities 
\begin{flalign}
\label{pndynamicalstochasticspin} 
& \mathcal{P}_n \big[ (i_2, j_2) \big| (i_1, j_1)  \big] = \psi_{\xi_x u_y; s_x; q; J_y} \big( i_1, j_1; i_2, j_2 \b| \kappa_{x, y - 1} \big), 
\end{flalign}

\noindent where the function $\psi$ is given by Definition \ref{psiweightdefinition} below. Throughout, we assume that all parameters $q$, $\delta$, $\mathcal{J}$, $U$, $\Xi$, and $S$ are chosen such that $\psi_{u_y \xi_x; s_x; q; J_y} \big( i_1, j_1; i_2, j_2 \b| \kappa_{x, y - 1}\big)$ is nonnegative for each path ensemble $\mathcal{E}_n$; such a choice is possible, and we will give some examples later in the paper.

Choosing $(i_2, j_2)$ according to the above transition probabilities yields a random directed path ensemble $\mathcal{E}_{n + 1}$, now defined on $\mathbb{T}_{n + 1}$; the probability distribution of $\mathcal{E}_{n + 1}$ is then denoted by $\mathcal{P}_{n + 1}$. We define $\mathcal{P} = \lim_{n \rightarrow \infty} \mathcal{P}_n$ as the probability measure for the dynamical stochastic higher spin vertex model.

To complete the description of the model, we require the $\psi$ stochastic weights, which are given by the following definition. 

\begin{definition}
	
\label{psiweightdefinition} 

Fix parameters $J \in \mathbb{Z}_{> 0}$ and $u, q, s, \kappa \in \mathbb{C}$. For any $i_1, j_1, i_2, j_2 \in \mathbb{Z}_{\ge 0}$, define $\psi (i_1, j_1; i_2, j_2) = \psi \big( i_1, j_1; i_2, j_2 \b| \kappa) = \psi_{u; s; q; J} \big( i_1, j_1; i_2, j_2 \b| \kappa \big)$ by 
\begin{flalign}
\label{psistochastic} 
\begin{aligned}
& \psi (i_1, j_1; i_2, j_2) \\
& \quad = \textbf{1}_{i_1 + j_1 = i_2 + j_2} q^{(j_2 - i_1) J} \left(\displaystyle\frac{u}{s} \right)^{j_1} \displaystyle\frac{(q^{i_1 - j_2 + 1}; q)_{j_2} (q^{j_2 - J}; q)_{j_1} (s u q^J; q)_{i_1 - j_2} (s^2 q^{i_1 - j_2}; q)_{j_1}}{(su; q)_{i_1 + j_1} (q; q)_{j_2} (q^{j_2 - J}; q)_{j_1 - j_2}} \\
& \qquad \times \displaystyle\frac{\big( u s^{-1} \kappa^{-1} q^{-i_1}; q\big)_{j_2} \big( q^{1 - i_1 - J} u^{-1} s^{-1} \kappa^{-1}; q \big)_{j_1} \big( q \kappa^{-1} ; q \big)_{j_1} \big( q^{j_2 - 2 i_1 + 1} s^{-2} \kappa^{-1} ; q \big)_{i_1 - j_2}}{ \big( q^{1 - j_2} \kappa^{-1}; q \big)_{j_1} \big( q^{j_2 - i_1 - J + 1} s^{-2} \kappa^{-1}; q \big)_{j_1} \big( q^{j_2 - 2 i_1 - J} s^{-2} \kappa^{-1}; q \big)_{j_2} \big( q^{2j_2 - 2i_1 - J + 1} s^{-2} \kappa^{-1}; q \big)_{i_1 - j_2}}  \\
& \qquad \times {_{10}	W_9} \left( q^{-j_2} \kappa^{-1}; q^{-j_1}, q^{-j_2}, q^{j_1 - J} \kappa^{-1}, q^{1 - i_1} s^{-2} \kappa^{-1}, s u^{-1} q^{i_1 - j_2 + 1}, u s q^{i_1 - j_2 + J}, q^{-i_1} \kappa^{-1} \right). 
\end{aligned} 
\end{flalign}

\noindent In \eqref{psistochastic}, $(a; q)_k$ denotes the \textit{$q$-Pochhammer symbol}, defined by
\begin{flalign}
\label{productq}
(a; q)_k = \displaystyle\prod_{j = 0}^{k - 1} (1 - q^j a), \quad \text{and} \quad (a; q)_k = \displaystyle\prod_{j = 1}^{-k} \displaystyle\frac{1}{1 - q^j a}, 
\end{flalign}

\noindent if $k \in \mathbb{Z}_{\ge 0}$ and $k \in \mathbb{Z}_{< 0}$, respectively. Furthermore, for any $r \in \mathbb{Z}_{> 0}$ and $a_1, a_4, a_5, \ldots , a_{r + 1}; z \in \mathbb{C}$,
\begin{flalign}
\label{poisedhypergeometric}
_{r + 1} W_r (a_1; a_4, a_5, \ldots , a_{r + 1}; q, z) = \displaystyle\sum_{k = 0}^{\infty} \displaystyle\frac{z^k (a_1; q)_k}{ \big( q; q \big)_k} \displaystyle\frac{1 - a_1 q^{2k}}{1 - a_1} \displaystyle\prod_{j = 4}^{r + 1} \displaystyle\frac{(a_j; q)_k}{\big( q a_1 / a_j; q \big)_k},
\end{flalign}

\noindent denotes the \textit{very-well poised basic hypergeometric series}. 
\end{definition}

For the dynamical stochastic higher spin vertex model to be well-defined we require that all of the $\psi (i_1, j_1; i_2, j_2)$ be nonnegative and that $\sum_{i_2, j_2} \psi (i_1, j_1; i_2, j_2) = 1$; the former was imposed as a condition on the parameters of the model. However, the latter is not so immediately apparent to us, and we provide an indirect proof (through a Pieri-type identity for a family of symmetric functions studied in \cite{ERSF}) as Proposition \ref{psisumstochastic} in Section \ref{PsiStochasticity}.

\subsection{Two Degenerations} 

\label{DegenerationsParticle}

Although the $\psi$ weights \eqref{psistochastic} might not appear so pleasant, they are governed by five parameters ($q$, $u$, $s$, $J$, and $\kappa$) and are thus quite general. In particular, by specializing these parameters in various ways, one can obtain a number of stochastic processes that have been studied in the past, as well as new integrable models that seemingly have not been considered previously. 

In the simplest case, we can consider the dynamical stochastic higher spin vertex model in which $\delta$ tends to $0$; in this setting, all of the $\kappa_{x, y}$ also tend to $0$, so the model is no longer dynamical. Letting $\kappa$ tend to $0$ in \eqref{psistochastic}, one can verify that the $\psi$ weights tend to the $R_{\alpha}^{(J)}$ weights given by equation (3.16) in \cite{SHSVML} (see also equation (5.6) of \cite{HSVMRSF}), which are the transition functions of the \textit{stochastic higher spin vertex model}. Thus, setting $\delta = 0$ in the dynamical stochastic higher spin vertex model yields Corwin-Petrov's \cite{SHSVML} stochastic higher spin vertex model. 

That model is already remarkably general. In particular, under suitable choices of parameters, it can be specialized \cite{HSVMRSF,SHSVML} to the \textit{stochastic six-vertex model} \cite{CFSAEPSSVMCL,PTAEPSSVM,SSVM,SVMCP,SVMRSASH} and the \textit{$q$-Hahn boson model} \cite{TBP,TBPT,IZRCMFSS,LAEP}. Both are in turn known to specialize to many ostensibly different models in the KPZ universality class \cite{EUC}. For instance, the former degenerates to the ASEP \cite{CFSAEPSSVMCL,PTAEPSSVM,SSVM}; the latter to the $q$-totally asymmetric simple exclusion process (TASEP) \cite{DIGM,DTT,MP,SHSVMAEP,RTAEM}, some multi-particle asymmetric diffusion models (MADMs) \cite{ESOFP,AEP,DWE}, and various directed polymer models \cite{RWBRE,MP,SWLP}; and both to the KPZ equation \cite{MP}. Thus, our dynamical stochastic higher spin vertex model degenerates to dynamical generalizations of some of those models.

As another example, one could set $\mathcal{J} = (1, 1, \ldots )$; then the $\psi$ weights \eqref{psistochastic} degenerate to the transition functions for Borodin's \textit{stochastic IRF models} given by equation (1.1) of \cite{ERSF}. Those models are also quite general and include as special cases dynamical versions of the ASEP and stochastic six-vertex model; see Section 1 and Section 9 of \cite{ERSF}. 

The purpose of this section is to describe two explicit degenerations of our dynamical stochastic higher spin vertex model, neither of which appears to be a special case of the models mentioned above. The first (given in Section \ref{LambdaSpecializationModelGeneral}) is a dynamical variant of the $q$-Hahn boson model, and the latter (given in Section \ref{DynamicalPartialExclusion}) is a dynamical discrete-time partial exclusion process.

\subsubsection{The Dynamical \texorpdfstring{$q$}{}-Hahn Boson Specialization} 

\label{LambdaSpecializationModelGeneral}

In order to simplify the $\psi$ stochastic weights, we seek a specialization of the parameters $u$, $s$, $q$, $J$, and $\kappa$ for which the $_{10} W_9$ hypergeometric series on the right side of \eqref{psistochastic} factors completely. In simpler non-dynamical setting, such a factorization was found in Proposition 6.7 of \cite{FSRF}, which considers the case $u = s$. 

It will happen that a similar factorization occurs here. Specifically, if we set $u = s$, then 
\begin{flalign}
\label{weightmodel1}
\begin{aligned}
\psi & \big( i_1, j_1; i_2, j_2 \b| \kappa \big) \\
& = (s^2 q^J)^{j_2} \displaystyle\frac{\textbf{1}_{j_2 \le i_1} (q; q)_{i_1}}{(q; q)_{i_1 - j_2} (q; q)_{j_2}}  \displaystyle\frac{ \ (q^{-J}; q)_{j_2} (s^2 q^J; q)_{i_1 - j_2} }{ (s^2; q)_{i_1} } \displaystyle\frac{(s^2 q^{i_1} \kappa; q)_{i_1 - j_2} (q^{i_1 - j_2 + 1} \kappa; q)_{j_2}}{(s^2 q^{i_1 + J - j_2} \kappa; q)_{i_1 - j_2} (s^2 q^{2i + J - 2j_2 + 1} \kappa; q)_{j_2}}.
\end{aligned}
\end{flalign} 

One can establish \eqref{weightmodel1} as an application of Jackson's identity for very well-poised, balanced ${{_8} W_7}$ terminating basic hypergeometric series (see equation (2.6.2) of \cite{BHS}). However, this is also a quick consequence of Lemma \ref{dynamicalparticlesvlambda} and Lemma \ref{lpsiweights} from Section \ref{DynamicalParticles} below.  

Changing variables $b = s^2$, $c = q^J$, and $a = bc$ on the right side of \eqref{weightmodel1} leads to the following stochastic weights; observe in the definition below that $a / b \in \mathbb{C}$ can be arbitrary and need not be an integer power of $q$. 

\begin{definition}
	
	\label{stochasticparticlesjumps} 
	
	Fix parameters $q, a, b, \kappa \in \mathbb{C}$. For all nonnegative integers $i, j$ (including $i = \infty$ if $|q| < 1$), define 
	\begin{flalign*}
	\varphi \big( j \b| i \big) = \varphi_{q, a, b, \kappa} \big( j \b| i \big) = a^j \displaystyle\frac{\textbf{1}_{j \le i} (q; q)_{i}}{(q; q)_{j} (q; q)_{i - j}} \displaystyle\frac{(b / a; q)_{j} (a; q)_{i - j}}{ (b; q)_{i}} \displaystyle\frac{(q^{i} b \kappa; q)_{i - j} (q^{i - j + 1} \kappa; q)_{j}}{(q^{i - j} a \kappa; q)_{i - j} (q^{2i - 2j + 1} a \kappa; q)_{j}}. 
	\end{flalign*}
\end{definition}	

Let us briefly outline two ways to verify that these weights indeed satisfy the stochasticity property $\sum_{j = 0}^{\infty} \varphi \big(j \b| i \big) = 1$. The first is by specializing $u = s$ in Proposition \ref{psisumstochastic}, using \eqref{weightmodel1}, relabeling $(a, b) = (s^2 q^J, s^2)$, and then analytically continuing in $a$ to expand the range of $a / b$ from $q^{\mathbb{Z}_{> 0}}$ to all of $\mathbb{C}$. The second is as a quick application of the Rogers identity, stated as \eqref{hypergeometric65sumterminating} below, for very well-poised, balanced ${{_6} W_5}$ terminating basic hypergeometric series (in particular, apply \eqref{hypergeometric65sumterminating} with $a$ set to $1 / q^{2i} a \kappa$, $b$ set to $b / a$, $c$ set to $1 / q^i \kappa$, and $n$ set to $i$). 

Granting this stochasticity, we will proceed to define a probability measure $\mathcal{M}$, called the \textit{dynamical $q$-Hahn boson model}, on the set of directed path ensembles with (multiplicative) dynamical parameter, that is governed by the $\varphi$ weights. This model will depend on several parameters. To that end, fix $q, \delta \in \mathbb{C}$ and countably infinite sets of complex numbers $B = (b_1, b_2, \ldots ) \subset \mathbb{C}$ and $C = (c_1, c_2, \ldots ) \subset \mathbb{C}$; further denote $a_{x, y} = b_x c_y$ for each $x, y \in \mathbb{Z}_{> 0}$. Given a directed path ensemble $\mathcal{E}$, we can associate a (multiplicative) dynamical parameter $\kappa_{x, y}$ to each $(x, y) \in \mathbb{Z}_{\ge 0} \times \mathbb{Z}_{> 0}$ by setting $\kappa_{1, 0} = \delta$ and replacing \eqref{kappamultiplicative} with $\kappa_{x, y + 1} = q^{-2 j_1 (x, y + 1)} c_{y + 1} \kappa_{x, y}$ and $\kappa_{x + 1, y} = q^{2 i_2 (x, y)} b_x \kappa_{x, y}$; observe that this is the effect of substituting $b_x = s_x^2$ and $c_y = q^{J_y}$ in \eqref{kappamultiplicative}. 

Now, define $\mathcal{M} = \mathcal{M} (B, C; q; \delta)$ in the same way as $\mathcal{P}$ was defined in Section \ref{ProbabilityMeasures}, as a limit of measures $\mathcal{M}_n$ (on triangles $\mathbb{T}_n$) that are entirely analogous to the measures $\mathcal{P}_n$, except that \eqref{pndynamicalstochasticspin} is replaced by 
\begin{flalign}
\label{pndynamicalstochasticspin2} 
& \mathcal{M}_n \big[ (i_2, j_2) \big| (i_1, j_1)  \big] = \varphi_{q, a_{x, y}, b, \kappa_{x, y - 1}} \big( j_2 \b| i_1 \big) \textbf{1}_{i_1 + j_1 = i_2 + j_2} . 
\end{flalign}

We assume that the parameters $q$, $\delta$, $B$, and $C$ are chosen such that the weights on the right side of \eqref{pndynamicalstochasticspin2} are always nonnegative. This can be ensured, for instance, if $0 < b_{x'} < a_{x, y} < 1$ for each $x, x', y \in \mathbb{Z}_{> 0}$, $q \in (0, 1)$, and $\delta \le 0$. 

The $\delta = 0$ (non-dynamical) case of the $\varphi \big( i \b| j \big)$ from Definition \ref{stochasticparticlesjumps} coincides with the stochastic weights of the \textit{$q$-Hahn boson model} introduced by Povolotsky in equation (8) of \cite{IZRCMFSS} and later studied in a number of other works \cite{RWBRE,HSVMRSF,TBP,TBPT,LAEP}, partly due to its connections with directed random polymers \cite{RWBRE} and KPZ universality \cite{LAEP}. This is the reason we refer to the above model as the ``dynamical'' $q$-Hahn boson model. 

To explicitly relate it with the model from Section \ref{ModelGeneral}, consider the dynamical stochastic higher spin vertex model with parameters $q$, $\delta$, $U = (u_1, u_2, \ldots )$, $\Xi = (\xi_1, \xi_2, \ldots)$, and $S = (s_1, s_2, \ldots )$. Set $U$, $S$, and $\Xi$ such that $u_1 = u_2 = \cdots = u$ for some $u \in \mathbb{C}$ and $u \xi_j = s_j$ for all $j \in \mathbb{Z}_{> 0}$. Further reparameterize $b_x = s_x^2$, $c_y = q^{J_y}$, and $a_{x, y} = b_x c_y$ for each $x, y \in \mathbb{Z}_{> 0}$. 

Then, due to \eqref{weightmodel1}, \eqref{pndynamicalstochasticspin2} follows from \eqref{pndynamicalstochasticspin}. Therefore, as least when each $c_y \in q^{\mathbb{Z}_{> 0}}$, the dynamical $q$-Hahn boson model is a special case of the general dynamical stochastic higher spin vertex model. In fact, many of our results can be analytically continued in the $\{ c_y \}$, thereby making them valid for the general dynamical $q$-Hahn boson model as well; however, such analytic continuations will often be omitted, as they do not serve as the primary focus of this work. 

There are a number of stochastic processes that can be formed as further degenerations of the dynamical $q$-Hahn boson model. The next section provides one such example, which is a dynamical discrete-time partial exclusion process. A list of some other degenerations is provided in Section \ref{JumpRandomWalk}.

\subsubsection{The Dynamical \texorpdfstring{$(J; \gamma)$}{}-PEP}

\label{DynamicalPartialExclusion}

In this section we describe a further degeneration of the dynamical $q$-Hahn boson model to a type of dynamical partial exclusion process, which we refer to as the dynamical $(J; \gamma)$-partial exclusion process (PEP); the $J = 1$ case of this process in fact coincides with the midpoint corner growth model from Section \ref{DynamicalCornerGrowth} (see Remark \ref{cornerpartialexclusion}). Here, we only define the dynamical $(J; \gamma)$-PEP; the specific way in which to derive it as a degeneration of the dynamical $q$-Hahn boson model will be detailed in Example \ref{dynamicjgammaexclusion} below. 

However, let us remark that this degeneration involves taking limit as the parameter $q$ from the boson model tends to $1$. It is also possible to consider an analog of the dynamical $(J; \gamma)$-PEP, in which $q < 1$ remains bounded away from $1$. In the $J = 1$ case, that process is detailed in Example \ref{dynamicalasymmetricmidpointcorner} below, where it is explained that the resulting system is a dynamical variant of the asymmetric $p > \frac{1}{2}$ midpoint corner growth model from Section \ref{DynamicalCornerGrowth}. 

Before describing the dynamical $(J; \gamma)$-PEP, let us first provide an informal explanation of its non-dynamical version, which we refer to as the \textit{$(J; \infty)$-PEP}. It can also be viewed as a limit degeneration of Povolotsky's original $q$-Hahn boson model \cite{IZRCMFSS}, although we have seen no mention to it until now.

 Suppose infinitely many boxes $\ldots,  B_{-1}, B_0, B_1, \ldots $ are placed next to each other in a line, in that order from left to right; each box $B_i$ has $J + 1$ \textit{slots}, and each slot can contain at most one ball (so each box contains at most $J + 1$ balls). Initially, the boxes $B_0, B_{-1}, \ldots $ are \textit{full} (meaning that all of their $J + 1$ slots are occupied by balls), and the remaining boxes $B_1, B_2, \ldots $ are \textit{empty} (meaning that none of their slots are occupied). At each time step, we pick one slot $s_i \in [1, J + 1]$ uniformly at random and independently for each $i \in \mathbb{Z}$. Then, in parallel, each box $B_i$ retains its ball (if it has one) in slot $s_i$ and transfers all of its other balls to box $B_{i + 1}$, filling the slots in $B_{i + 1}$ not equal to $s_{i + 1}$ at random. 
 
 One would then be interested in the long-time behavior of the \textit{current} for this model, that is, the number of balls that are past some site $x$ after a large time $T$. However, although this model appears to be quite simple, we are unaware of any attempt to analyze it in either the probability or mathematical physics literature; we will outline some asymptotic results in Section \ref{ObservablesAsymptoticsDynamicalExclusion} below. 
 
 Now let us define the general, dynamical version of this model, which is dependent on an additional parameter $\gamma$. 

\begin{definition}
	
	\label{dynamicjgammaexclusion1} 
	
	Fix $J \in \mathbb{Z}_{> 0}$ and a positive real number $\gamma > J + 1$. The \textit{dynamical $(J; \gamma)$-PEP} is an interacting particle system $\eta (t) = \big\{ \eta_1 (t), \eta_2 (t), \ldots \big\}_{t \in \mathbb{Z}_{\ge 0}}$, in which $\eta_k (t) \in \{ 0, 1, \ldots , J + 1 \}$ for all $k, t \in \mathbb{Z}_{> 0}$. At time $0$, we set $\eta_1 (0) = \eta_2 (0) = \cdots = 0$. Then, assume this model has been run up to time $s \in \mathbb{Z}_{\ge 0}$; $\eta (s + 1)$ is defined from $\eta (s)$ as follows. Set $X_0 (s) = J$ and let $X_1 (s), X_2 (s), \ldots \in \mathbb{Z}_{\ge 0}$ be mutually independent random variables such that $X_k (s) \in \big\{ \eta_k (s) - 1, \eta_k (s) \big\}$ almost surely and, for each $k \ge 1$,  
	\begin{flalign}
	\label{dynamicxk}
	\begin{aligned}
	\mathbb{P} \big[ X_k (s) = \eta_k (s) - 1 \big] & = \displaystyle\frac{\eta_k (s)}{J + 1} \left( 1 + \displaystyle\frac{ J + 1 - \eta_k (s) }{\Upsilon} \right); \\
	 \mathbb{P} \big[ X_k (s) = \eta_k (s) \big] & = \displaystyle\frac{J + 1 - \eta_k (s)}{J + 1} \left( 1 - \displaystyle\frac{\eta_k (s) }{\Upsilon} \right), 
	 \end{aligned}
	\end{flalign}
	
	\noindent for all $k > 1$. Then set $\eta_k (s + 1) = X_{k - 1} (s) - X_k (s) + \eta_k (s)$ for each $k \ge 1$. 
	
	To explain the definition of $\Upsilon$ above, let $\mathfrak{h}_t (x) = \mathfrak{h}_t (x; \eta) = \sum_{j = x}^{\infty} \eta_j (t)$ and $\mathfrak{a}_t (x) = \mathfrak{a}_t (x; \eta) = \sum_{j = 1}^{x - 1} \big( J + 1 - \eta_j (t) \big)$, for any positive integers $x$ and $t$. In particular, the \textit{current} $\mathfrak{h}_t (x)$ denotes the number of particles at or to the right of $x$ at time $t$, and the \textit{anti-current} $\mathfrak{a}_t (x)$ denotes the number of ``available'' slots that can accommodate particles to the left of site $x$. 
	
	Then, $\Upsilon = \Upsilon (x; s) = \Upsilon (\eta; x; s) = \gamma + \mathfrak{h}_s (x) + \mathfrak{a}_s (x)$; this is the value of $\Upsilon$ that is used in the probabilities \eqref{dynamicxk}, which are nonnegative since $\Upsilon \ge \gamma \ge J + 1 \ge \eta_k (s)$.

\end{definition}

	This model can be alternatively viewed as a partial exclusion process on $\mathbb{Z}_{> 0}$ in which each site $k \in \mathbb{Z}_{> 0}$ can hold at most $J + 1$ particles. Initially (at time $0$), no site is occupied by any particles. At each (discrete) time step $t > 0$, $J$ particles enter the leftmost site (site 1), and then particles move to the right as follows. If a site $k \in \mathbb{Z}_{> 0}$ had $i = i_k (t - 1)$ particles at time $t - 1$, then it transfers $i$ (all) particles to site $k + 1$ with probability $(J - i + 1) (\Upsilon - i) / (J + 1) \Upsilon$, and it transfers $i - 1$ (all but one) particles to site $k + 1$ with the complementary probability $i (\Upsilon + J + 1 - i) / (J + 1) \Upsilon$; this occurs independently and at parallel for all sites $k \in \mathbb{Z}_{> 0}$. Through this interpretation, we see that sending $\gamma$ to $\infty$ indeed degenerates this model to the $(J; \infty)$-PEP informally described above. 
	
	The parameter $\Upsilon$ is an analog of the height function for the TASEP when viewed as a corner growth model. Since $J$ particles enter the system at each step, $\Upsilon$ can be expressed entirely in terms of the current $\mathfrak{h}_t (x)$ through
	\begin{flalign}
	\label{dynamicalparameterpartialexclusion}
	\Upsilon = \gamma + 2 \mathfrak{h}_t (x) + (J + 1) (x - 1) - J t. 
	\end{flalign}
	
	\begin{rem}
		
	\label{cornerpartialexclusion}
	
	Let us consider the dynamical $(J; \gamma)$-PEP with $J = 1$, so that $\eta_k (s) \in \{ 0, 1, 2 \}$ for all $k$ and $s$. Then, the probabilities \eqref{dynamicxk} imply that $X_k (s) = 0$ if $\eta_k (s) = 0$ and $X_k (s) = 1$ if $\eta_k (s) = 2$, almost surely. If $\eta_k (s) = 1$, then $\mathbb{P} \big[ X_k (s) = 0 \big] = \frac{1}{2} (1 + \Upsilon^{-1}) $ and $\mathbb{P}\big[ X_k (s) = 1 \big] = \frac{1}{2} (1 - \Upsilon^{-1})$. Combining this with \eqref{dynamicalparameterpartialexclusion} matches the distribution of $\Upsilon (x, t)$ of the dynamical $(1; \gamma)$-PEP with that of the shifted height function $\zeta_t \big( x - \frac{t}{2} - 1 \big) + \gamma$ of the dynamical midpoint corner growth model from Section \ref{DynamicalCornerGrowth}, for all $x \in \mathbb{Z}_{> 0}$ and $t \in \mathbb{Z}_{\ge 0}$. 
	
	It is also possible to associate the dynamical $(J; \gamma)$-PEP with a dynamical corner growth model when $J > 1$, in which case the slopes of the broken lines can take a larger range of values. However, this is a bit more complicated so we will not pursue it here. 
	\end{rem}

	\subsection{Observables and Asymptotics}
	
	\label{ObservablesAsymptoticsDynamicalExclusion}
	
	As mentioned previously, one of our reasons for introducing the dynamical stochastic higher spin vertex models is that they are integrable, in that it is possible to provide explicit integral identities for large families of observables for these models. 
	
	For the dynamical $q$-Hahn boson model, this result is given by Theorem \ref{fusedhigherspin}. In what follows, $\mathfrak{h}_N (x)$ denotes the number of paths in the model at or to the right of $(x, N)$. 
		
		\begin{thm}
			
			\label{fusedhigherspin}
			
			Fix $0 < q < 1$, $\delta \in \mathbb{C}$, $B = (b_1, b_2, \ldots ) \subset \mathbb{C}$, and $C = (c_1, c_2, \ldots ) \subset \mathbb{C}$. Assume that there exists $\mathcal{J} = (J_1, J_2, \ldots ) \subset \mathbb{Z}_{> 0}$ such that $c_y = q^{J_y}$ for each $y \in \mathbb{Z}_{\ge 0}$. Furthermore, assume that all parameters are chosen such that $\varphi_{q, b_x, a_{x, y}, \kappa_{x, y - 1}} \big( j \b| i \big) \ge 0$ for all nonnegative integers $x$ and $y$ any $\mathcal{E}$. Moreover, assume that $b_j \notin \big[ 1, \max_i |q|^{1 - J_i} \big]$ for each $j$. 
			
			Consider the dynamical $q$-Hahn boson model $\mathcal{M} (B, C; q; \delta)$, as defined in Section \ref{LambdaSpecializationModelGeneral}, run with $\mathcal{J}$-step initial data for some time $N \in \mathbb{Z}_{> 0}$. Then, for any $k \in \mathbb{Z}_{> 0}$ and $x_1, x_2, \ldots , x_k \in \mathbb{Z}_{> 0}$, we have that 
			\begin{flalign}
			\label{fusedhigherspinintegral}
			\begin{aligned}
			& \displaystyle\frac{1}{\big( \delta^{-1}; q \big)_k} \mathbb{E} \left[ \displaystyle\prod_{j = 0}^{k - 1} \bigg( \delta^{-1} q^j - q^{-\mathfrak{h}_N (x_{j + 1})} \displaystyle\prod_{i = 1}^N c_i \displaystyle\prod_{i = 1}^{x_j - 1} b_i \bigg) \big(q^j - q^{\mathfrak{h}_N (x_{j + 1})} \big) \right] \\
			& \qquad \qquad \quad = \displaystyle\frac{q^{\binom{k}{2}}}{(2 \pi \textbf{\emph{i}})^k} \displaystyle\oint \cdots \displaystyle\oint \displaystyle\prod_{1 \le i < j \le k} \displaystyle\frac{z_i - z_j}{z_i - q z_j} \displaystyle\prod_{j = 1}^k \left( \displaystyle\prod_{i = 1}^{x_j - 1} \displaystyle\frac{1 - z_j}{1 - b_i^{-1} z_j} \displaystyle\prod_{i = 1}^N \displaystyle\frac{1 - c_i z_j}{1 - z_j} \right) \displaystyle\frac{d z_j}{z_j},
			\end{aligned}
			\end{flalign}
			
			\noindent where the $z_j$ contours are positively oriented and contain $\bigcup_{i = 1}^N \{ 1, q^{-1}, \ldots , q^{1 - J_i} \}$ but no other poles of the integrand. 
		\end{thm}
		
		\begin{rem}
		
		\label{deltaintegral}
		
		Observe that the right side of \eqref{fusedhigherspinintegral} is independent of $\delta$. In particular, by sending $\delta$ to $0$, we find that the observable on the left side of \eqref{fusedhigherspinintegral} is equal to $\mathbb{E} \big[ \prod_{j = 0}^{k - 1} \big(q^j - q^{\mathfrak{h}_N (x_{j + 1})} \big) \big]$, where the expectation is now taken with respect to the non-dynamical $q$-Hahn boson model under $\mathcal{J}$-step initial data. This comparison of the dynamical to non-dynamical model will be particularly useful for asymptotic analysis of the dynamical $(J, \infty)$-PEP in Section \ref{DynamicPartialExclusionAsymptotic}. 
		
		\end{rem}

	Theorem \ref{fusedhigherspin} is not stated above in its most general form. For example, it admits an extension to the general dynamical stochastic higher spin vertex model $\mathcal{P}$ that is not much more intricate than \eqref{fusedhigherspinintegral}; we mention it in Remark \ref{generalfusedcontour} below. 
	
	Furthermore, it is possible to analytically continue the $\{ c_y \}$ parameters so that they need not reside in $q^{\mathbb{Z}_{> 0}}$ (at the cost of running the model under different, non-deterministic initial data). However, as stated above, such analytic continuations do not serve the primary focus of this work and thus we will not pursue them here.

	Recall that the dynamical $q$-Hahn boson model admits a degeneration to the dynamical $(J, \gamma)$-PEP. Applying this degeneration to Theorem \ref{fusedhigherspin} yields the following corollary; recall in the below that the current $\mathfrak{h}_t (x)$ denotes the number of particles at or to the right of $x$ at time $t$. 
	
	\begin{cor}
		
		\label{kexclusionidentity} 
		
		Let $J$ and $k$ be positive integers and $\gamma > J + 1$ be a real number. For any positive integers $x_1, x_2, \ldots , x_k$, we have that 
		\begin{flalign*}
		& \displaystyle\frac{1}{\prod_{j = 0}^{k - 1}	 (\gamma + j) } \mathbb{E} \left[ \displaystyle\prod_{j = 0}^{k - 1} \big( NJ - (J + 1) x_j - \mathfrak{h}_N (x_j + 1) - \gamma - j \big) \big( \mathfrak{h}_N (x_j + 1) - j \big) \right] \\
		& \qquad = \displaystyle\frac{1}{(2 \pi \textbf{\emph{i}})^k} \displaystyle\oint \cdots \displaystyle\oint \displaystyle\prod_{1 \le i < j \le k} \displaystyle\frac{y_i - y_j}{y_i - y_j - 1} \displaystyle\prod_{j = 1}^k \left( \displaystyle\frac{y_j - J - 1 }{y_j} \right)^{x_j + 1} \left( \displaystyle\frac{ y_j - 1}{y_j  -J - 1}\right)^N d y_j, 
		\end{flalign*}
		
		\noindent where the contours for $y_j$ are positively oriented and enclose $2, 3, \ldots , J + 1$ but not $0$, and where the expectation on the left side is taken with respect to the dynamical $(J, \gamma)$-PEP. 
		
	\end{cor}
	
	In Section \ref{NonDynamickAsymptotic} and Section \ref{DynamicPartialExclusionAsymptotic} below, we show how Corollary \ref{kexclusionidentity} can be used to access asymptotics for the current $\mathfrak{h}_T (x)$ of dynamical $(J; \gamma)$-PEP, both when $\gamma$ tends to $\infty$ and remains finite. 
	
	In either case, we will see that if $T$ is large, then with high probability sites to the left of $(J / (J + 1) - \varepsilon) T$ will be fully occupied (with $J + 1$ particles) and sites to the right of $(J / (J + 1) + \varepsilon) T$ will be empty (with $0$ particles), for any fixed $\varepsilon > 0$. The transition from full sites to empty sites will occur in a region centered at $JT / (J + 1)$. In the non-dynamical case $\gamma = \infty$, this region will be of size $T^{1 / 2}$; in the dynamical case $\gamma < \infty$, it will be of size $T^{1 / 4}$. 
	
	In particular, in the non-dynamical setting, we will be interested in asymptotics for
	\begin{flalign}
	\label{htsr}
	\mathfrak{H}_T (s, r) = T^{- 1 / 2} \mathfrak{h}_{\lfloor r T \rfloor} \bigg( \Big\lfloor \displaystyle\frac{J r T}{J + 1} + s T^{1 / 2} \Big\rfloor \bigg), 
	\end{flalign}
	
	\noindent for $s, r \in \mathbb{R}$ and $T \in \mathbb{Z}_{> 0}$. 
	
	We will see that this current is governed by the heat equation. Specifically, define $\mathcal{H} (s, r)$ by the (unique) solution to the normalized heat equation $J \partial_r \mathcal{H} = 2 (J + 1)^2 \partial_s^2 \mathcal{H}$, with initial data $\mathcal{H} (s, 0) = (J + 1) |s| \textbf{1}_{s < 0}$; this solution has an explicit integral form given by \eqref{hsr} below. Then, we show as Proposition \ref{momentm} the convergence  
	\begin{flalign}
	\label{hh}
	\mathfrak{H}_T (s, r) \rightarrow \mathcal{H} (s, r), \qquad \text{as $T$ tends to $\infty$.} 
	\end{flalign}

	\noindent  In the dynamical setting, we are instead interested in asymptotics for the current 
	\begin{flalign}
	\label{htsrgamma}
	\mathfrak{H}_T (s, r; \gamma) = T^{-1 / 4} \mathfrak{h}_{\lfloor r T \rfloor} \bigg( \Big\lfloor \displaystyle\frac{J r T}{J + 1} + s T^{1 / 4} \Big\rfloor \bigg).
	\end{flalign}

	 We will find that the normalized current $\mathfrak{H}_T (s, r; \gamma)$ becomes random in the large $T$ limit. To explain further, recall that $\chi_{a, b}$ is an \textit{$(a, b)$-Gamma distributed random variable} if 
	\begin{flalign*}
	\mathbb{P} [\chi_{a, b} \le r] = \Gamma (a)^{-1} \displaystyle\int_0^r b^a x^{a - 1} e^{-bx} dx,
	\end{flalign*}
	
	\noindent for each $r \ge 0$, where the normalization constant $\Gamma (a)$ is the Gamma function. Then, as Theorem \ref{momentmdynamic}, we will show the convergence 
	\begin{flalign}
	\label{convergencedynamical1}
	\mathfrak{H}_T (s, r; \gamma) \rightarrow \sqrt{\displaystyle\frac{s^2 (J + 1)^2}{4} + \chi_{a, b}} - \displaystyle\frac{s (J + 1)}{2}, \qquad \text{where $a = \gamma$ and $b = \sqrt{\displaystyle\frac{2 \pi}{r J}}$}. 
	\end{flalign}
	
	\noindent Both convergences \eqref{hh} and \eqref{convergencedynamical} are established in a certain sense of moments (although it is plausible that this sense of convergence can be improved through a more refined analysis). We refer to Theorem \ref{momentm} and Theorem \ref{momentmdynamic} below for a more precise statement. 
	
	\begin{rem}
		\label{dynamicalnondynamicalcornergrowth} 
		
		Combining \eqref{dynamicalparameterpartialexclusion}, Remark \ref{cornerpartialexclusion}, and the $J = 1$ cases of \eqref{hh} and \eqref{convergencedynamical1} yields 
		\begin{flalign*}
		T^{-1 / 2} \zeta_{\lfloor r T \rfloor} (s T^{1 / 2}) \rightarrow 2 s + 2 \mathcal{H} (r, s), \quad \text{and} \quad T^{-1 / 4} \zeta_{\lfloor r T \rfloor } (s T^{1 / 4}) \rightarrow 2 \displaystyle\sqrt{s^2 + \chi_{a, b}},
		\end{flalign*}
		
		\noindent as $T$ tends to $\infty$, for the height functions $\zeta$ of the non-dynamical and dynamical (with parameter $\gamma$) midpoint corner growth models from Section \ref{DynamicalCornerGrowth}, respectively; here $\mathcal{H}(r, s)$ denotes the solution to a normalized heat equation $\partial_r \mathcal{H} = 8 \partial_s^2 \mathcal{H}$, with initial data $\mathcal{H} (r, s) = 2 |s| \textbf{1}_{s < 0}$.  
		\end{rem}

	\subsection{Outline}
	
	\label{Outline} 
	
	The algebraic framework underlying the proofs of the results stated above comes from two sources, both of which have by now been used quite extensively within the context of (quantum and elliptic) integrability. 
	
	The first is known as the \textit{fusion procedure}; it was introduced in a representation theoretic setting by Kirillov, Reshetikhin, and Sklyanin \cite{ERT} as a way of producing higher dimensional representations of quantum groups from lower dimensional representations. The matrix elements of the $R$-matrices from these fused representations give rise to vertex weights for models that are exactly solvable (but often not stochastic) through the Yang-Baxter equation. 
	
	The second is to use algebraic properties (branching, Cauchy, and Pieri identities) satisfied by families of symmetric functions. This was first implemented through the Schur measures and Schur processes introduced by Okounkov \cite{IWRP} and Okounkov-Reshetikhin \cite{CPLGR,EIMAS}, respectively. The branching rules for the underlying symmetric functions can be used to produce Markov processes with integrable dynamics \cite{DP}, and the Cauchy and Pieri rules ensure that these processes are indeed stochastic. 
	
	The idea to combine these two methods first appeared in the Section 6 of the work of Borodin \cite{FSRF}. Here, the underlying symmetric functions were two-parameter deformations of the Schur functions introduced by Povolotsky \cite{IZRCMFSS} and later called the \textit{spin Hall-Littlewood functions} \cite{SP}. As observed in \cite{FSRF}, these symmetric functions can be viewed as partition functions for certain vertex models (related to representations of the affine quantum group $U_q (\widehat{\mathfrak{sl}_2})$) that are solvable through the Yang-Baxter equation. It was therefore possible to consider ``fused analogs'' \cite{FSRF, SP} of these symmetric functions that were soon later used by Corwin-Petrov \cite{SHSVML} to produce a large family of stochastic, quantum integrable systems, called the stochastic higher spin vertex models. Thus, the solvability of these models can be viewed as a consequence of the properties of these spin Hall-Littlewood symmetric functions and (sometimes fusion through) the Yang-Baxter equation.   
	
	Very recently, Borodin \cite{ERSF} analyzed elliptic generalizations of the spin Hall-Littlewood functions, which were introduced by Felder-Varchenko \cite{AEQG} as Bethe eigenfunctions for Felder's \cite{ISAEP} elliptic quantum group $E_{\tau, \eta} (\mathfrak{sl}_2)$. These elliptic functions (in fact, versions of them with more general boundary data) were also analyzed in the works of Tarasov-Varchenko \cite{HFQAAEQG} and Rosengren \cite{EQGEHSRS} under the name of \textit{elliptic weight functions}; therefore, we refer to the ones studied in \cite{ERSF} as \textit{elliptic weight (symmetric) functions} as well. It was shown in \cite{ERSF} that the elliptic weight functions also satisfy branching, Cauchy, and Pieri identities; thus it was possible to use these identities to create stochastic models, called stochastic IRF models, that are integrable through what is known as the \textit{dynamical Yang-Baxter equation}. 
	
	 Our results on the dynamical stochastic higher spin vertex model are based on a combination of the results from \cite{ERSF} on the elliptic weight functions with the fusion procedure applied to $E_{\tau, \eta} (\mathfrak{sl}_2)$. This happens to be more intricate than what was done in the quantum case (through spin Hall-Littlewood functions \cite{HSVMRSF,SHSVML}) and essentially proceeds in four parts.
	 
	 The first part is to introduce ``fused versions'' of the elliptic weight functions studied in \cite{ERSF}, and explain how they can be evaluated as partition functions of a vertex model with certain ``fused weights'' $W_J = W_J (i_1, j_1; i_2, j_2)$. This is a mild modification of what was done in the quantum setting (see, for instance, Proposition 5.5 of \cite{HSVMRSF}) and is implemented in Section \ref{EllipticSymmetric} below. 
	 
	 The second part is to explicitly determine	 these fused weights. In the quantum case, this was done by Mangazeev \cite{ESVM} and Corwin-Petrov \cite{SHSVML} (see also the older work of Kirillov-Reshetikhin \cite{REA}). In that setting, the fused weights are expressible in terms of $q$-Racah polynomials, which are balanced ${{_4} \varphi_3}$ basic hypergeometric series. In our situation, the fused weights are considerably more complex, as they are given by very-well poised, balanced ${{_{12}} v_{11}}$ elliptic hypergeometric series (see Theorem \ref{fusedweighthypergeometric}); as a result, we do not know of a way to use the methods from \cite{SHSVML,ESVM,REA} to access them. Instead, we evaluate the $W_J$ by implementing a recursive procedure, originally due to Date-Jimbo-Kuniba-Miwa-Okado \cite{ESSOSM} in the closely related context of the fused eight vertex solid-on-solid (SOS) model \cite{ESSOSMLHP,ESSOSM,FEVSOSM}. This is done in Section \ref{FusedWeights} below. 
	 
	The issue at this point is that these fused weights are not stochastic and therefore do not give rise to stochastic processes. Thus, the third part is to multiply these fused weights by a ``stochastic correction'' $C_J = C_J (i_1, j_1; i_2, j_2)$ so that these ``stochastically corrected'' weights $\sigma_J = \sigma_J (i_1, j_1; i_2, j_2)$ sum to one and still give rise to an integrable model. In the quantum setting, this amounted to a quick gauge transformation of the original weights $W_J$ (see Remark 2.2 of \cite{SHSVML}). This relatively simple stochastic correction could be explained by the fact that the stochastically corrected weights $\{ \sigma_1 \}$ still satisfied the Yang-Baxter equation. Thus, the fusion procedure could instead be applied to the $\{ \sigma_1 \}$ weights to directly produce the $\sigma_J$ weights that were stochastic and defined an integrable system; indeed, this was the route taken in Section 3 of \cite{SHSVML}.
	
	In the elliptic case, finding a suitable family of stochastic corrections $\{ C_J \}$ appears to be considerably more involved. By considering the trigonometric degeneration (letting the $\tau$ parameter in $E_{\tau, \eta} (\mathfrak{sl}_2)$ tend to $\textbf{i} \infty$) of this elliptic setup,  \cite{ERSF} proposed a family of stochastic corrections $\{ C_1 \}$ such that the resulting $\{ \sigma_1 \}$ weights still gave rise to an integrable stochastic process. However, the trouble is that this stochastic correction does not appear to be compatible with fusion, in that the $\{ \sigma_1 \}$ weights do not satisfy the dynamical Yang-Baxter equation. Thus, it does not seem to quickly follow that one can fuse them. 
	
	In Section \ref{Stochastic} below we overcome this issue by proposing a generalization $\{ C_J \}$ of the $\{ C_1 \}$ corrections from  \cite{ERSF} so that the associated $\sigma_J$ weights give rise (in a way to be explained in Section \ref{Transition}) to an integrable stochastic process. These stochastic corrections $\{ C_J \}$ are a bit more intricate than the $\{ C_1 \}$ corrections produced in \cite{ERSF}. In fact, showing that the resulting $\sigma_J$ weights are indeed stochastic is not immediately apparent to us, and we present a proof of this fact (through a Pieri-type identity for the elliptic weight functions) as Proposition \ref{psisumstochastic}. 
	
	The fourth part is to establish contour integral identities for observables of the dynamical stochastic higher spin vertex models. This follows from a suitable analytic continuation of similar results from \cite{ERSF} in the $J = 1$ setting and is implemented in Section \ref{ObservablesDegenerations} below. As an application, we use these identities to access the large-scale statistics of the dynamical $(J; \gamma)$-PEP in Section \ref{ObservablesDegenerations}.

\subsection*{Acknowledgments}

The author heartily thanks Alexei Borodin for many valuable conversations, for helpful comments on an early draft of this manuscript, and for sending him an early copy of the paper \cite{ERSF}. This work was funded by the NSF Graduate Research Fellowship under grant number DGE1144152.

\section{Miscellaneous Preliminaries} 

In this section we collect some preliminaries that will be useful to us later. In Section \ref{Parameters} we provide a brief list of some common notation to appear in this paper. In Section \ref{Hypergeometric} we recall the definitions of some elliptic functions and hypergeometric series that we will encounter throughout the article. In Section \ref{InteractingParticles} we define the current for a path ensemble and explain how the dynamical stochastic higher spin vertex model can be interpreted as an interacting particle system. 
	
\subsection{Notation} 

\label{Parameters}

In this section we list some notation that will appear commonly throughout the article. We begin by listing several complex parameters that will serve different purposes in this paper; explaining their relationships, and specifying their domains. 

\begin{itemize}

\item Fixed parameters: Let $q \in \mathbb{C}$ and $\eta, \tau \in \mathbb{C}$, with $\Im \tau > 0$. We often impose $q = e^{- 4 \pi \textbf{i} \eta}$. 

\item Spectral parameters: Let $U = (u_1, u_2, \ldots ) \subset \mathbb{C}$ and $W = (w_1, w_2, \ldots ) \subset \mathbb{C}$ denote countably infinite families of complex numbers. We often relate them by $u_j = e^{2 \pi \textbf{i} (\eta - w_j)}$ for each $j$.

\item Spin parameters: Let $S = (s_0, s_1, \ldots ) \subset \mathbb{C}$ and $L = (\Lambda_0, \Lambda_1, \ldots ) \subset \mathbb{C}$ denote countably infinite sets of complex numbers. We often relate them by $s_j = e^{2 \pi \textbf{i} \eta \Lambda_j}$ for each $j$.

\item Inhomogeneity parameters: Let $\Xi = (\xi_0, \xi_1, \ldots ) \subset \mathbb{C}$ and $Z = (z_0, z_1, \ldots ) \subset \mathbb{C}$ denote countably infinite sets of complex numbers. We often relate them by $\xi_j = e^{2 \pi \textbf{i} z_j}$ for each $j$, and we also often denote $v_{i, j} = w_j - z_i - \eta$ for each $i, j$. 

\item Fusion parameters: Let $\mathcal{J} = (J_1, J_2, \ldots ) \subset \mathbb{C}$. We will often impose that $\mathcal{J} \subset \mathbb{Z}_{\ge 1}$. 

\item $q$-Hahn parameters: Let $B = (b_1, b_2, \ldots ) \subset \mathbb{C}$ and $C = (c_1, c_2, \ldots ) \subset \mathbb{C}$ denote countably infinite sets of complex numbers; further denote $a_{x, y} = b_x c_y$ for each $x, y > 1$. We often relate $B$ and $C$ to the above parameters by $b_x = s_x^2$ and $c_y = q^{J_y}$ for each $x, y$. 

\item Multiplicative dynamical parameters: Let $\delta \in \mathbb{C}$. For any path ensemble (or assignment of arrow configurations to the positive quadrant satisfying arrow configuration), define $\kappa_{x, y}$ for each $(x, y) \in \mathbb{Z}_{> 0} \times \mathbb{Z}_{\ge 0}$ through $\kappa_{1, 0} = \delta$ and \eqref{kappamultiplicative}. 

\item Additive dynamical parameters: Let $\lambda \in \mathbb{C}$. For any path ensemble (or assignment of arrow configurations to the positive quadrant satisfying arrow configuration), define $\Phi_{x, y}$ for each $(x, y) \in \mathbb{Z}_{> 0} \times \mathbb{Z}_{\ge 0}$ through $\Phi_{1, 0} = \lambda$ and 
\begin{flalign}
\label{lambdaadditive}
\Phi_{x, y + 1} = 2 \eta J_{y + 1} - 4 \eta j_1 (x, y + 1) + \Phi_{x, y}; \qquad \kappa_{x + 1, y} = 4 \eta i_2 (x, y) - 2 \eta \Lambda_x + \Phi_{x, y}.  
\end{flalign} 

\noindent In particular, the multiplicative dynamical parameters $\kappa_{x, y}$ can be obtained from the additive dynamical parameters $\Phi_{x, y}$ by setting $\delta = e^{- 2 \pi \textbf{i} \lambda}$ and $\kappa_{x, y} = e^{-2 \pi \textbf{i} \Phi_{x, y}}$ for all $x \in \mathbb{Z}_{> 0}$ and $y \in \mathbb{Z}_{\ge 0}$. 

\end{itemize}

We will also come across various types of vertex weights, which are listed below. 

\begin{itemize} 
	
\item General hypergeometric stochastic weights: Fixing $u, s, q, \kappa \in \mathbb{C}$ and $J \in \mathbb{Z}$, the stochastic weight $\psi_{u; s; q; J} \big( i_1, j_1; i_2, j_2 \b| \kappa \big)$ governing the dynamical stochastic higher spin vertex model is given by Definition \ref{psiweightdefinition} for any $i_1, j_1, i_2, j_2 \in \mathbb{Z}_{\ge 0}$. 

\item $q$-Hahn stochastic weights: Fixing $q, a, b, \kappa \in \mathbb{C}$, the stochastic weight $\varphi_{q, a, b, \kappa} \big( j \b| i \big)$ governing the dynamical $q$-Hahn boson model is given by Definition \ref{stochasticparticlesjumps} for any $i_1, j_1, i_2, j_2 \in \mathbb{Z}_{\ge 0}$. 

\item Unfused vertex weights: Fixing $v, \lambda, \Lambda \in \mathbb{C}$, the vertex weight $W_1 (i_1, j_1; i_2, j_2 \b| v, \lambda, \Lambda)$ is given by \eqref{weightsj1} for any integers $i_1, i_2 \in \mathbb{Z}_{\ge 0}$ and $j_1, j_2 \in \{ 0, 1 \}$. 

\item Fused vertex weights: Fixing $v, \lambda, \Lambda \in \mathbb{C}$ and $J \in \mathbb{Z}_{> 0}$, one defines the fused vertex weight $W_J (i_1, j_1; i_2, j_2 \b| v, \lambda, \Lambda)$ by \eqref{fusedweightdefinition} (and also explicitly by Theorem \ref{fusedweighthypergeometric}) for any integers $i_1, j_1, i_2, j_2 \in \mathbb{Z}_{\ge 0}$. 

\item Summed fused weights: Fixing $v, \lambda, \Lambda \in \mathbb{C}$ and $J \in \mathbb{Z}_{> 0}$, the summed fused vertex weight $\widehat{W}_J (i_1, j_1; i_2, j_2 \b| v, \lambda, \Lambda)$ is given by \eqref{wjwsums} for any integers $i_1, j_1, i_2, j_2 \in \mathbb{Z}_{\ge 0}$. These are explicitly related to the fused vertex weights $W_J$ by \eqref{hatwwj}. 

\item Stochastic corrections: Fixing $\lambda, \Lambda \in \mathbb{C}$, the stochastic correction $C_J (i_1, j_1; i_2, j_2 \b| \lambda, \Lambda)$ is given by Definition \ref{cstochastic} for any integers $i_1, j_1, i_2, j_2 \in \mathbb{Z}_{\ge 0}$. 

\item Stochastic weights: Fixing $v, \lambda, \Lambda \in \mathbb{C}$, the stochastic weight $\sigma_J (i_1, j_1; i_2, j_2 \b| v, \lambda, \Lambda)$ is given by \eqref{sigmajdefinition} for any integers $i_1, j_1, i_2, j_2 \in \mathbb{Z}_{\ge 0}$. It is shown as Lemma \ref{lpsiweights} that under the transformations mapping the $\eta, v, \lambda, \Lambda$ parameters to the $\eta, u, \kappa, s$ parameters, $\sigma_J$ becomes the general hypergeometric stochastic weight $\psi$ above. 

\item In some contexts (in Section \ref{DynamicalParticles} and Section \ref{PsiStochasticity}) we will abbreviate $\sigma (j) = \sigma_J (i_1, J; i_2, j_2)$, especially when $i_1 = 0$. 

\end{itemize}
	
A \emph{non-negative signature} $\mu = (\mu_1, \mu_2, \ldots , \mu_n)$ of \emph{length} $n$ is a non-increasing sequence of $n$ integers $\mu_1 \ge \mu_2 \ge \cdots \ge \mu_n \ge 0$; let us provide some notation for signatures. 

\begin{itemize}
	
\item Signatures will commonly be denoted by $\mu$ and $\nu$. 

\item The length of $\mu$ is denoted $\ell (\mu) = n$. 

\item The set of non-negative signatures of length $n$ is denoted by $\Sign_n^+$, and the set of all non-negative signatures is denoted by $\Sign^+ = \bigcup_{n = 0}^{\infty} \Sign_n^+$. 

\item For any signature $\mu$ and integer $j$, let $m_j (\mu)$ denote the number of indices $i$ for which $\mu_i = j$, that is, $m_j (\mu)$ is the multiplicity of $j$ in $\mu$. 

\item For any $\mu \in \Sign^+$ and $k \in \mathbb{Z}_{\ge 0}$, let $\mathfrak{h}_{\mu} (k)$ denote the number of indices $j$ for which $\mu_j \ge x$. 

\item For each positive integer $n$, $\textbf{P}_n$ denotes the measure on $\Sign^+$, given explicitly by Definition \ref{measuresignaturesvertexmodel}, that is governed by the stochastic higher spin vertex model $\mathcal{P}$.

\item For each positive integer $n$, $\textbf{M}_n$ denotes the measure on $\Sign^+$, given explicitly by Definition \ref{measuresignaturesvertexmodel}, that is governed by the dynamical $q$-Hahn boson model $\mathcal{M}$.
	
\end{itemize}
	
\subsection{Theta Functions and Hypergeometric Series}

\label{Hypergeometric}

In this section we briefly recall the definitions and some properties of several types of elliptic functions and hypergeometric series that we will encounter later in this paper. A more comprehensive survey of this very vast theory can be found in the book \cite{BHS}. As in Section \ref{Parameters}, we fix $\eta, \tau, q \in \mathbb{C}$ with $\Im \tau > 0$. 	

Throughout this article, we will repeatedly come across the \emph{first Jacobi theta function} $\theta_1 (z; \tau) = \theta (z; \tau) = \theta (z)$, defined (for any $z \in \mathbb{C}$) by 
\begin{flalign}
\label{theta1}
\theta (z) = - \displaystyle\sum_{j = - \infty}^{\infty} \exp \left( \pi \textbf{i} \tau \left( j + \displaystyle\frac{1}{2} \right)^2 + 2 \pi \textbf{i} \left( j + \displaystyle\frac{1}{2} \right) \left( z + \displaystyle\frac{1}{2} \right) \right), 
\end{flalign}

\noindent which converges since $\Im \tau > 0$. Observe that these theta functions degenerate to trigonometric functions, in that $\lim_{\tau \rightarrow \textbf{i} \infty} e^{- \pi \textbf{i} \tau / 4} \theta (z) = 2 \sin (\pi z)$. We will use the notation $f(z)$ to denote either $\theta (z)$ or $\sin (\pi z)$, depending on whether the limit as $\tau$ tends to $\textbf{i} \infty$ has been taken. The former case will sometimes be referred to as the \emph{elliptic case} and the latter as the \emph{trigonometric case}. 

In either case, $f$ satisfies the \textit{Riemann identity} 
\begin{flalign}
\label{quarticrelationf}
\begin{aligned}
f(x + z) & f(x - z) f(y + w) f(y - w) \\
& = f(x + y) (x - y) f(z + w) f (z - w) + f(x + w) f(x - w) f(y + z) f(y - z), 
\end{aligned}
\end{flalign}

\noindent for any $w, x, y, z \in \mathbb{C}$.

For any complex number $a \in \mathbb{C}$ and nonnegative integer $k$, the \textit{rational Pochhammer symbol} $(a)_k$; \emph{$q$-Pochhammer symbol} $(a; q)_k$ and \emph{elliptic Pochhammer symbol} $[a]_k = [a; \eta]_k$ are defined by 
\begin{flalign}
\label{productbasicelliptic}
(a)_k = \prod_{j = 0}^{k - 1} (a + j); \qquad (a; q)_k = \prod_{j = 0}^{k - 1} (1 - q^j a); \qquad [a]_k = \prod_{j = 0}^{k - 1} f(a - 2 \eta j).
\end{flalign}

\noindent The definitions of these Pochhammer symbols are extended to negative integers $k$ by 
\begin{flalign*}
(a)_k = \displaystyle\prod_{j = 1}^{-k} \displaystyle\frac{1}{a - j}; \qquad (a; q)_k = \displaystyle\prod_{j = 1}^{-k} \displaystyle\frac{1}{1 - q^{-j} a}; \qquad [a]_k = \displaystyle\prod_{j = 1}^{-k} \displaystyle\frac{1}{f (a + 2 \eta j)}.
\end{flalign*}

\noindent One can quickly verify that the $q$-Pochhammer symbols satisfy the identities
\begin{flalign}
\label{basichypergeometric1}
\begin{aligned}
&  (a; q)_k (a q^k; q)_m = (a; q)_{k + m}; \qquad \qquad (q^{m - k} a; q)_{m - k} (q^{2m - 2k + 1} a; q)_k = \displaystyle\frac{(q^{m - k} a; q)_{m + 1}}{1 - q^{2m - 2k} a}; \\
& (q^{-k} a; q)_m = \displaystyle\frac{(a; q)_m \big( \frac{q}{a}; q \big)_k}{q^{mk} \big( \frac{1}{q^{m - 1} a}; q \big)_k}; \qquad \quad  (a; q)_{m - k} = \displaystyle\frac{q^{\binom{k + 1}{2} - mk}}{(-a)^k} \displaystyle\frac{(a; q)_m}{\big( \frac{1}{a q^{m - 1}}; q \big)_k}; \\
& (a; q)_k (-a; q)_k = (a^2; q)_k;  \qquad \qquad \quad \displaystyle\frac{(q^2 a; q^2)_k}{(a; q^2)_k} = \displaystyle\frac{1 - q^{2k} a}{1 - a} = \displaystyle\frac{a^{-1} - q^{2k}}{a^{-1} - 1}; \\
& \qquad \qquad \qquad \qquad \quad \displaystyle\lim_{q \rightarrow 1} (1 - q)^{-k} (q^a; q)_k = (a)_k. 
\end{aligned}
\end{flalign}

\noindent for any $a \in \mathbb{C}$ and $k, m \in \mathbb{Z}$. Similarly, the elliptic Pochhammer symbols satisfy the identities 
\begin{flalign}
\label{elliptichypergeometricidentities1}
\begin{aligned} 
& [a]_m = (-1)^m \big[ 2 \eta (m - 1) - a \big]_k; \qquad [a]_{m - k} = \displaystyle\frac{[a]_m}{[a - 2 \eta k]_k} = \displaystyle\frac{(-1)^k [a]_m}{\big[ 2 \eta (m - 1) - a \big]_k}; \\
& \qquad \qquad \qquad \quad \big[ a \big]_k \big[ a + 2 \eta (k + 1) \big]_{m - k} = \displaystyle\frac{[a]_{m + 1}}{f (a + 2 \eta k)}. 
\end{aligned} 
\end{flalign}

\noindent From the Pochhammer symbols, we can define the \emph{basic hypergeometric series},
\begin{flalign}
\label{basichypergeometric}
_p \varphi_r \left( \begin{array}{ccccc}  a_1, a_2, \ldots , a_p \\ b_1, b_2, \ldots , b_r \end{array} \right| q, z \bigg) = \displaystyle\sum_{k = 0}^{\infty} \displaystyle\frac{z^k}{ \big( q; q \big)_k} \displaystyle\prod_{j = 1}^p (a_j; q)_k \displaystyle\prod_{j = 1}^r (b_j; q)_k^{-1},
\end{flalign}

\noindent and the \emph{elliptic hypergeometric series}, 
\begin{flalign}
\label{elliptichypergeometric}
_p e_r \left( \begin{array}{ccccc}  a_1, a_2, \ldots , a_p \\ b_1, b_2, \ldots , b_r \end{array} \right| z  \bigg) = \displaystyle\sum_{k = 0}^{\infty} \displaystyle\frac{z^k}{ [-2 \eta]_k} \displaystyle\prod_{j = 1}^p [a_j]_k \displaystyle\prod_{j = 1}^r [b_j]_k^{-1},
\end{flalign}

\noindent for any positive integers $p, r$ and complex numbers $a_1, a_2, \ldots , a_p; b_1, b_2, \ldots , b_r; z \in \mathbb{C}$. Here, we must assume that the series \eqref{basichypergeometric} and \eqref{elliptichypergeometric} converge. In the situations we come across, this will be guaranteed since these series will always \textit{terminate}. For \eqref{basichypergeometric} this means that some $a_j = q^{-r}$, and for \eqref{elliptichypergeometric} this means that some $a_j = 2 \eta r$, where in both cases $r$ is a nonnegative integer. 

While the basic hypergeometric series are over a century old, elliptic hypergeometric series were introduced to the mathematics community much more recently, by Frenkel and Turaev \cite{ESEMHF} in 1997 (although they also implicitly appeared in the physics work of Date-Jimbo-Kuniba-Miwa-Okado \cite{ESSOSM} almost a decade previously) in their study of elliptic solutions to the Yang-Baxter equation. Since that work, elliptic hypergeometric series have been studied extensively; we refer to the surveys \cite{EHF, TEHF} or Chapter 11 of the book \cite{BHS} for more information. 

It is well-understood that both basic and elliptic hypergeometric series become significantly more manageable under certain specializations of their parameters. The two we will discuss here are \textit{very well-poised hypergeometric series} and \textit{balanced hypergeometric series}. 

The hypergeometric series \eqref{basichypergeometric} is called \emph{very well-poised} if $p = r + 1$, $q a_1 = a_{k + 1} b_k$ for each $k \in [1, r]$, $a_2 = q a_1^{1 / 2}$, and $a_3 = - q a_1^{1 / 2}$. In this case, we denote the resulting basic hypergeometric series by $_{r + 1} W_r (a_1; a_4, a_5, \ldots , a_{r + 1}; q, z)$, which is explicitly given by \eqref{poisedhypergeometric} (due to the fifth and sixth identities in \eqref{basichypergeometric1}).

 Analogously, the \emph{very-well poised elliptic hypergeometric series} is defined by 
\begin{flalign}
\label{elliptichypergeometricpoised}
_{r + 1} v_r (a_1; a_6, a_7, \ldots , a_{r + 1}; z) = \displaystyle\sum_{k = 0}^{\infty} \displaystyle\frac{z^k [a_1]_k}{ [-2 \eta]_k} \displaystyle\frac{f (a_1 - 4 \eta k)}{f (a_1)} \displaystyle\prod_{j = 6}^{r + 1} \displaystyle\frac{[a_j]_k}{[a_1 - a_j - 2 \eta ]_k},
\end{flalign}

\noindent which can similarly be viewed as a special case of the elliptic hypergeometric series \eqref{elliptichypergeometric} where $a_2$, $a_3$, $a_4$, and $a_5$ are fixed in terms of $a_1$, $\tau$, and $\eta$; see Section 11.3 of \cite{BHS}. We also call the very-well poised elliptic hypergeometric series $_{r + 1} v_r (a_1; a_6, a_7, \ldots , a_{r + 1})$ \textit{balanced} if 
\begin{flalign*}
\displaystyle\frac{(r - 5) (a_1 - 2 \eta)}{2} = \sum_{j = 6}^{r + 1} a_j - 2 \eta, 
\end{flalign*}

\noindent which is often imposed in order to ensure total ellipticity of the sum \eqref{elliptichypergeometricpoised}. 

There are several situations when very-well poised or balanced hypergeometric series simplify; one is given by the Rogers summation identity for very well-poised $_6 W_5$, terminating basic hypergeometric series, which states (see equation (2.4.2) of \cite{BHS}) that
\begin{flalign}
\label{hypergeometric65sumterminating}
_6 W_5 \left( a; b, c, q^{-n}; q, \displaystyle\frac{a q^{n + 1}}{bc} \right) = \displaystyle\frac{(aq; q)_n \left( \frac{aq}{bc}; q \right)_n}{\left( \frac{aq}{b}; q \right)_n \left( \frac{aq}{c}; q \right)_n}.  
\end{flalign}

\noindent for any $a, b, c, q \in \mathbb{C}$ and $n \in \mathbb{Z}_{\ge 0}$. 

One has a similar simplification for the very well-poised, balanced $_{10} v_{9}$ elliptic hypergeometric series. Specifically, the \textit{(terminating) elliptic Jackson identity} (which was originally found by Frenkel-Turaev \cite{ESEMHF}; see also equation (11.3.19) of \cite{BHS}, in which all parameters are divided by $-2 \eta$) states that 
\begin{flalign}
\label{hypergeometric109sumterminating} 
_{10} v_9 (a; b, c, d, e, -n; 1) = \displaystyle\frac{[a - 2 \eta]_n [a - b - c - 2 \eta]_n [a - b - d - 2 \eta]_n [a - c - d - 2 \eta]_n}{[a - b - 2 \eta]_n [a - c - 2 \eta]_n [a - d - 2 \eta]_n [a - b - c - d - 2 \eta]_n},
\end{flalign}

\noindent for any nonnegative integer $n$. 

We will encounter the above identities and simplifications later in the paper; specifically, we require \eqref{hypergeometric65sumterminating} in Section \ref{Stochastic} to verify stochasicity of certain weights, and we come across \eqref{hypergeometric109sumterminating} in Section \ref{FusedWeights} when we consider how fused weights in the elliptic quantum group simplify under specific choices of parameters.

\subsection{Currents and Interacting Particle Systems}

\label{InteractingParticles}

In this section we define the current of a vertex model and explain how random path ensembles can be viewed as interacting particle systems. 

We begin with the former. Let $\mathcal{E}$ be a directed path ensemble (as in Section \ref{PathEnsembles}, with or without dynamical parameter), in which every path exits the strip $\mathbb{Z}_{\ge 0} \times [0, N]$ through its top boundary, for each $N > 0$. Associated with each vertex $(X, Y) \in \mathbb{Z}_{> 0} \times \mathbb{Z}_{\ge 0}$ in $\mathcal{E}$ is an observable called the \textit{current} $\mathfrak{h} (X, Y) = \mathfrak{h}_Y (X)$. To define it, color a directed path in $\mathcal{E}$ blue if it emanates from the $y$-axis, and color it red if it emanates from the $x$-axis. Then, $\mathfrak{h}_Y (X)$ denotes the number of red paths that intersect the line $y = Y$ to the left of $(X, Y)$ subtracted from the number of blue paths that intersect the line $y = Y$ at or to the right of $(X, Y)$. In the special case when there are no red paths (meaning that all paths enter through the $y$-axis), $\mathfrak{h}_Y (X)$ denotes the number of paths that intersect the line $y = Y$ at or to the right of $(X, Y)$. Observe further that $\mathfrak{h}_0 (1) = 0$, independent of the boundary data.

Understanding scaling properties (for instance, a law of large numbers and fluctuations) of the current of random path ensembles has long been a topic of great interest within the probability and mathematical physics literature \cite{CFSAEPSSVMCL,PTAEPSSVM,RWBRE,AEP,SHSVMM,ERSF,MP,SSVM,IEJM,SWLP,SVMCP,PTAP,SLIS,SHSVMAEP,LAEP}. When the random path ensemble happens to be integrable, one can hope to establish exact identities for this current and use these identities to recover information about the large-scale behavior of $\mathfrak{h}$. In the case of path ensembles without a dynamical parameter (where the corresponding vertex models are the stochastic higher spin vertex models introduced in \cite{SHSVML}), such identities were recently established in \cite{SHSVMM,HSVMRSF,SHSVML,SHSVMAEP}, which were later used to study asymptotics in \cite{CFSAEPSSVMCL,PTAEPSSVM,SHSVMM,IEJM,PTAP,SHSVMAEP}. For \textit{unfused} path ensembles with a dynamical parameter (meaning paths cannot share horizontal edges), analogous identities were obtained in \cite{ERSF}, in which work a scaling limit was also analyzed in one special case. One of the goals of the present paper is to extend those results to obtain exact identities and asymptotics for the currents of some fused, random directed path ensembles with a dynamical parameter; we will explain this further in Section \ref{ObservablesDegenerations}. 

Observe that the rules for the dynamical parameter, listed as \eqref{kappamultiplicative} in the multiplicative case and as \eqref{lambdaadditive} in the additive case, can be quickly interpreted in terms of the current. Specifically, one defines the dynamical parameters through  
\begin{flalign}
\label{dynamicalparametercurrent}
\Phi_{x, y} = \lambda - 4 \eta \mathfrak{h} (x, y) + 2 \eta \displaystyle\sum_{k = 1}^y J_k - 2 \eta \displaystyle\sum_{k = 1}^{x - 1} \Lambda_k; \qquad \kappa_{x, y} = q^{-2 \mathfrak{h} (x, y)} \delta \displaystyle\prod_{k = 1}^{x - 1} b_k \displaystyle\prod_{k = 1}^y c_k, 
\end{flalign}

\noindent for each vertex $(x, y)$, in the additive and multiplicative cases, respectively; above, we denoted $b_x = s_x^2$ and $c_y = q^{J_y}$ for each $x, y$. 

Now let us discuss interacting particle systems. To that end, let $\mathcal{P}^{(T)}$ denote the restriction of the random path ensemble given by the dynamical stochastic higher spin vertex model $\mathcal{P}$ (from Section \ref{ProbabilityMeasures}) to the strip $\mathbb{Z}_{> 0} \times [0, T]$. Assume that all paths in this restriction almost surely exit the strip $\mathbb{Z}_{> 0} \times [0, T]$ through its top boundary; this will be the case, for instance, if $\mathcal{P}_n \big[ (i_2, j_2) = (0, k) \big| (i_1, j_1) = (0, k) \big] < 1$ uniformly in $n$ and $k$. Define $\mathcal{M}^{(T)}$ in an entirely analogous way, with the difference that $\mathcal{P}$ is replaced by the probability measure $\mathcal{M}$ for the dynamical $q$-Hahn boson model (from Section \ref{LambdaSpecializationModelGeneral}).

We will use the probability measure $\mathcal{P}^{(T)}$ (or $\mathcal{M}^{(T)}$) to produce a discrete-time interacting particle system on $\mathbb{Z}_{> 0}$, defined up to time $T - 1$, as follows. Sample a line ensemble $\mathcal{E}$ randomly under $\mathcal{P}^{(T)}$, and consider the arrow configuration it associates with some vertex $(p, t) \in \mathbb{Z}_{> 0} \times [1, T]$. We will place $k$ particles at position $p$ and time $t - 1$ if and only if $i_1 = k $ at the vertex $(p, t)$. Therefore, the paths in the path ensemble $\mathcal{E}$ correspond to space-time trajectories of the particles. 

We can associate a configuration of $n$ particles in $\mathbb{Z}_{\ge 0}$ with a signature of length $n$ as follows. A signature $\mu = (\mu_1, \mu_2, \ldots , \mu_n)$ is associated with the particle configuration that has $m_j (\mu)$ particles at position $j$, for each non-negative integer $j$. Stated alternatively, $\mu$ is the ordered set of positions in the particle configuration. If $\mathcal{E}$ is a directed line ensemble on $\mathbb{Z}_{\ge 0} \times [0, n]$, let $p_n (\mathcal{E}) \in \Sign^+$ denote the signature associated with the particle configuration produced from $\mathcal{E}$ at time $n$. 

Then, $\mathcal{P}^{(n)}$ and $\mathcal{M}^{(n)}$ induce probability measures on $\Sign^+$,  defined as follows. 

\begin{definition}
	
	\label{measuresignaturesvertexmodel} 
	
	Let $q$, $\delta$, $U$, $\Xi$, $S$, and $\mathcal{J}$ be sets of parameters as in Section \ref{ModelGeneral}. For any $n \in \mathbb{Z}_{> 0}$, let $\textbf{P}_n$ denote the measure on $\Sign^+$ (dependent on $q$, $\delta, U, \Xi, S, \mathcal{J}$) defined by setting $\textbf{P}_n (\mu) = \mathcal{P}^{(n)} [p_n (\mathcal{E}) = \mu]$, for each $\mu \in \Sign^+$, where we run the dynamical stochastic higher spin vertex model $\mathcal{P}$ with $\mathcal{J}$-step initial data. 
	
	Similarly, let $q$, $\delta$, $B$, and $C$ be sets of parameters as in Section \ref{LambdaSpecializationModelGeneral}, and assume that $c_y = q^{J_y}$ for each $y > 0$, where $\mathcal{J} = (J_1, J_2, \ldots ) \subset \mathbb{Z}_{> 0}$. For any positive integer $n$, let $\textbf{M}_n$ denote the measure on $\Sign^+$ (dependent on $q, \delta, B, C$) defined by setting $\textbf{M}_n (\mu) = \mathcal{M}^{(n)} [p_n (\mathcal{E}) = \mu]$, for each $\mu \in \Sign^+$, where we again run the dynamical $q$-Hahn boson model $\mathcal{M}$ with $\mathcal{J}$-step initial data.
\end{definition}

\section{Fused Elliptic Symmetric Functions} 

\label{EllipticSymmetric} 

The integrability of the dynamical stochastic higher spin vertex model (defined in Section \ref{ProbabilityMeasures}) is based on properties of a family of a elliptic symmetric functions that arise from fused representations of Felder's \cite{ISAEP} elliptic quantum group $E_{\tau, \eta} (\mathfrak{sl}_2)$. The idea to use elliptic symmetric functions to produce and solve dynamical stochastic processes was very recently introduced in a simpler setting by Borodin \cite{ERSF}. There, the symmetric functions he studied (called elliptic weight functions) can be viewed as partition functions for vertex models that take place on unfused directed path ensembles with a dynamical parameter. Recall here that the term ``unfused'' means that no two paths of the ensemble can share a horizontal edge or, equivalently, $j_1, j_2 \in \{ 0, 1 \}$ at each vertex of the ensemble.

However, the dynamical stochastic higher spin models take place on more general path ensembles. Thus, in order to analyze these systems from this symmetric function viewpoint, we must generalize the elliptic weight functions to partition functions that can give \textit{fused} (meaning not unfused) directed path ensembles non-zero weight. The purpose of this section is to explain how to do this through the fusion procedure in $E_{\tau, \eta} (\mathfrak{sl}_2)$.

To that end, we begin by defining the elliptic weight functions in Section \ref{EllipticWeight} and Section \ref{EllipticGroupSymmetricFunctions}, and then explaining how fusion applies to $E_{\tau, \eta} (\mathfrak{sl}_2)$ in Section \ref{Fusion}. We next define fused versions of the elliptic weight functions in Section \ref{FusionSymmetric}.

\subsection{Elliptic Weight Symmetric Functions}

\label{EllipticWeight}

In this section we recall the definition and some properties of a family of elliptic symmetric functions studied in \cite{ERSF}. There are several ways to define these functions. The first is through the partition function of a vertex model with specific weights. The second is through the action of certain operators in the elliptic quantum group $E_{\tau, \eta} (\mathfrak{sl}_2)$ on finitary vectors of the infinite tensor product $V_{\Lambda_1} \otimes V_{\Lambda_2} \otimes \cdots$ of evaluation Verma modules. These two definitions are equivalent; the first and second are explained in Sections 7 and 3 of \cite{ERSF}, respectively. In this section we adopt the former definition; we describe the latter in Section \ref{EllipticGroupSymmetricFunctions}. 

The functions we introduce, which will be denoted by $B_{\mu / \nu}$ and $D_{\mu / \nu}$ (for signatures $\mu$ and $\nu$), will be partition functions of unfused directed path ensembles with a dynamical parameter that are weighted in a specific way. Each path ensemble consists of a collection of vertices and arrow configurations assigned to each vertex. To each vertex will be associated a \emph{vertex weight}, which will depend on the position, dynamical parameter, and arrow configuration of the vertex. The weight of the directed path ensemble will then be the product of the weights of all vertices in the ensemble. 

Let us explain this in more detail. For this section, we ``shift'' all directed path ensembles to the left by one in order to make our presentation consistent with that of Borodin \cite{HSVMRSF}. That is, any directed path in a path ensemble either contains an edge from $(-1, k + 1)$ to $(0, k + 1)$ or from $(k, 0)$ to $(k, 1)$, for some non-negative integer $k$. Arrow configurations are now defined on $\mathbb{Z}_{\ge 0} \times \mathbb{Z}_{> 0}$ instead of on $\mathbb{Z}_{> 0}^2$; see Figure \ref{fgpaths}.

\begin{figure}
	
	\begin{center}

		\begin{tikzpicture}[
		>=stealth,
		scale = 1.2
		]

		\draw[->, black, dashed	] (.5, 0) -- (.5, 3);
		\draw[->, black, dashed] (.5, 0) -- (3, 0);
		\draw[->,black, thick] (0, .5) -- (.45, .5);
		\draw[->,black, thick] (0, 1) -- (.45, 1);
		\draw[->,black, thick] (0, 1.5) -- (.45, 1.5);
		\draw[->,black, thick] (0, 2) -- (.45, 2);
		
		\draw[->, black, thick	] (.5, 0) -- (.5, .45);
		\draw[->,black, thick] (.53, .55) -- (.53, .95);
		\draw[->,black, thick] (.47, .55) -- (.47, .95);
		\draw[->,black, thick] (.55, 1) -- (.95, 1);
		\draw[->,black, thick] (.55, 1.5) -- (.95, 1.5);
		\draw[->,black, thick] (.53, 1.05) -- (.53, 1.45);
		\draw[->,black, thick] (.47, 1.05) -- (.47, 1.45);
		\draw[->,black, thick] (.53, 1.55) -- (.53, 1.95);
		\draw[->,black, thick] (.47, 1.55) -- (.47, 1.95);
		\draw[->,black, thick] (.55, 2.05) -- (.55, 2.45);
		\draw[->,black, thick] (.5, 2.05) -- (.5, 2.45);
		\draw[->,black, thick] (.45, 2.05) -- (.45, 2.45);

		\draw[->,black, thick] (1.05, 1.5) -- (1.45, 1.5);
		\draw[->,black, thick] (1, 1.55) -- (1, 1.95);
		\draw[->,black, thick] (1, 1.05) -- (1, 1.45);
		\draw[->,black, thick] (1, 2.05) -- (1, 2.45);
		\draw[->, black, thick	] (1, 0) -- (1, .45);
		\draw[->,black, thick] (1.05, .5) -- (1.45, .5);

		\draw[->,black, thick] (1.5, .55) -- (1.5, .95);
		\draw[->,black, thick] (1.5, 1.05) -- (1.5, 1.45);
		\draw[->,black, thick] (1.47, 1.55) -- (1.47, 1.95);
		\draw[->,black, thick] (1.53, 1.55) -- (1.53, 1.95);
		\draw[->,black, thick] (1.5, 2.05) -- (1.5, 2.45);
		\draw[->,black, thick] (1.55, 2) -- (1.95, 2);

		\draw[->,black, thick] (2.05, .5) -- (2.45, .5);
		\draw[->,black, thick] (2.05, 2) -- (2.45, 2);
		\draw[->, black, thick] (2, 0) -- (2, .45);
		
		\draw[->, black, thick] (2.5, .55) -- (2.5, .95);
		\draw[->, black, thick] (2.5, 1.05) -- (2.5, 1.45);
		\draw[->, black, thick] (2.5, 1.55) -- (2.5, 1.95);
		\draw[->, black, thick] (2.47, 2.05) -- (2.47, 2.5);
		\draw[->, black, thick] (2.53, 2.05) -- (2.53, 2.5);

		\filldraw[fill=gray!50!white, draw=black] (.5, .5) circle [radius=.05] node[scale = .6, below = 30]{$0$} node[scale = .6, left = 9, below = 3]{$1$};
		\filldraw[fill=gray!50!white, draw=black] (.5, 1) circle [radius=.05] node[scale = .6, left = 9, below = 3]{$2$};
		\filldraw[fill=gray!50!white, draw=black] (.5, 1.5) circle [radius=.05] node[scale = .6, left = 9, below = 3]{$3$};
		\filldraw[fill=gray!50!white, draw=black] (.5, 2) circle [radius=.05] node[scale = .6, left = 9, below = 3]{$4$};
		
		\filldraw[fill=gray!50!white, draw=black] (1, .5) circle [radius=.05] node[scale = .6, below = 30]{$1$};
		\filldraw[fill=gray!50!white, draw=black] (1, 1) circle [radius=.05];
		\filldraw[fill=gray!50!white, draw=black] (1, 1.5) circle [radius=.05];
		\filldraw[fill=gray!50!white, draw=black] (1, 2) circle [radius=.05];
		
		\filldraw[fill=gray!50!white, draw=black] (1.5, .5) circle [radius=.05] node[scale = .6, below = 30]{$2$};
		\filldraw[fill=gray!50!white, draw=black] (1.5, 1) circle [radius=.05];
		\filldraw[fill=gray!50!white, draw=black] (1.5, 1.5) circle [radius=.05];
		\filldraw[fill=gray!50!white, draw=black] (1.5, 2) circle [radius=.05];

		\filldraw[fill=gray!50!white, draw=black] (2, .5) circle [radius=.05] node[scale = .6, below = 30]{$3$};
		\filldraw[fill=gray!50!white, draw=black] (2, 1) circle [radius=.05];
		\filldraw[fill=gray!50!white, draw=black] (2, 1.5) circle [radius=.05];
		\filldraw[fill=gray!50!white, draw=black] (2, 2) circle [radius=.05];
		
		\filldraw[fill=gray!50!white, draw=black] (2.5, .5) circle [radius=.05] node[scale = .6, below = 30]{$4$};
		\filldraw[fill=gray!50!white, draw=black] (2.5, 1) circle [radius=.05];
		\filldraw[fill=gray!50!white, draw=black] (2.5, 1.5) circle [radius=.05];
		\filldraw[fill=gray!50!white, draw=black] (2.5, 2) circle [radius=.05];

		\draw[->, black, dashed] (7.5, 0) -- (7.5, 3);
		\draw[->, black, dashed] (7.5, 0) -- (10, 0);
		
		\draw[->, black, thick	] (7.5, 0) -- (7.5, .45);
		\draw[->,black, thick] (7.5, .55) -- (7.5, .95);
		\draw[->,black, thick] (7.55, 1) -- (7.95, 1);

		\draw[->,black, thick] (8.05, 0) -- (8.05, .45);
		\draw[->,black, thick] (8, 0) -- (8, .45);
		\draw[->,black, thick] (7.95, 0) -- (7.95, .45);
		\draw[->,black, thick] (8.03, .55) -- (8.03, .95);
		\draw[->,black, thick] (7.97, .55) -- (7.97, .95);
		\draw[->,black, thick] (8, 1.55) -- (8, 1.95);
		
		\draw[->,black, thick] (8.05, 1.5) -- (8.45, 1.5);
		\draw[->,black, thick] (8.03, 1.05) -- (8.03, 1.45);
		\draw[->,black, thick] (7.97, 1.05) -- (7.97, 1.45);
		\draw[->,black, thick] (8.05, 1) -- (8.45, 1);
		\draw[->,black, thick] (8.05, .5) -- (8.45, .5);
		\draw[->,black, thick] (8.05, 2) -- (8.45, 2);

		\draw[->,black, thick] (8.5, .55) -- (8.5, .95);
		\draw[->,black, thick] (8.5, 1.05) -- (8.5, 1.45);
		\draw[->,black, thick] (8.47, 1.55) -- (8.47, 1.95);
		\draw[->,black, thick] (8.53, 1.55) -- (8.53, 1.95);
		\draw[->,black, thick] (8.47, 2.05) -- (8.47, 2.45);
		\draw[->,black, thick] (8.53, 2.05) -- (8.53, 2.45);
		\draw[->,black, thick] (8.55, 2) -- (8.95, 2);
		\draw[->,black, thick] (8.55, 1) -- (8.95, 1);

		\draw[->,black, thick] (9.05, .5) -- (9.45, .5);
		\draw[->,black, thick] (9.05, 2) -- (9.45, 2);
		\draw[->, black, thick] (9, 0) -- (9, .45);
		\draw[->, black, thick] (9, 1) -- (9, 1.45);
		\draw[->, black, thick] (9, 1.55) -- (9, 1.95);
		\draw[->, black, thick] (9, 2.05) -- (9, 2.45);
		
		\draw[->, black, thick] (9.5, .55) -- (9.5, .95);
		\draw[->, black, thick] (9.5, 1.05) -- (9.5, 1.45);
		\draw[->, black, thick] (9.5, 1.55) -- (9.5, 1.95);
		\draw[->, black, thick] (9.47, 2.05) -- (9.47, 2.5);
		\draw[->, black, thick] (9.53, 2.05) -- (9.53, 2.5);

		\filldraw[fill=gray!50!white, draw=black] (7.5, .5) circle [radius=.05] node[scale = .6, below = 30]{$0$}  node[scale = .6, left = 3]{$1$};
		\filldraw[fill=gray!50!white, draw=black] (7.5, 1) circle [radius=.05]  node[scale = .6, left = 3]{$2$};
		\filldraw[fill=gray!50!white, draw=black] (7.5, 1.5) circle [radius=.05] node[scale = .6, left = 3]{$3$};
		\filldraw[fill=gray!50!white, draw=black] (7.5, 2) circle [radius=.05] node[scale = .6, left = 3]{$4$};
		
		\filldraw[fill=gray!50!white, draw=black] (8, .5) circle [radius=.05] node[scale = .6, below = 30]{$1$};
		\filldraw[fill=gray!50!white, draw=black] (8, 1) circle [radius=.05];
		\filldraw[fill=gray!50!white, draw=black] (8, 1.5) circle [radius=.05];
		\filldraw[fill=gray!50!white, draw=black] (8, 2) circle [radius=.05];
		
		\filldraw[fill=gray!50!white, draw=black] (8.5, .5) circle [radius=.05] node[scale = .6, below = 30]{$2$};
		\filldraw[fill=gray!50!white, draw=black] (8.5, 1) circle [radius=.05];
		\filldraw[fill=gray!50!white, draw=black] (8.5, 1.5) circle [radius=.05];
		\filldraw[fill=gray!50!white, draw=black] (8.5, 2) circle [radius=.05];
		
		\filldraw[fill=gray!50!white, draw=black] (9, .5) circle [radius=.05] node[scale = .6, below = 30]{$3$};
		\filldraw[fill=gray!50!white, draw=black] (9, 1) circle [radius=.05];
		\filldraw[fill=gray!50!white, draw=black] (9, 1.5) circle [radius=.05];
		\filldraw[fill=gray!50!white, draw=black] (9, 2) circle [radius=.05];

		\filldraw[fill=gray!50!white, draw=black] (9.5, .5) circle [radius=.05] node[scale = .6, below = 30]{$4$};
		\filldraw[fill=gray!50!white, draw=black] (9.5, 1) circle [radius=.05];
		\filldraw[fill=gray!50!white, draw=black] (9.5, 1.5) circle [radius=.05];
		\filldraw[fill=gray!50!white, draw=black] (9.5, 2) circle [radius=.05];

		\end{tikzpicture}
		
	\end{center}
	
	\caption{\label{fgpaths} To the left is a path that would be counted by $B_{\mu / \nu}$, with $\mu = (4, 4, 2, 1, 0, 0, 0)$ and $\nu = (3, 1, 0)$; to the right is a path that would be counted by $D_{\mu / \nu}$, with $\mu = (4, 4, 3, 2, 2)$ and $\nu = (3, 1, 1, 1, 0)$.  } 
	
\end{figure}
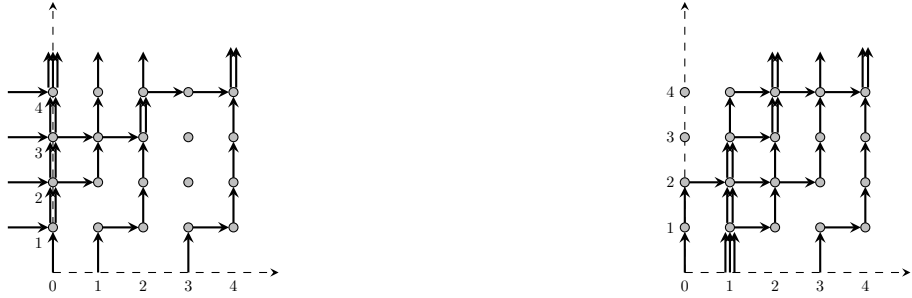

Fix parameters $\eta, \tau \in \mathbb{C}$, and recall the definition of the elliptic function $f (z) = \theta (z)$ from \eqref{theta1}. For any $v, \lambda, \Lambda \in \mathbb{C}$, define the \textit{vertex weights} $W(i_1, j_1; i_2, j_2) = W \big( i_1, j_2; i_2, j_2 \b| v, \lambda \big) = W_1 \big( i_1, j_2; i_2, j_2 \b| v, \lambda \big) = W_1 \big( i_1, j_2; i_2, j_2 \b| v, \lambda, \Lambda \big) $	 by 
\begin{flalign}
\label{weightsj1}
\begin{aligned}
W (k, 0; k, 0) &= \displaystyle\frac{f\big( \eta (\Lambda - 2k) - v \big) f (\lambda + 2 k \eta)}{f \big( \eta \Lambda - v \big) f (\lambda)}, \\
W (k, 1; k + 1, 0) &= \displaystyle\frac{f\big( v + \lambda + \eta (2k + 2 - \Lambda) \big) f ( 2 \eta)}{f \big(  \eta \Lambda - v \big) f (\lambda)}, \\
W (k, 0; k - 1, 1) &= \displaystyle\frac{f\big( \lambda - v \eta (2k - 2 - \Lambda) \big) f \big( 2 \eta (\Lambda + 1 - k) \big) f (2 k \eta) }{f \big( \eta \Lambda - v \big) f (\lambda) f (2 \eta) }, \\
W (k, 1; k, 1) &= \displaystyle\frac{f\big( \eta (2k - \Lambda) - v \big) f \big( \lambda + 2 \eta (k - \Lambda) \big)}{f \big( \eta \Lambda - v \big) f (\lambda)}, 
\end{aligned}
\end{flalign}

\noindent for each nonnegative integer $k$, and $W (i_1, j_1; i_2, j_2) = 0$ if $(i_1, j_1; i_2, j_2)$ is not of the above form. 

Now, further fix parameters $\lambda \in \mathbb{C}$, $W = (w_1, w_2, \ldots ) \subset \mathbb{C}$; $Z = (z_0, z_1, \ldots ) \subset \mathbb{C} $; and $L = (\Lambda_0, \Lambda_1, \ldots) \subset \mathbb{C}$. Also set $v_{i, j} = w_j - z_i - \eta$ for each $i, j$. 

Let $\mathcal{E}$ be a directed path ensemble, all of whose vertices are contained in some region $[-1, \infty) \times [0, N + 1]$, for some positive integer $N$. Associate an additive dynamical parameter $\Phi_{x, y}$ to each vertex in $\mathcal{E}$ through the additive rules listed \eqref{lambdaadditive}, where $(J_1, J_2, \ldots )$ from that explanation is set to $(1, 1, \ldots )$ and all parameters are shifted to the left one coordinate (so in particular $\Phi_{0, 0} = \lambda$). 

Define an \textit{interior vertex} of $\mathcal{E}$ to be a vertex of $\mathcal{E}$ contained in the strip $[0, \infty) \times [1, N] \subset \mathbb{Z}_{\ge 0} \times \mathbb{Z}_{> 0}$; these are the vertices of $\mathcal{E}$ satisfying arrow conservation. Set the \textit{weight of an interior vertex} $(x, y)$ of $\mathcal{E}$ to be $W_1 (i_1, j_1; i_2, j_2 \b| v, \Phi, \Lambda)$, where $(i_1, j_1; i_2, j_2)$ is the arrow configuration at $(x, y)$; $v = v_{x, y}$; $\Phi = \Phi_{x, y}$ is the dynamical parameter at $(x, y)$; and $\Lambda = \Lambda_x$. 

Next define the \emph{weight} of $\mathcal{E}$ to be the product of the weights of the interior vertices of $\mathcal{E}$. Assuming that all paths of $\mathcal{E}$ are finite, this product is well defined since all but finitely many vertices in $[0, \infty) \times [0, N]$ have vertex type $(0, 0; 0, 0)$, and $W (0, 0; 0, 0) = 1$ at any $(x, y) \in \mathbb{Z}_{\ge 0} \times \mathbb{Z}_{> 0}$. 

Now, using these weights we can define the elliptic weight functions $B_{\mu / \nu }$ and $D_{\mu / \nu}$, which were given in Section 7 (or Section 3) of \cite{ERSF}. We will provide an alternative definition (which was actually the original one from \cite{ERSF}) in terms of the elliptic quantum group in Section \ref{EllipticGroupSymmetricFunctions}.

\begin{definition}
	\label{definitionb}
	Let $M, N \ge 0$ be integers and $\mu \in \Sign_{M + N}^+$ and $\nu \in \Sign_N^+$ be signatures. Then, the function $B_{\mu / \nu} \big( w_1, w_2, \ldots , w_M \b| \lambda - 2 \eta M \big) = B_{\mu / \nu} \big( W \b| \lambda - 2 \eta M,Z, L \big)$ denotes the sum of the weights of all unfused directed path ensembles, consisting of $M + N$ paths whose interior vertices are contained in the rectangle $[0, \mu_1] \times [1, M]$, satisfying the following two properties.
	
	\begin{itemize}
		\item{Every path contains one edge that either connects $(-1, k)$ to $(0, k)$ for some $k \in [1, M]$ or connects $(\nu_j, 0)$ to $(\nu_j, 1)$ for some $j \in [1, N]$; the latter part of this statement holds with multiplicity, meaning that $m_i (\nu)$ paths connect $(i, 0)$ to $(i, 1)$ for each $i \in \nu$.}
		\item{ Every path contains an edge connecting $(\mu_k, M)$ to $(\mu_k, M + 1)$, for some $k \in [1, M + N]$; again, this holds with multiplicity.}
		
	\end{itemize} 
\end{definition}

\begin{definition}
	\label{definitiond} 
	Let $M, N \ge 0$ be integers and $\mu, \nu \in \Sign_M^+$ be signatures. Then, the function $D_{\mu / \nu} \big( w_1, w_2, \ldots , w_N \b| \lambda + 2 \eta N \big) = D_{\mu / \nu} \big( W \b| \lambda + 2 \eta N, Z, L \big)$ denotes the sum of the weights of all unfused directed path ensembles, consisting of $M$ paths whose interior vertices are contained in the rectangle $[0, \mu_1] \times [1, N]$, satisfying the following two properties (both with multiplicity). 
	
	\begin{itemize}
		\item{Each path contains an edge connecting $(\nu_k, 0)$ to $(\nu_k, 1)$ for some $k \in [1, M]$.}
		
		\item{Each path contains an edge connecting $(\mu_k, N)$ to $(\mu_k, N + 1)$ for some $k \in [1, M]$.}
	\end{itemize}
	
\end{definition}

Our choice to use $\lambda - 2 \eta M$ in Definition \ref{definitionb} was made in order to make our presentation consistent with that in \cite{ERSF}. In particular, since $\Phi_{0, 0} = \lambda$, the dynamical parameter at $(0, M)$ is $\lambda - 2 \eta M$, as in the definition of $B$. Analogously, our choice to use $\lambda + 2 \eta N$ in Definition \ref{definitiond} corresponds to the dynamical parameter at $(0, N)$ being $\lambda + 2 \eta N$. 

We further define $B_{\mu}$ to be the function $B_{\mu / \nu}$ when $\nu$ is empty; similarly, we define $D_{\mu} = D_{\mu / \nu}$, where $\nu = (0, 0, \ldots , 0)$ is the signature with $\ell (\mu)$ zeroes. Examples of the types of paths counted by $B_{\mu / \nu}$ and $D_{\mu / \nu}$ are depicted in Figure \ref{fgpaths}. 

It was shown in Section 3 of \cite{ERSF} that the $B$ and $D$ functions are symmetric, that is, that they are invariant under permutations of the $\{ w_j \}$. Corollary 4.3 of \cite{ERSF} states that these functions also satisfy the following branching, Cauchy, and Pieri type identities, some of which we will encounter later. In what follows, we denote $\Lambda_{[i, j]} = \sum_{k = i}^j \Lambda_k$ for any $j \ge i \ge 0$ and define the \textit{normalized} $D$-symmetric functions 
\begin{flalign*}
D_{\mu / \nu}^{(\text{n})} \big( w_1, w_2, \ldots , w_N \b| \lambda \big) = D_{\mu / \nu} \big( w_1, w_2, \ldots , w_N \b| \lambda \big) \displaystyle\prod_{k = 0}^{N - 1} f (\lambda + 2 \eta k),
\end{flalign*}

\noindent for any $\mu, \nu \in \Sign_M^+$ and $M, N \in \mathbb{Z}_{> 0}$. 

\begin{prop}[{\cite[Corollary 4.3]{ERSF}}]
	
	\label{bdidentities}  	
	
	Let $\lambda \in \mathbb{C}$, $W = (w_1, w_2, \ldots ) \subset \mathbb{C}$, $Z = (z_0, z_1, \ldots ) \subset \mathbb{C}$; $U = (u_1, u_2, \ldots) \subset \mathbb{C}$; and $L = (\Lambda_0, \Lambda_1, \ldots ) \subset \mathbb{C}$. Further let $m \ge 1$ and $N \ge 0$ be integers. For any $\nu \in \Sign_N$ and $\mu \in \Sign_{N + m + 1}$, we have the \textit{branching identity}
	\begin{flalign}
	\label{branching}
	B_{\mu / \nu} \big( w_1, w_2, \ldots, w_{m + 1} \b| \lambda \big) = \displaystyle\sum_{\kappa \in \Sign_{N + m}} B_{\mu / \kappa} \big( w_1, w_2, \ldots , w_m \b| \lambda \big) B_{\kappa / \nu} \big( w_{m + 1} \b| \lambda + 2 \eta m \big). 
	\end{flalign}
	
	\noindent Now, assume for each $n \ge 0$ that 
	\begin{flalign*}
	\displaystyle\lim_{R \rightarrow \infty} \displaystyle\frac{f(w_i - u_j + \lambda - 2 \eta \big( \Lambda_{[0, R]} - 2n \big)}{f \big( \lambda - 2 \eta (\Lambda_{[0, R]} - 2n) \big)} \displaystyle\prod_{k = 0}^R \displaystyle\frac{f \big( z_k - w_i + \eta (1 - \Lambda_k) \big) f \big( z_k - u_j + \eta (\Lambda_k + 1) \big) }{f \big( z_k - u_i + \eta (\Lambda_k + 1) \big) f \big( z_k - u_j + \eta (1 - \Lambda_k) \big) } = 0. 
	\end{flalign*}
	
	\noindent Then, for any positive integers $k$ and $r$, we have the \textit{Pieri identity} 
	\begin{flalign}
	\label{stochasticbranching}
	\begin{aligned}
	& \displaystyle\sum_{\kappa \in \Sign_{N + 1}^+} D_{\kappa} \big( u_1, u_2, \ldots , u_k \b| \lambda \big) B_{\kappa / \nu} \big( w_1 \b| \lambda  + 2 \eta k \big) \\
	& \qquad = \left( \displaystyle\frac{f \big( \lambda + w_1 + \eta (\Lambda_0 - 1 - 2N) \big) }{f \big( z - w_1 + \eta (\Lambda_0 + 1) \big) } \displaystyle\frac{f(2 \eta)}{f (\lambda)} \displaystyle\prod_{j = 1}^k \displaystyle\frac{u_j - w_1 - 2 \eta}{u_j - w_1} \right)  D_{\nu} \big( u_1, u_2, \ldots , u_k \b| \lambda + 2 \eta\big), 
	\end{aligned}]
	\end{flalign}
	
	\noindent and the \textit{Cauchy identity}
	\begin{flalign}
	\label{stochasticsum}
	\begin{aligned}
	& \displaystyle\sum_{\kappa \in \Sign_M^+} D_{\kappa}^{(\text{\emph{n}})} \big( u_1, u_2, \ldots , u_k \b| \lambda \big) B_{\kappa} \big( w_1, w_2, \ldots , w_r \b| \lambda  + 2 \eta k \big) \\
	& \qquad = f(2 \eta)^r \displaystyle\prod_{j = 1}^r \displaystyle\frac{f \big( \lambda - z_0 + w_j + \eta (2r - 1 - \Lambda_0) \big)}{f \big( \lambda + 2 \eta (j - 1) \big) f \big( z_0 - w_j + \eta (\Lambda_0 + 1) \big)} \displaystyle\prod_{i = 1}^k \displaystyle\prod_{j = 1}^r \displaystyle\frac{f(u_i - w_j - 2 \eta)}{f(u_i - w_j)} . 
	\end{aligned}
	\end{flalign}                                                                                                                                                    
\end{prop}

Let us conclude this section by remarking that a seemingly similar setup appeared in the work of Schlosser \cite{EELNP}, who considered partition functions of vertex models on directed ensembles consisting of non-intersecting paths, under some special elliptic vertex weights. In some cases, he showed that these partition functions factor, leading to combinatorial derivations of (multivariate generalizations of) the Frenkel-Turaev elliptic Jackson identity \eqref{hypergeometric109sumterminating}. We are unsure if those identities can be adapted into the setting outlined above.

\subsection{Interpretation Through the Elliptic Quantum Group \texorpdfstring{$E_{\tau, \eta} (\mathfrak{sl}_2)$}{}}

\label{EllipticGroupSymmetricFunctions}

The algebraic framework underlying the proof of Proposition \ref{bdidentities} is the representation theory of the elliptic quantum group $E_{\tau, \eta} (\mathfrak{sl}_2)$. Although we will not provide a proof of this proposition here (one can be found in Section 4 of \cite{ERSF}), we will provide a very brief description of this elliptic quantum group and one of its infinite-dimensional representations given by the evaluation Verma module. We then explain how the $B$-symmetric functions (from Definition \ref{definitionb}) are expressible as matrix elements of operators in the elliptic quantum group. A similar interpretation holds for the $D$-symmetric functions (from Definition \ref{definitiond}) as well; however, this is a bit more intricate and we will moreover not need it, so we do not discuss it here. For more information, we refer to Section 3 of \cite{ERSF}. 

The \textit{elliptic quantum group} $E_{\tau, \eta} (\mathfrak{sl}_2)$ was introduced by Felder \cite{ISAEP} and is generated by four families of operators, $a (\lambda, w)$, $b (\lambda, w)$, $c (\lambda, w)$, and $d (\lambda, w)$ (where $w$ and $\lambda$ range over $\mathbb{C}$), which are subject to 16 families of relations. These relations correspond to commutation identities, often called the \emph{dynamical Yang-Baxter equation}, between certain explicit families of $4 \times 4$ matrices. We will not need all of these explicit relations here, but they can be found in equation (3) of \cite{REQG}. In fact, we only use three consequences of these 16 identities, given by  
	\begin{flalign}
\label{caacdbbd}
\begin{aligned}
c (\lambda, w + 2 \eta) a (\lambda - 2 \eta, w) & = a (\lambda, w + 2 \eta) c (\lambda - 2 \eta, w); \\
d (\lambda, w + 2 \eta) b (\lambda + 2 \eta, w) & = b (\lambda, w + 2 \eta) d (\lambda + 2 \eta, w), 
\end{aligned}
\end{flalign}

\noindent and 
\begin{flalign}
\label{dacb}
\begin{aligned}
d (\lambda, w + 2 \eta) a (\lambda + 2 \eta, w) + & c (\lambda, w + 2 \eta) b (\lambda - 2 \eta, w) \\
& = a (\lambda, w + 2 \eta) d (\lambda - 2 \eta, w) + b (\lambda, w + 2 \eta) c (\lambda + 2 \eta, w), 
\end{aligned}
\end{flalign}

\noindent which appear as the first, second, and fifth identities in equation (4) in \cite{REQG}.

Some representations of $E_{\tau, \eta} (\mathfrak{sl}_2)$ can be defined in terms of the theta functions \eqref{theta1}. The one of interest to us will be the \emph{evaluation Verma module} $V = V_{\Lambda} = V_{\Lambda; z}$ (dependent on two complex parameters $\Lambda, z \in \mathbb{C}$), which is spanned by a semi-infinite collection of vectors $\{ e_0, e_1, \ldots \}$. 

Denoting $v = w - z - \eta$, the generators of $E_{\tau, \eta} (\mathfrak{sl}_2)$ act on this basis through 
\begin{flalign*}
a (\lambda, w) e_k &= W(k, 0; k, 0) e_k; \qquad \qquad \quad b (\lambda, w) e_k = W(k, 1; k + 1, 0) e_{k + 	1}, \\
c (\lambda, w) e_k &= W (k, 0; k - 1, 1) e_{k - 1}; \qquad d (\lambda, w) e_k = W (k, 1; k, 1) e_k, 
\end{flalign*}

\noindent where the vertex weights $W(i_1, j_1; i_2, j_2) = W_1 \big( i_1, j_2; i_2, j_2 \b| v, \lambda, \Lambda \big) $ are given by \eqref{weightsj1}. 

It will also be useful to define variants $\widetilde{a}$, $\widetilde{b}$, $\widetilde{c}$, and $\widetilde{d}$ of these operators, which act on functions $G: \mathbb{C} \rightarrow V$ through the identities 
\begin{flalign}
\label{modifiedabcd}
\begin{aligned}
\widetilde{a} (w) G (\lambda) & = a (\lambda, w) G (\lambda - 2 \eta); \qquad \widetilde{b} (w) G (\lambda) = b (\lambda, w) G (\lambda + 2 \eta); \\
\widetilde{c} (w) G (\lambda) & = c (\lambda, w) G (\lambda - 2 \eta); \qquad \widetilde{d} (w) G (\lambda) = d (\lambda, w) G (\lambda + 2 \eta). 
\end{aligned} 
\end{flalign}

Now let us explain how $E_{\tau, \eta} (\mathfrak{sl}_2)$ acts on tensor products of evaluation Verma modules. To that end, let $V = V_{\Lambda}$ be an evaluation Verma module, and let $\mathsf{h}$ denote its \textit{weight}; specifically, for any function $F: \mathbb{C} \rightarrow \mathbb{C}$ and integer $k \ge 0$, we set $F ( \mathsf{h}) e_k = F (\Lambda - 2k) e_k$.

As explained in Section 2 of \cite{REQG}, the tensor product $V_1 \otimes V_2$ of two $E_{\tau, \eta} (\mathfrak{sl}_2)$-modules $V_1$ and $V_2$ is again an $E_{\tau, \eta} (\mathfrak{sl}_2)$-module in which $a$, $b$, $c$, and $d$ act on $v_1 \otimes v_2$ (for any $v_1 \in V_1$ and $v_2 \in V_2$) by 
\begin{flalign}
\label{ellipticv1v2}
\left[ \begin{matrix} a (\lambda, w) & b (\lambda, w) \\ c (\lambda, w) & d (\lambda, w) \end{matrix}\right] v_1 \otimes v_2 = \left[ \begin{matrix} a_2 \big( \lambda - 2 \eta \mathsf{h}^{(1)}, w \big) & b_2 \big( \lambda - 2 \eta \mathsf{h}^{(1)}, w \big) \\ c_2 \big( \lambda - 2 \eta \mathsf{h}^{(1)}, w \big) & d_2 \big( \lambda - 2 \eta \mathsf	{h}^{(1)}, w \big) \end{matrix}\right] \left[ \begin{matrix} a_1 (\lambda, w) & b_1 (\lambda, w) \\ c_1 (\lambda, w) & d_1 (\lambda, w) \end{matrix}\right] v_1 \otimes  v_2. 
\end{flalign}

 Above, $a_1, b_1, c_1, d_1$ and $a_2, b_2, c_2, d_2$ denote the actions of $a, b, c, d$ on the $V_1$ and $V_2$ components of $V_1 \otimes V_2$, respectively. Furthermore, $\mathsf{h}^{(1)}$ denotes the action of $\mathsf{h}$ on the $V_1$-component of $V_1 \otimes V_2$ after the action of $\left[\begin{smallmatrix} a_1 & b_1 \\ c_1 & d_1 \end{smallmatrix} \right]$ on $V_1$. The weight $\mathsf{h}$ of $V_1 \otimes V_2$ is equal to the sum $\mathsf{h}^{(1)} + \mathsf{h}^{(2)}$ of the weights of $V_1$ and $V_2$, meaning that $F (\mathsf{h}) v = F \big( \mathsf{h}^{(1)} + \mathsf{h}^{(2)} \big) v$ for any $v \in V_1 \otimes V_2$. 

In Section 3 of \cite{ERSF}, the above setup was used to provide the original definition for the $B_{\mu / \nu}$-symmetric functions. To explain that viewpoint further, fix $M \in \mathbb{Z}_{> 0}$ and $W = (w_1, w_2, \ldots , w_M ) \subset \mathbb{C}$, $Z = (z_0, z_1, \ldots ) \subset \mathbb{C}$, and $L = (\Lambda_0, \Lambda_1, \ldots ) \subset \mathbb{C}$. For each $k \in \mathbb{Z}_{\ge 0}$, let $V_k = V_{\Lambda_0; z_0} \otimes V_{\Lambda_1; z_1} \otimes \cdots \otimes V_{\Lambda_k; z_k}$; observe that $V_k$ maps into $V_{k + 1}$ under the injection sending $v_k \in V_k$ to $v_k \otimes e_0 \in V_{k + 1}$. 

Now let $V = \varinjlim V_k$ denote the direct limit of the $\{ V_k \}$; then $v \in V$ if and only if it is \textit{finitary}, that is, expressible in the form $v_r \otimes e_0 \otimes e_0 \otimes \cdots $ for some integer $r \ge 0$ and $v_r \in V_r$. There exists a basis for $V$ is indexed by signatures, namely $\{ e_{\mu} \}_{\mu \in \Sign^+}$, where $e_{\mu} = e_{m_0 (\mu)} \otimes e_{m_1 (\mu)} \otimes \cdots \in V$. In view of \eqref{ellipticv1v2} and the fact that $a (\lambda, w) e_0 = e_0$ in any evaluation Verma module, the action of $b (\lambda, w) \in E_{\tau, \eta} (\mathfrak{sl}_2)$ on $V$ is well-defined. 

Using this notation, one can verify (for instance, see Section 7 of \cite{ERSF}, or Section 4 of \cite{HSVMRSF} for simpler quantum analog) that Definition \ref{definitionb} is equivalent to the $B_{\mu / \nu}$-functions satisfying
\begin{flalign}
\label{bsumv}
\left( \displaystyle\prod_{j = 1}^M \widetilde{b}_j (\lambda, w_j) \right) e_{\nu} = \displaystyle\sum_{\mu \in \Sign_{M + N}} B_{\mu / \nu} \big( w_1, w_2, \ldots , w_M \b| Z, L, \lambda \big) e_{\mu}, 
\end{flalign}

\noindent for any $\nu \in \Sign_N^+$; the product on the right side of the above equality is well-defined since the $\widetilde{b} (\lambda, w_j)$ operators commute (see the fourth identity listed in equation (3) of \cite{REQG}). 

One can define the $D$-symmetric functions (from Definition \ref{definitiond}) in a similar way, by replacing the appearances of $\widetilde{b} (\lambda, w)$ in \eqref{bsumv} with $\widetilde{d} (\lambda, w)$. However, since $d (\lambda, w) e_0 \ne e_0$, this requires a pre-normalization of the $\widetilde{d}$-operators; we refer to Section 3 of \cite{ERSF} for a more detailed explanation.

\subsection{Fusion in the Elliptic Quantum Group} 

\label{Fusion}

Having established this notation, let us discuss fusion. The fusion procedure was originally introduced by Kulish-Reshetikhin-Sklyanin \cite{ERT} as a way of producing higher dimensional representations of the quantum group $U_q (\widehat{\mathfrak{sl}}_2)$ from lower dimensional representations. In the context of the elliptic quantum group $E_{\tau, \eta} (\mathfrak{sl}_2)$, fusion of evaluation Verma modules was considered first in Section 8 of the work of Felder-Varchenko \cite{REQG}. Fusion was also explained from a statistical mechanical viewpoint \cite{ESSOSM, REA} as a way of producing new quantum (or elliptic) integrable systems from more restrictive ones. In this section we review the latter framework in the case of $E_{\tau, \eta} (\mathfrak{sl}_2)$, with the intent of introducing the fused weights $W_J (i_1, j_1; i_2, j_2)$ given by Definition \ref{fusedweightdefinition}. 

Fusion in $E_{\tau, \eta} (\mathfrak{sl}_2)$ can be expressed in terms of repeated applications of the operators $\widetilde{a}$, $\widetilde{b}$, $\widetilde{c}$, and $\widetilde{d}$ on (tensor products of) $V$. Applying products of these operators to a single vector $e_k \in V$ produces products of weights, which are given by the following definition.

\begin{definition}
	
	\label{vertexmodelsinglerow}
	
	Fix $J \in \mathbb{Z}_{> 0}$. Let $i_1, i_2 \in \mathbb{Z}_{\ge 0}$; $\mathcal{J}_1 = (j_{1, 1}, j_{1, 2}, \ldots , j_{1, J}) \in \{ 0, 1 \}^J$, and $\mathcal{J}_2 = (j_{2, 1}, j_{2, 2}, \ldots , j_{2, J}) \in \{ 0, 1 \}^J$. Denote $|\mathcal{J}_1|$ and $|\mathcal{J}_2|$ by the number of ones in $\mathcal{J}_1$ and $\mathcal{J}_2$, respectively. 
	
	Set $W_J (i_1, \mathcal{J}_1; i_2, \mathcal{J}_2 \b| v, \lambda ) = 0$ if $i_1 + |\mathcal{J}_1| \ne i_2 + |\mathcal{J}_2|$. Otherwise set
	\begin{flalign}
	\label{definitionhatww}
	\begin{aligned}
	W_J (i_1, \mathcal{J}_1; i_2, \mathcal{J}_2) & = \Bigg\langle X_J X_{J - 1} \cdots X_1 (\lambda) e_k, e_k \Bigg\rangle, 
	\end{aligned}
	\end{flalign}
	
	\noindent where the $X_k$ are defined as follows. If $(j_{1, k}, j_{2, k}) = (0, 0)$, then $X_k = \widetilde{a} \big( w + 2 \eta (k - 1), \lambda \big)$; if $(j_{1, k}, j_{2, k}) = (1, 0)$, then $X_k = \widetilde{b} \big( w + 2 \eta (k - 1), \lambda \big)$; if $(j_{1, k}, j_{2, k}) = (0, 1)$, then $X_k = \widetilde{c} \big( w + 2 \eta (k - 1), \lambda \big)$; and if $(j_{1, k}, j_{2, k}) = (1, 1)$, then $X_k = \widetilde{d} \big( w + 2 \eta (k - 1), \lambda \big)$. 
	
	Further define 
	\begin{flalign}
	\label{wjwsums}
	W_J (i_1, j_1; i_2, \mathcal{K}) = \displaystyle\sum_{|\mathcal{J}_1| = j_1} W_J (i_1, \mathcal{J}_1; i_2, \mathcal{K}); \quad \widehat{W}_J (i_1, j_1; i_2, j_2) = \displaystyle\sum_{|\mathcal{J}_2| = j_2} W_J (i_1, j_1; i_2, \mathcal{J}_2), 
	\end{flalign}
	
	\noindent for any nonnegative $i_1$, $j_1$, $i_2$, $j_2$, and $\mathcal{K} \in \{ 0, 1 \}^J$. 
\end{definition}

Let us define an alternative, more pictorial interpretation of the weights $W_J$ and $\widehat{W}_J$ above through vertex models; this will be quite similar to the vertex model interpretation of \eqref{bsumv} given by Definition \ref{definitionb}. To that end, fix complex numbers $w, z, \Lambda, \lambda \in \mathbb{C}$, and consider a column $\{ 1 \} \times \{1, 2, \ldots,  J \} \subset \mathbb{Z}_{> 0}^2$. We associate the complex numbers $z$ and $\Lambda$ with the column, and we associate the \textit{spectral parameter} $w_k = w + 2 \eta (k - 1)$ with the $k$th row of this column for each integer $k \in [1, J]$. We furthermore denote $v = w - z - \eta$ and $v_k = w_k - z - \eta$ for each $k$.

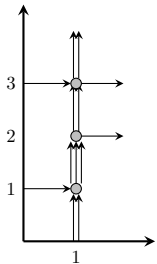
\begin{figure}
	
	\begin{center} 
		
		\begin{tikzpicture}[
		>=stealth,
		scale = .7
		]
		
		\draw[->, black, thick] (-6, -2) -- (-3.5, -2);
		\draw[->, black, thick] (-6, -2) -- (-6, 2.5);
		
		\draw[->, black] (-6, -1) -- (-5.1, -1);
		\draw[->, black] (-6, 1) -- (-5.1, 1);

		\draw[->, black] (-4.95, -2) -- (-4.95, -1.1);
		\draw[->, black] (-5.05, -2) -- (-5.05, -1.1);
		
		\draw[->, black] (-4.9, -.9) -- (-4.9, -.1);
		\draw[->, black] (-5, -.9) -- (-5, -.1);
		\draw[->, black] (-5.1, -.9) -- (-5.1, -.1);
		\draw[->, black] (-5.05, .1) -- (-5.05, 1);
		\draw[->, black	] (-4.95, -.1) -- (-4.95, 1);
		\draw[->, black] (-5, 0) -- (-4.1, 0);
		\draw[->, black] (-5, 1) -- (-4.1, 1);
		\draw[->, black] (-4.95, 1.1) -- (-4.95, 2);
		\draw[->, black] (-5.05, 1.1) -- (-5.05, 2);
		
		\filldraw[fill=gray!50!white, draw=black] (-5, -1) circle [radius=.1] node[scale = .7, left = 29]{$1$} node[scale = .7, below = 30	]{$1$};
		\filldraw[fill=gray!50!white, draw=black] (-5, 0) circle [radius=.1] node[scale = .7, left = 29]{$2$}; 
		\filldraw[fill=gray!50!white, draw=black] (-5, 1) circle [radius=.1]  node[scale = .7, left = 29]{$3$};

		\end{tikzpicture}
		
	\end{center}

	\caption{\label{arrows2} To the left is an example of a column with three vertices. Here, $i_1 = 2$; $\mathcal{J}_1 = (1, 0, 1)$; $i_2 = 2$; and $\mathcal{J}_2 = (0, 1, 1)$. } 
\end{figure}

Let us consider an unfused path ensemble on this column. Specifically, through each vertex $(1, y)$ of this column will be some number of incoming and outgoing arrows; see Figure \ref{arrows2} for an example when $J = 3$. Denote the arrow configuration at any vertex $(1, y)$ by $(i_1, j_1; i_2, j_2) = \big( i_1 (1, y), j_1 (1, y); i_2 (1, y), j_2 (1, y) \big)$. Since the ensemble is unfused, we must have that $j_1, j_2 \in \{ 0, 1 \}$ at all vertices $(1, y)$.

Associated with each vertex $(1, y)$ will also be an additive dynamical parameter $\Phi_{1, y} \in \mathbb{C}$, defined as in Section \ref{PathEnsembles} with $(J_1, J_2, \ldots ) = (1, 1, \ldots )$ and the difference that we set the topmost dynamical parameter $\Phi_{1, J}$ equal to $\lambda$ (instead of $\Phi_{1, 0} = \lambda$). Specifically, one first sets $\Phi_{1, J} = \lambda$. Then, if $\Phi_{1, y + 1}$ is defined for some $y \in [1, J - 1]$, one defines $\Phi_{1, y}$ as follows. If $j_1 (1, y + 1) = 0$ then $\Phi_{1, y} = \Phi_{1, y + 1} - 2 \eta$, and if $j_1 (1, y + 1) = 1$ then $\Phi_{1, y} = \Phi_{1, y + 1} + 2 \eta$; see Figure \ref{lambdachange} for an example. 

Now we can associate with each vertex $(1, y)$ the vertex weight $W_1 \big( i_1, j_1; i_2, j_2 \b| v_k, \Phi, \Lambda \big)$, where $i_1$, $j_1$, $i_2$, $j_2$, and $\Phi$ are all taken at $(1, y)$. We define the \textit{weight of the column} $\{ 1 \} \times \{1, 2, \ldots , J \}$ to be the product of the $J$ weights associated with its vertices. 

Then, following the notation from Definition \ref{vertexmodelsinglerow}, if $\mathcal{J}_1 = \big( j_1 (1, 1), j_1 (1, 2), \ldots , j_1 (1, J) \big)$; $\mathcal{J}_2 = \big( j_2 (1, 1), j_2 (1, 2), \ldots , j_2 (1, J) \big)$; $i_1  (1, 1) = i_1$; and $i_2 (1, J) = i_2$, then the weight of the column $\{1 \} \times \{1, 2, \ldots , J \}$ will be equal to $W_J (i_1, \mathcal{J}_1; i_2, \mathcal{J}_2)$. Moreover, $W_J (i_1, j_1; i_2, \mathcal{K})$ is equal to the sum of the weights of all $J$-element columns satisfying $\mathcal{K} = \big( j_2 (1, 1), j_2 (1, 2), \ldots , j_2 (1, J) \big)$; $i_1  (1, 1) = i_1$; and $i_2 (1, J) = i_2$. Similarly, $\widehat{W}_J (i_1, j_1; i_2, j_2)$ is equal to the sum of the weights of all $J$-element columns that have $i_1$ incoming vertical arrows, $j_1$ incoming horizontal arrows, $i_2$ outgoing vertical arrows, and $j_2$ outgoing horizontal arrows; Figure \ref{arrows2} provides an example of such a column with $J = 3$ and $(i_1, j_1; i_2, j_2) = (2, 2; 2, 2)$.

The following lemma, which is very similar to the first part of Lemma 2.1.1 of \cite{ESSOSM}, states that the vertex weights $W_J (i_1, j_1; i_2, \mathcal{K})$ defined above do not depend on the explicit choice of $\mathcal{K}$, as long as $|\mathcal{K}|$ remains fixed.  

\begin{lem}
	
	\label{exchangeability}
	
	Fix a positive integer $J$ and nonnegative integers $i_1$, $j_1$, $i_2$, and $j_2$. Let $\mathcal{K}, \mathcal{K}' \in \{ 0, 1 \}^J$ satisfy $|\mathcal{K}| = j_2 = |\mathcal{K}'|$. Then, $W_J (i_1, j_1; i_2, \mathcal{K}) = W_J (i_1, j_1; i_2, \mathcal{K}')$. 
\end{lem}

\begin{proof} 
	
	We may first assume that $i_1 + j_1 = i_2 + j_2$; otherwise $W_J (i_1, j_1; i_2, \mathcal{K}) = 0 = W_J (i_1, j_1; i_2, \mathcal{K}')$ by arrow conservation, and the lemma holds. 
	
	Now denote $\mathcal{K} = (k_1, k_2, \ldots , k_J)$ and $\mathcal{K}' = (k_1', k_2', \ldots , k_J')$. It suffices to establish the lemma when $\mathcal{K}$ and $\mathcal{K}'$ differ by a transposition of two consecutive indices, that is, when there exists an integer $s \in [1, J - 1]$ such that $k_i = k_i'$ when $i \notin \{ s, s + 1 \}$, but such that $(k_s, k_{s + 1}) = (k_{s + 1}', k_s')$. Therefore, we also assume that this is the case. Moreover, since $\mathcal{K}$ and $\mathcal{K}'$ only differ in two consecutive indices, we can further assume $J = 2$.
	
	If either $(k_1, k_2) = (0, 0)$ or $(k_1, k_2) = (1, 1)$ held, then we must have that $\mathcal{K}' = (k_1', k_2') = (k_1, k_2) = \mathcal{K}$ since $|\mathcal{K}| = |\mathcal{K}'|$, from which the lemma would follow. Thus, we can additionally assume $\mathcal{K} = (k_1, k_2) = (0, 1)$ and $\mathcal{K}' = (k_1', k_2') = (1, 0)$. 
	
	Now, there are three cases to consider, depending on the value of $j_1 \in \{ 0, 1, 2 \}$. Each case follows from one of the equations \eqref{caacdbbd} or \eqref{dacb} listed above. For instance, when $j_1 = 1$, the left side of \eqref{dacb} can be equated with $W_2 (i_1, j_1; i_2, \mathcal{K})$, and the right side can be equated with $W_2 (i_1, j_1; i_2, \mathcal{K}')$. Hence, the case $j_1 = 1$ follows from \eqref{dacb}; similarly, the remaining cases $j_1 = 0$ and $j_1 = 2$ follow from the first and second parts of \eqref{caacdbbd}, respectively.  
\end{proof} 

\begin{figure}
	
	\begin{center}
		
		\begin{tikzpicture}[
		>=stealth,
		scale = .8
		]

		\draw[-> ] (-6, -5) -- (-5.1, -5);

		\draw[->] (-5.05, -3.9) -- (-5.05, -3);
		\draw[->] (-4.95, -3.9) -- (-4.95, -3);
		\draw[->] (-5.05, -4.9) -- (-5.05, -4.1);
		\draw[->] (-4.95, -4.9) -- (-4.95, -4.1);
		\draw[->] (-5, -6) -- (-5, -5.1);
		
		\filldraw[fill=gray!50!white, draw=black] (-5, -5) circle [radius=.1] node[right=2, scale=.8]{$\Phi - 2 \eta$};
		\filldraw[fill=gray!50!white, draw=black] (-5, -4) circle [radius=.1] node[right=2, scale=.8]{$\Phi$};

		\draw[-> ] (1, -5) -- (1.9, -5);
		\draw[-> ] (2.1, -5) -- (3, -5);
		\draw[->, black ] (1, -4) -- (1.9, -4);

		\draw[->] (2, -6) -- (2, -5.1);
		\draw[->] (2, -4.9) -- (2, -4.1);
		\draw[->] (1.95, -3.9) -- (1.95, -3);
		\draw[->] (2.05, -3.9) -- (2.05, -3);

		\filldraw[fill=gray!50!white, draw=black] (2, -5) circle [radius=.1] node[right=15, below = 3, scale=.8]{$\Phi + 2 \eta$};
		\filldraw[fill=gray!50!white, draw=black] (2, -4) circle [radius=.1] node[right=2, scale=.8]{$\Phi$};	
		
		\end{tikzpicture}
		
	\end{center}
	
	\caption{ \label{lambdachange} Some possible evolutions for the dynamical parameter $\Phi$ are shown above.  }
\end{figure}
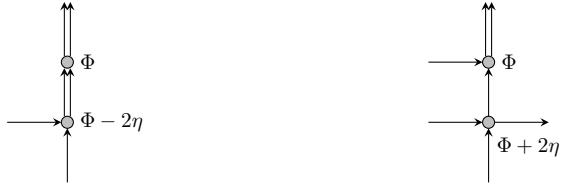

We can now define a fused weight. 

\begin{definition}
	
	\label{fusedweightdefinition}
	
	Fix $J \in \mathbb{Z}_{> 1}$ and $i_1, j_1, i_2, j_2 \in \mathbb{Z}_{\ge 0}$. Define the \emph{fused weight} $W_J (i_1, j_1; i_2, j_2) = W_J \big( i_1, j_1; i_2, j_2 \b| v, \lambda \big)  = W_J (i_1, j_1; i_2, \mathcal{J}_2)$, for any $\mathcal{J}_2 \in \{ 0, 1 \}^J$ satisfying $|\mathcal{J}_2| = j_2$; by Lemma \ref{exchangeability}, $W (i_1, j_1; i_2, j_2)$ does not depend on the choice of $\mathcal{J}_2$. 
\end{definition}

\noindent Theorem \ref{fusedweighthypergeometric} in Section \ref{FusionSymmetric} (to be established in Section \ref{FusedWeights}) explicitly evaluates the $W_J (i_1, j_1; i_2, j_2)$ fused weights as elliptic hypergeometric series. Before stating that result, however, one can quickly observe that 
\begin{flalign}
\label{wjlargesmallj}
W_J (i_1, j_1; i_2, j_2) = 0, \quad \text{if $j_1 \not\in \{ 0, 1, \ldots , J \}$, $j_2 \notin \{ 0, 1, \ldots , J \}$, or $i_1 + j_1 \ne i_2 + j_2$.}
\end{flalign}

\noindent Furthermore, \eqref{wjwsums} and Lemma \ref{exchangeability} imply that
\begin{flalign}
\label{hatwwj}
W_J (i_1, j_1; i_2, j_2) =  \binom{J}{j_2}^{-1} \widehat{W}_J (i_1, j_1; i_2, j_2). 
\end{flalign}

\subsection{Fusion for Symmetric Functions} 

\label{FusionSymmetric} 

In Section \ref{Fusion} we observed how partition functions of a vertex model on one column simplify when the spectral parameters $w_1, w_2, \ldots , w_J$ are chosen such that $w_k = w + 2 \eta (k - 1)$ for some $w \in \mathbb{C}$ and all $k \in [1, J]$. The purpose of this section is to understand what happens to the symmetric functions $B_{\mu / \nu} \big (w_1, w_2, \ldots , w_J \b| \lambda \big)$ and $D_{\mu / \nu} \big( w_1, w_2, \ldots , w_J \b| \lambda \big)$ when one specializes the $\{ w_k \}$ in the same way. 

In fact, let us consider a slight generalization. In this section we fix $r \in \mathbb{Z}_{> 0}$; $J_1, J_2, \ldots , J_r \in \mathbb{Z}_{> 0}$; $W = (w_1, w_2, \ldots , w_r) \subset \mathbb{C}$; $Z = (z_0, z_1, \ldots ) \subset \mathbb{C}$; and $L = (\Lambda_0, \Lambda_1, \ldots ) \subset \mathbb{C}$. Further define $v_{i, j} = w_j - z_i - \eta$ for each $i, j$; set $J = \sum_{k = 1}^r J_i$; and set $\varpi_k = \big( w_k, w_k + 2 \eta, \ldots , w_k + 2 \eta (J_k - 1) \big)$ for each $k \in [1, r]$. 

We would like to understand $B_{\mu / \nu} \big( \varpi_1, \varpi_2, \ldots , \varpi_r \b| \lambda \big)$ and $D_{\mu / \nu} \big( \varpi_1, \varpi_2, \ldots , \varpi_r \b| \lambda \big)$ for signatures $\mu$ and $\nu$. From Definition \ref{definitionb} and Definition \ref{definitiond}, these functions can be expressed as partition functions of a vertex model (with an additive dynamical parameter) on $J$ rows. Proposition \ref{bdfunctionsfused} below will explain how they can be alternatively expressed as partition functions of a vertex model on $r$ rows but different weights. 

To explain further, let us associate with any directed path ensemble $\mathcal{E}$ on $\mathbb{Z}_{\ge 0} \times \{ 1, 2, \ldots , r \}$ (as in Section \ref{EllipticSymmetric}, we shift all path ensembles to the left one coordinate so that, for example, $\Phi_{0, 0} = \lambda$) an additive dynamical parameter through the rules listed in \eqref{lambdaadditive}, with the $(J_1, J_2, \ldots )$ from that explanation set to the $(J_1, J_2, \ldots )$ here. We define the \textit{fused weight of an interior vertex} $(x, y) \in \mathcal{E}$ to be $W_{J_y} (i_1, j_1; i_2, j_2 \b| v_{x, y}, \Phi, \Lambda_x)$, where the arrow configuration $(i_1, j_1; i_2, j_2)$ and dynamical parameter $\Phi$ are taken at the vertex $(x, y)$. We moreover define the \textit{fused weight of the ensemble} $\mathcal{E}$ to be the product of the fused weights of the interior vertices of $\mathcal{E}$.

\begin{prop}
	\label{bdfunctionsfused}
	
	The following two statements hold. 	
	
	\begin{enumerate}
	
	\item{Let $N \in \mathbb{Z}_{\ge 0}$, $\mu \in \Sign_{N + J}^+$, and $\nu \in \Sign_N^+$. Then, $B_{\mu / \nu} \big( \varpi_1, \varpi_2, \ldots , \varpi_r \b| \lambda - 2 \eta J \big) = B_{\mu / \nu} \big( W \b| \lambda - 2 \eta J, Z, L \big)$ equals the sum of the fused weights of all (not necessarily unfused) directed path ensembles, consisting of $N + J$ paths whose interior vertices are contained in the rectangle $[0, \mu_1] \times [1, r]$, satisfying the following two properties (both with multiplicity).
	
	\begin{itemize}
		\item{Every path contains one edge that either connects $(-1, k)$ to $(0, k)$ for some $k \in [1, r]$ or connects $(\nu_j, 0)$ to $(\nu_j, 1)$ for some $j \in [1, N]$. Furthermore, there exist $J_k$ paths connecting $(-1, k)$ to $(0, k)$ for each $k \in [1, r]$. }
		\item{ Every path contains an edge connecting $(\mu_k, M)$ to $(\mu_k, M + 1)$, for some $k \in [1, N + J]$.}
		
	\end{itemize}
	
	}
	
	\item{Let $M \in \mathbb{Z}_{\ge 0}$ and $\mu, \nu \in \Sign_M^+$. Then, $D_{\mu / \nu} \big( \varpi_1, \varpi_2, \ldots , \varpi_r \b| \lambda + 2 \eta J \big) = D_{\mu / \nu} \big( W \b| \lambda + 2 \eta J, Z, L\big)$ equals the sum of the fused weights of all (not necessarily unfused) directed path ensembles, consisting of $M$ paths whose interior vertices are contained in the rectangle $[0, \mu_1] \times [1, N]$, satisfying the following two properties (both with multiplicity). 
	
	\begin{itemize}
		\item{Each path contains an edge connecting $(\nu_k, 0)$ to $(\nu_k, 1)$ for some $k \in [1, M]$.}
		
		\item{Each path contains an edge connecting $(\mu_k, N)$ to $(\mu_k, N + 1)$ for some $k \in [1, M]$.}
	\end{itemize}

	}
	
	\end{enumerate}
	
\end{prop}

In particular, Proposition \ref{bdfunctionsfused} implies the $B_{\mu / \nu}$ and $D_{\mu / \nu}$ symmetric functions can be defined as partition functions of (potentially fused) vertex models on $r$ rows, in which the vertex weight at $(x, y)$ is a $W_{J_y}$ fused weight instead of a $W_1$ weight. The proof of this result is similar to that of its quantum analog, given by Proposition 5.5 in \cite{HSVMRSF}. 

\begin{proof}[Proof of Proposition \ref{bdfunctionsfused}]
	
	We only establish the first part of the proposition, as the proof of the second part is very similar. We also only establish this in the case $r = 1$, as the general case follows from the branching identity \eqref{branching}. 
	
	To that end, for any unfused path ensemble $\mathcal{E}$ on $[0, \mu_1] \times [1, J]$, define $\mathcal{J}^{(0)} (\mathcal{E}), \mathcal{J}^{(1)} (\mathcal{E}), \ldots \in \{ 0, 1 \}^J$ as follows. For each integer $k \ge 0$, the $J$-tuple $\mathcal{J}^{(k)} = \mathcal{J}^{(k)} (\mathcal{E}) = \big( j_1^{(k)}, j_2^{(k)}, \ldots , j_J^{(k)} \big)$ is produced by setting $j_s^{(k)}$ equal to $1$ if there is an incoming horizontal arrow at the vertex $(k, s)$ and equal to $0$ otherwise. 
	
	Then, letting $m_{[i, j)} (\mu) = \sum_{k = i}^{j - 1} m_k (\mu)$ and $\Lambda_{[i, j)} = \sum_{k = i}^{j - 1} \Lambda_k$ for any $j \ge i \ge 0$, Definition \ref{definitionb} implies that $B_{\mu / \nu} \big(\varpi_1 \b| \lambda \big)$ is equal to 
	\begin{flalign}
	\label{b1sum}
	\displaystyle\sum_{\mathcal{J}^{(1)}, \mathcal{J}^{(2)}, \ldots } \displaystyle\prod_{k = 0}^{\infty} W_J \Big( m_k (\nu), \mathcal{J}^{(k)}; m_k (\mu), \mathcal{J}^{(k + 1)} \b| \lambda + 2 \eta \Lambda_{[0, k)} + 4 \eta m_{[0, k)} (\mu) \Big),  
	\end{flalign}
	
	\noindent where $\mathcal{J}^{(0)}$ is set to $(1, 1, \ldots , 1)$, the other $\mathcal{J}^{(k)}$ can be arbitrary, and we recall the definition of $W_J (i_1, \mathcal{J}_1; i_2, \mathcal{J}_2)$ from Definition \ref{vertexmodelsinglerow}. The expression \eqref{b1sum} arises from evaluating $B_{\mu / \nu}$ as a sum of weights of directed path ensembles as in Definition \ref{definitionb}; by writing the weight of each such ensemble $\mathcal{E}$ as a product of the weights of its columns; and by observing that the value of the dynamical parameter at the vertex $(k, J)$ (at the top of the $k$th column) is $\lambda + 2 \eta \Lambda_{[0, k)} + 4 \eta m_{[0, k)} (\mu)$.  	
	
	Next, we sum over $\mathcal{J}^{(1)}$, fixing all other $\mathcal{J}^{(i)}$. From Lemma \ref{exchangeability}, \eqref{b1sum} becomes 
	\begin{flalign}
	\label{b2sum}
	\begin{aligned}
	& W_J \Big( m_0 (\nu), J; m_0 (\mu), \big| \mathcal{J}^{(1)} \big| \b| \lambda \Big) \displaystyle\sum_{\mathcal{J}^{(2)}, \mathcal{J}^{(3)}, \ldots }  W_J \Big( m_1 (\nu), \big| \mathcal{J}^{(1)} \big| ; m_1 (\mu), \mathcal{J}^{(2)} \b| \lambda + 2 \eta \Lambda_0 + 4 \eta m_0 (\mu) \Big) \\
	& \qquad \times \displaystyle\prod_{k = 2}^{\infty} W_J \Big( m_k (\nu), \mathcal{J}^{(k)}; m_k (\mu), \mathcal{J}^{(k + 1)} \b| \lambda + 2 \eta \Lambda_{[0, k)} + 4 \eta m_{[0, k)} (\mu) \Big). 
	\end{aligned} 
	\end{flalign}
	
	\noindent Repeating (summing over the other $\mathcal{J}^{(i)}$ in order), we obtain that \eqref{b2sum} is equal to 
	\begin{flalign}
	\label{b3sum}
	\displaystyle\prod_{k = 0}^{\infty} W_J \Big( m_k (\nu), |\mathcal{J}^{(k)}|; m_k (\mu), |\mathcal{J}^{(k + 1)}| \b| \lambda + 2 \eta \Lambda_{[0, k)} + 4 \eta m_{[0, k)} (\mu) \Big), 
	\end{flalign}
	
	\noindent where $\big| \mathcal{J}^{(0)} \big| = J$ and the other $\big| \mathcal{J}^{(k)} \big|$ are defined inductively by $\big| \mathcal{J}^{(k)} \big| + m_k (\nu) = \big| \mathcal{J}^{(k + 1)} \big| + m_k (\mu)$. This matches with the first part of the proposition since there is only one path ensemble satisfying the two listed properties, and it can be verified that weight of this path ensemble is \eqref{b3sum}. 
\end{proof}

In view of the above proposition, it will be of interest to explicitly determine the fused $W_J$ weights. We can do this quickly if $j_1 = J$, in which case Lemma \ref{exchangeability} will be particularly useful. The special setting $j_1 = J$ will also be relevant when we discuss stochasticity later in Section \ref{PsiStochasticity}.

\begin{lem}
	
	\label{wjjj}
	
	For any $i, j \in \mathbb{Z}_{\ge 0}$ and $J \in \mathbb{Z}_{> 0}$ with $J \ge j$, we have that 
	\begin{flalign*}
	W_J \big( i, J; i + J - j, j \b| v, \lambda \big) & = f (2 \eta)^{J - j}  \displaystyle\frac{\big[ 2 \eta i - \eta \Lambda - v \big]_j }{[\eta \Lambda -v]_J } \\	
	& \qquad \times \displaystyle\frac{\big[ v + \lambda - \eta \Lambda + 2 \eta (i + 2 J - j - 1) \big]_{J - j} \big[ \lambda + 2 \eta (i + J - \Lambda - 1) \big]_j}{\big[ \lambda + 2 \eta (J - 1) \big]_J}. 
	\end{flalign*}
	
\end{lem}

\begin{proof}
	
	To evaluate $W_J \big( i, J; i + J - j, j \b| v, \lambda \big)$, we use Definition \ref{fusedweightdefinition} with $\mathcal{J}_2 = (0, 0, \ldots, 0, 1, 1, \ldots , 1)$, where the first $j$ elements of $\mathcal{J}_2$ are equal to $0$ and the last $J - j$ elements are equal to $1$. Therefore, since $j_1 = J$ implies that $\mathcal{J}_1$ (from Definition \ref{vertexmodelsinglerow}) must only consist of ones, we deduce that
	\begin{flalign*}
	W_J \big( i, J; i + J - j, j \b| v, \lambda \big) & = \displaystyle\prod_{k = 0}^{j - 1} W \big( i + J - j, 1; i + J - j, 1 \b| \lambda + 2 \eta (j - k -1), v + 2 \eta (J - j + k) \big) \\
	& \qquad \times \displaystyle\prod_{k = 0}^{J - j - 1} W \big( i + k, 1; i + k + 1, 0 \b| \lambda + 2 \eta (J - k - 1), v + 2 \eta k \big), 
	\end{flalign*}
	
	\noindent from which we deduce the explicit form of $W_J \big( i, J; i + J - j, j \b| v, \lambda \big)$ using the definition \eqref{weightsj1} of the $W_1$ weights. 
\end{proof}

More generally, we have the following theorem, which establishes the general fused weights in terms of very well-poised, balanced elliptic $_{12} v_{11}$ hypergeometric series (from \eqref{elliptichypergeometricpoised}). In the similar setting of the fused eight-vertex SOS model, such weights were identified by Date-Jimbo-Kuniba-Miwa-Okado \cite{ESSOSM,FEVSOSM}. For completeness (and since their setting does not precisely coincide with ours), we will provide a proof of the below theorem in Section \ref{FusedWeightsGeneral}.

\begin{thm}
	
	\label{fusedweighthypergeometric}
	
	Let $J \in \mathbb{Z}_{> 0}$ and $i_1, j_1, i_2, j_2 \in \mathbb{Z}_{\ge 0}$. Set 
	\begin{flalign*}
	a_1 & = \lambda + 2 \eta (2 j_1 + j_2 - J); \quad a_6 = 2 \eta j_1; \quad a_7 = 2 \eta j_2; \quad a_8 = \lambda + 2 \eta j_1; \\
	a_9 & = \lambda + 2 \eta (i_1 + 2 j_1 - J - 1 - \Lambda); \quad a_{10} = \eta \Lambda + v + 2 \eta (j_2 - i_1 - 1); \\
	a_{11} & = \eta \Lambda - v - 2 \eta (i_1 - j_2 + J); \qquad \quad a_{12} = \lambda + 2 \eta (i_2 + j_1 + j_2 - J). 
	\end{flalign*}
	
	\noindent Then, 
	\begin{flalign*}
	W_J & \big( i_1, j_1; i_2, j_2 \b| v, \lambda \big) \\
	& = f (2 \eta)^{i_2 - i_1} \displaystyle\frac{[2 \eta \Lambda]_{i_1}}{[2 \eta \Lambda]_{i_2}} \displaystyle\frac{\big[ 2 \eta i_1 \big]_{j_2} \big[ 2 \eta (J - j_2) \big]_{j_1} \big[ \eta \Lambda - v - 2 \eta (i_1 + j_1) \big]_{J - j_1 - j_2} \big[ 2 \eta (\Lambda - i_1 + j_2) \big]_{j_1} }{\big[ 2 \eta j_1 \big]_{j_1} \big[ \eta \Lambda - v \big]_J  } \\
	& \quad \times \displaystyle\frac{ \big[ \lambda + 2 \eta i_2 \big]_{J - j_1 - j_2}  \big[ \lambda + 2 \eta (i_1 + 2j_1 - J) - \eta \Lambda - v \big]_{j_2} \big[ \lambda + v + 2 \eta (i_1 + 2j_1 - 1) - \eta \Lambda \big]_{j_1}}{\big[ \lambda + 2 \eta (2j_1 + j_2 - J) \big]_{j_2} \big[ \lambda + 2 \eta j_1 \big]_{J - j_1 - j_2} \big[ \lambda + 2 \eta (2 j_1 + j_2 - J - 1) \big]_{j_1}} \\
	& \quad {\times _{12} v_{11}} (a_1; a_6, a_7, a_8, a_9, a_{10}, a_{11}, a_{12}; 1). 
	\end{flalign*}

\end{thm}

In view of this theorem, the fused $B$ and $D$ functions can be viewed as partition functions of (fused) vertex models whose weights are given by $_{12} v_{11}$ elliptic hypergeometric series. Interestingly, another family of symmetric functions whose single-variable specializations are $_{12} v_{11}$ elliptic hypergeometric integrals was introduced and studied by Rains \cite{SAF, TEHI}. Those functions seems to be different from the fused $B$ and $D$ functions described above. In particular, through a suitable choice of parameters in Proposition \ref{bdidentities}, one can show that the latter satisfy a Cauchy identity that does not appear to hold for the former. Furthermore, the functions of Rains seem to satisfy a different collection of algebraic properties that imply quite general identities for elliptic hypergeometric integrals on the $BC_n$ root system, and we do not know how to derive such general statements from the fused $B$ and $D$ functions. In fact, we are not aware of any connection between his symmetric functions with the fused $B$ and $D$ functions, and we would certainly find one intriguing. 

In a different direction, it will later be relevant to observe that the $_{12} v_{11}$ elliptic hypergeometric series above simplify considerably in the case $v = - \eta \Lambda$. This result is more precisely stated in the following theorem, to be established later, in Section \ref{FusedWeightsGeneral}. 

\begin{thm}
	
	\label{vlambdaspecialization}
	
	Fix $J \in \mathbb{Z}_{> 0}$ and $i \in \mathbb{Z}_{\ge 0}$, and assume that $v = - \eta \Lambda$. Then, for any $j_1, j_2 \in \{ 0, 1, \ldots , J \}$, we have that 
	\begin{flalign*}
	W_J \big( i_1, j_1; i_2, j_2 \b| v, \lambda\big) & = f (2 \eta)^{i_2 - i_1} \displaystyle\frac{\textbf{\emph{1}}_{i_1 \ge j_2} \big[ 2 \eta J \big]_J}{\big[ 2 \eta j_1 \big]_{j_1} \big[ 2 \eta (J - j_1) \big]_{J - j_1}} \displaystyle\frac{\big[ 2 \eta \Lambda \big]_{i_1}}{\big[ 2 \eta \Lambda \big]_{i_2}} \displaystyle\frac{ \big[ 2 \eta i_1 \big]_{j_2} \big[ 2 \eta (\Lambda - i_1) \big]_{J - j_2} }{\big[ 2 \eta \Lambda \big]_J} \\
	& \qquad \times \displaystyle\frac{\big[ \lambda + 2 \eta i_2 \big]_{J - j_1} \big[ \lambda + 2 \eta (i_2 + j_1 - \Lambda - 1) \big]_{j_1}}{\big[ \lambda + 2 \eta j_1 \big]_{J - j_1} \big[ \lambda + 2 \eta (2j_1 - J - 1) \big]_{j_1} }. 
	\end{flalign*}
	
\end{thm}

Although IRF models with weights given by $_{12} v_{11}$ hypergeometric series (similar to the ones provided by Theorem \ref{fusedweighthypergeometric}) were studied at length in \cite{ESSOSMLHP,ESSOSM,FEVSOSM}, we are not aware of any previous specialized analysis of the degenerated vertex model whose weights factor as in Theorem \ref{vlambdaspecialization}. We will see that, after a trigonometric degeneration and stochastic correction, this model will becomes the dynamical $q$-Hahn boson model introduced in Section \ref{LambdaSpecializationModelGeneral}. 

Before proceeding to the proofs of the two theorems stated above, let us mention that Theorem \ref{vlambdaspecialization} might also have an application to the theory of symmetric functions. Indeed, in the quantum degeneration of our elliptic setting (obtained by letting $\tau$ and $- \lambda$ both tend to $\textbf{i} \infty$), the elliptic weight functions studied in this section become the spin Hall-Littlewood functions that were analyzed in \cite{FSRF}. Very recently, Borodin and Wheeler \cite{SP} observed that applying the fusion procedure and then specializing as in (the quantum degeneration of) Theorem \ref{vlambdaspecialization} transforms the spin Hall-Littlewood symmetric functions into one-parameter deformations of the $q$-Whittaker functions, which they called \textit{spin $q$-Whittaker functions}. It is plausible that, using Theorem \ref{vlambdaspecialization}, one can introduce and analyze elliptic generalizations of these spin $q$-Whittaker functions, which should also satisfy branching, Cauchy, and Pieri type identities. However, we will not pursue this further here.

\section{Fused Weights}

\label{FusedWeights}

The purpose of this section is to establish Theorem \ref{fusedweighthypergeometric}, which explicitly evaluates the fused weights $W_J \big( i_1, j_1; i_2, j_2 \b| v, \lambda \big)$ given by Definition \ref{fusedweightdefinition} in terms of elliptic hypergeometric series. As a consequence, we will also derive Theorem \ref{vlambdaspecialization}, which provides a special value of $v$ at which these weights simplify into a fully factored form.

The fused weights given by Theorem \ref{fusedweighthypergeometric} have essentially been determined in a number of very similar settings. In the quantum case, they were originally evaluated by Kirillov and Reshetikhin \cite{REA}, and later different proofs were found by Mangazeev \cite{ESVM} and Corwin-Petrov \cite{SHSVML}. In the elliptic case, fused weights similar to those in Theorem \ref{fusedweighthypergeometric} were first found by Date-Jimbo-Kuniba-Miwa-Okado \cite{ESSOSM} in their study of fusion of the eight-vertex SOS model. Since then, similar derivations have been performed in a number of contexts, including from the perspective of (modular) elliptic hypergeometric series \cite{ESEMHF}; matrix elements of co-representations of the elliptic quantum group \cite{EQGEHS}; and co-braidings in dynamical quantum groups \cite{PADQG, EQGEHSRS}. 

Unfortunately, we are not aware of any work that evaluates these weights in our setup and terminology. The closest we know of seems to be that of Koelink-van Norden-Rosengren \cite{EQGEHS} (or \cite{PADQG, EQGEHSRS}), which evaluates matrix elements of co-representations of $E_{\tau, \eta} (\mathfrak{sl}_2)$; in principle, one could derive our weights from their identities, but this would require introducing some additional notation that we wish to avoid. 

Thus, we instead derive our weights in a similar way to what was done in the original work of Date-Jimbo-Kuniba-Miwa-Okado \cite{ESSOSM}. This will proceed as follows. First, in Section \ref{WeightRecursion}, we find two recursions for the fused weights $W_J (i_1, j_1; i_2, j_2)$; these are given by Lemma \ref{fusedweightrecursion} and Lemma \ref{fusedweightrecursion2}. In particular, the second expresses $W_J (i_1, j_1; i_2, j_2)$ through weights of the form $W_k (i, j; i', 0)$ and $W_ k (i, j; i', k)$. Then, in Section \ref{FusedWeights0J}, we use the first recursion given by Lemma \ref{fusedweightrecursion} to obtain explicit, factored forms for the latter two types of weights. We then conclude with the proofs of Theorem \ref{fusedweighthypergeometric} and Theorem \ref{vlambdaspecialization} in Section \ref{FusedWeightsGeneral}. 

\subsection{A Recursion for the Fused Weights} 

\label{WeightRecursion}  

We begin with the following lemma, which provides a recursive identity for the $\widehat{W}_J$ fused weights (recall their definition from \eqref{wjwsums}). 

\begin{lem} 
	
	\label{fusedweightrecursion} 
	
	Fix a positive integer $J$, and let $i_1, j_1, i_2, j_2 \in \mathbb{Z}_{\ge 0}$. Then, 
	\begin{flalign*}
	& \widehat{W}_J \big( i_1, j_1; i_2, j_2 \b| v, \lambda \big) \\
	& = \widehat{W}_{J - 1} \big( i_1, j_1; i_2, j_2 \b| v, \lambda - 2 \eta \big) \displaystyle\frac{f \big( \eta \Lambda - 2 \eta (i_2 + J - 1) - v \big) f \big( \lambda + 2 \eta i_2 \big) }{f \big( \eta \Lambda - 2 \eta (J - 1) - v \big) f \big( \lambda \big)} \\
	& \quad + \widehat{W}_{J - 1} \big( i_1, j_1 - 1; i_2 - 1, j_2 \b| v, \lambda + 2 \eta \big) \displaystyle\frac{f \big( v + \lambda - \eta \Lambda + 2 \eta (i_2 - 1 + J) \big) f \big( 2 \eta \big) }{f \big( \eta \Lambda - 2 \eta (J - 1) - v \big) f \big( \lambda \big)} \\
	& \quad + \widehat{W}_{J - 1} \big( i_1, j_1; i_2 + 1, j_2 - 1 \b| v, \lambda - 2 \eta \big) \\
	& \qquad \qquad \qquad \times \displaystyle\frac{f \big( \lambda - v - \eta \Lambda + 2 \eta (i_2 + 1 - J) \big) f \big(2 \eta (\Lambda - i_2) \big) f \big( 2 \eta (i_2 + 1) \big) }{f \big( \eta \Lambda - 2 \eta (J - 1) - v \big) f \big( \lambda \big) f(2 \eta)} \\
	& \quad + \widehat{W}_{J - 1} \big( i_1, j_1 - 1; i_2, j_2 - 1 \b| v, \lambda + 2 \eta \big) \displaystyle\frac{f \big( 2 \eta (i_2 - J + 1) - \eta \Lambda - v \big) f \big( \lambda + 2 \eta (i_2 - \Lambda) \big) }{f \big( \eta \Lambda - 2 \eta (J - 1) - v \big) f \big( \lambda \big)}. 
	\end{flalign*}

\end{lem}

\begin{proof}
	
	Recall from \eqref{wjwsums} that $\widehat{W}_J \big( i_1, j_1; i_2, j_2 \b| v, \lambda \big)$ is a sum over $\mathcal{J}_1 = (j_{1, 1}, j_{1, 2}, \ldots , j_{1, J}) \in \{ 0, 1 \}^J$ and $\mathcal{J}_2 = (j_{2, 1}, j_{2, 2}, \ldots , j_{2, J}) \in \{0, 1 \}^J$. Letting $\widetilde{\mathcal{J}_1} = (j_{1, 1}, j_{1, 2}, \ldots , j_{1, J - 1}) \in \{ 0, 1 \}^{J - 1}$ and $\widetilde{\mathcal{J}_2} = (j_{2, 1}, j_{2, 2}, \ldots , j_{2, J - 1}) \in \{0, 1 \}^{J - 1}$, we may express the sum \eqref{wjwsums} over $\mathcal{J}_1$ and $\mathcal{J}_2$ by first summing over $j_{1, J}$ and $j_{2, J}$, and then summing over $\widetilde{\mathcal{J}_1}$ and $\widetilde{\mathcal{J}_2}$. 
	
	The sums over $j_{1, J}$ and $j_{2, J}$ give rise to $W_1$ weights, and the sums over $\widetilde{\mathcal{J}_1}$ and $\widetilde{\mathcal{J}_2}$ give rise to $\widehat{W}_{J - 1}$ weights. Specifically, this yields 
	\begin{flalign*}
	\widehat{W}_J \big( & i_1, j_1; i_2, j_2 \b| v, \lambda \big) \\
	& = \widehat{W}_{J - 1} \big( i_1, j_1; i_2, j_2 \b| v, \lambda - 2 \eta \big) W_1 \big( i_2, 0; i_2, 0 \b| v+ 2 \eta (J - 1), \lambda \big) \\
	& \qquad + \widehat{W}_{J - 1} \big( i_1, j_1 - 1; i_2 - 1, j_2 \b| v, \lambda + 2 \eta \big) W_1 \big( i_2 - 1, 1; i_2, 0 \b| v+ 2 \eta (J - 1), \lambda \big) \\
	& \qquad + \widehat{W}_{J - 1} \big( i_1, j_1; i_2 + 1, j_2 - 1 \b| v, \lambda - 2 \eta \big) W_1 \big( i_2 + 1, 0; i_2, 1 \b| v+ 2 \eta (J - 1), \lambda \big) \\
	& \qquad + \widehat{W}_{J - 1} \big( i_1, j_1 - 1; i_2, j_2 - 1 \b| v, \lambda + 2 \eta \big) W_1 \big( i_2, 1; i_2, 1 \b| v+ 2 \eta (J - 1), \lambda \big);
	\end{flalign*}
	
	\noindent here, the different values of $\lambda$ in the $\widehat{W}_{J - 1}$ weights arise from \eqref{modifiedabcd}. 
	
	Now the lemma follows from inserting the explicit forms of the $W_1$ weights (given by \eqref{weightsj1}) into the above equation. 
\end{proof}

Directly from Lemma \ref{fusedweightrecursion}, one can express $\widehat{W}_J \big( i_1, j_1; i_2, j_2 \big)$ as a sum of $4^J$ factored terms. However, Theorem \ref{fusedweighthypergeometric} states that this fused weight can in fact be expressed as a sum of $J$ fully factored terms. 

The proof of this fact is strongly based on a method that dates back at least to the work of Date-Jimbo-Kuniba-Miwa-Okado \cite{ESSOSM} in 1988 (although we also refer to the works \cite{REA, ESVM} for alternative derivations in the six-vertex case). In particular, it uses two remarkable facts. 

The first fact is that the weight $W_J (i_1, j_1; i_2, j_2)$ can be expressed as a sum of $J$ products of terms of the form $W_k (i_1, j_1; i_2, 0)$ and $W_k (i_1, j_1; i_2, k)$. This is given more explicitly by the following lemma, whose proof uses Lemma \ref{exchangeability} from Section \ref{Fusion}. 

\begin{lem}

\label{fusedweightrecursion2}

Fix $J \in \mathbb{Z}_{> 0}$ and $i_1, j_1, i_2, j_2 \in \mathbb{Z}_{\ge 0}$. Then, 
\begin{flalign*}
W_J \big( i_1, j_1; i_2, j_2 \b| v, \lambda \big) & =  \displaystyle\sum_{k = 0}^J W_{j_2} \big( i_1, k; i_1 + k - j_2, j_2 \b| v, \lambda - 2 \eta (J - 2j_1 - j_2 + 2k)  \big)  \\
& \qquad \times W_{J - j_2} \big( i_1 + k - j_2, j_1 - k; i_2, 0 \b| v  + 2 \eta j_2, \lambda \big). 
\end{flalign*}

\end{lem}

\begin{proof}

Denote $\mathcal{K} = (1, 1, \ldots,  1, 0, 0, \ldots , 0) \in \{ 0, 1 \}^J$, the $J$-dimensional $0-1$ vector whose first $j_2$ components are equal to $1$ and whose last $J - j_2$ components are equal to $0$. In view of Lemma \ref{exchangeability}, we have that 
\begin{flalign}
\label{wjksum} 
W_J \big( i_1, j_1; i_2, j_2 \b| v, \lambda \big) = W_J \big( i_1, j_1; i_2, \mathcal{K} \b| v, \lambda \big) = \displaystyle\sum_{|\mathcal{J}_1| = j_1} W_J (i_1, \mathcal{J}_1; i_2, \mathcal{K} \b| v, \lambda). 
\end{flalign}

\noindent Now let $\mathcal{J}_1^{(1)}$ denote the $j_2$-dimensional $0-1$ vector formed from the first $j_2$ coordinates of $\mathcal{J}_1$, and let $\mathcal{J}_1^{(2)}$ denote the $(J - j_2)$-dimensional $0-1$ vector formed from the last $J - j_2$ coordinates of $\mathcal{J}_1$. Then, 
\begin{flalign}
\label{wjj2jj2} 
\begin{aligned}
W_J (i_1, \mathcal{J}_1; i_2, \mathcal{K} \b| v, \lambda ) & = W_{j_2} \big( i_1, \mathcal{J}_1^{(1)}; i_1 + |\mathcal{J}_1^{(1)}| - j_2, j_2 \b| v, \widetilde{\lambda} \big) \\
& \qquad \times W_{J - j_2} \big( i_1 + |\mathcal{J}_1^{(1)}| - j_2, \mathcal{J}_1^{(2)}; i_2, 0 \b| v + 2 \eta j_2, \lambda \big).
\end{aligned}
\end{flalign} 

\noindent Here, $\widetilde{\lambda}$ is determined from $\lambda$ through the relations \eqref{modifiedabcd}. In particular, $\lambda$ increases by $2 \eta$ for every $1$ in $\mathcal{J}_1^{(2)}$ and decreases by $2 \eta$ for every $0$ in $\mathcal{J}_1^{(2)}$. Therefore, $\widetilde{\lambda} = \lambda - 2 \eta \big( J - j_2 - 2 | \mathcal{J}_1^{(2)} | \big)$. 

Thus, the lemma follows from inserting \eqref{wjj2jj2} into \eqref{wjksum} and denoting $k = \big| \mathcal{J}_1^{(1)} \big| = j_1 - \big| \mathcal{J}_1^{(2)} \big|$. 
\end{proof}

The second fact is that weights of the form $W_k (i_1, j_1; i_2, 0)$ and $W_k (i_1, j_1; i_2, k)$ factor completely. This was originally verified in the related context of the eight-vertex SOS model in the work of Date-Jimbo-Kuniba-Miwa-Okado \cite{ESSOSM}. The more recent work \cite{EQGEHS} strongly suggests that an alternative derivation should be accessible directly from repeated application of the dynamical Yang-Baxter equation and furthermore indicates that this factorization is a phenomenon that holds in a far more general setting (for example, one should be able to replace the evaluation Verma module with an arbitrary representation of $E_{\tau, \eta} (\mathfrak{sl}_2)$). However, since explaining that framework would require introducing some additional terminology that we wish to avoid, we instead verify the factorizations of the $W_k (i_1, j_1; i_2, 0)$ and $W_k (i_1, j_1; i_2, k)$ weights by closely following the original method of Date-Jimbo-Kuniba-Miwa-Okado \cite{ESSOSM} in the next section.

\subsection{Fused Weights in the Cases \texorpdfstring{$j_2 \in \{ 0, J \}$}{} }

\label{FusedWeights0J}

In this section we analyze fused weights of the form $W_k (i_1, j_1; i_2, 0)$ and $W_k (i_1, j_1; i_2, k)$. Using Lemma \ref{fusedweightrecursion}, we give explicit, fully factored forms for these weights; see Proposition \ref{wjj20} for the case $j_2 = 0$ and Proposition \ref{wjj2j} for the case $j_2 = k$. The proofs of these propositions are similar to those of Lemma 2.1.3 and Lemma 2.1.4 in \cite{ESSOSM}. 

\begin{prop}
	
	\label{wjj20} 
	
	Let $i \in \mathbb{Z}_{\ge 0}$, $J \in \mathbb{Z}_{> 0}$, and $j \in \mathbb{Z}$. If $j < 0$ or $j > J$, then $W_J (i, j; i + j - J, J) = 0$. Otherwise, 
	\begin{flalign}
	\label{wjj20identity}
	\begin{aligned}
	W_J \big( i, j; i + j, 0 \b| v, \lambda \big) & = f (2 \eta)^j \displaystyle\frac{\big[ 2 \eta J \big]_J }{\big[ 2 \eta (J - j) \big]_{J - j} \big[ 2 \eta j \big]_j} \displaystyle\frac{\big[ \lambda + 2 \eta (i + j) \big]_{J - j} }{\big[ \eta \Lambda - v \big]_J} \\
	& \qquad \times \displaystyle\frac{\big[ v + \lambda + 2 \eta (i + 2j - 1) - \eta \Lambda \big]_j \big[ \eta \Lambda - 2 \eta (i + j) - v \big]_{J - j}}{\big[ \lambda + 2 \eta j \big]_{J - j} \big[ \lambda + 2 \eta (2j - J - 1) \big]_j}.
	\end{aligned}
	\end{flalign}
	
\end{prop}

\begin{proof}
	
	By \eqref{wjlargesmallj}, the proposition holds when $j < 0$ or $j > J$, so it remains to establish \eqref{wjj20identity}. To that end, we induct on $J$; the cases $J \in \{ 0, 1 \}$ are quickly verified so suppose that \eqref{wjj20identity} holds when $J$ is replaced by $J - 1$, for all $0 \le j \le J - 1$. To establish \eqref{wjj20identity}, let us first apply Lemma \ref{fusedweightrecursion} in the case $j_2 = 0$. Using \eqref{wjlargesmallj} and \eqref{hatwwj}, this yields 
	\begin{flalign}
	\label{wjj20recursion}
	\begin{aligned}
	& W_J \big( i, j; i + j, 0 \b| v, \lambda \big) \\
	& \quad = W_{J - 1} \big( i, j; i + j, 0 \b| v, \lambda - 2 \eta \big) \displaystyle\frac{f \big( \eta \Lambda + 2 \eta (i + j + J - 1) - v \big) f \big( \lambda + 2 \eta (i + j) \big) }{f \big( \eta \Lambda - 2 \eta (J - 1) - v \big) f \big( \lambda \big)} \\
	& \qquad + W_{J - 1} \big( i, j - 1; i + j - 1, 0 \b| v, \lambda + 2 \eta \big) \displaystyle\frac{f \big( v + \lambda + 2 \eta (i + j - 1 + J) - \eta \Lambda \big) f (2 \eta) }{f \big( \eta \Lambda - 2 \eta (J - 1) - v  \big) f \big( \lambda \big)}. 
	\end{aligned}
	\end{flalign}
	
	\noindent From the inductive hypothesis, we have that 
	\begin{flalign}
	\label{wj1j201}
	\begin{aligned}
	W_{J - 1} & \big(  i, j; i + j, 0 \b| v, \lambda - 2 \eta \big) \\
	& = f (2 \eta)^j \displaystyle\frac{\big[ 2 \eta (J - 1) \big]_{J - 1}}{\big[ 2 \eta (J - j - 1) \big]_{J - j - 1} \big[ 2 \eta j \big]_j} \displaystyle\frac{ \big[ \eta \Lambda - 2 \eta (i + j) - v \big]_{J - j - 1} }{\big[ \eta \Lambda - v \big]_{J - 1}} \\
	& \qquad \times \displaystyle\frac{\big[ v + \lambda + 2 \eta (i + 2j - 2) - \eta \Lambda \big]_j \big[ \lambda + 2 \eta (i + j - 1) \big]_{J - j - 1} }{\big[ \lambda + 2 \eta (j - 1) \big]_{J - j - 1} \big[ \lambda + 2 \eta (2j - J - 1) \big]_j } 
	, 
	\end{aligned} 
	\end{flalign}
	
	\noindent and 
	\begin{flalign}
	\label{wj1j202}
	\begin{aligned}
	W_{J - 1} & \big( i, j - 1; i + j - 1, 0 \b| v, \lambda + 2 \eta \big) \\
	& = f(2 \eta)^{j - 1} \displaystyle\frac{\big[ 2 \eta (J - 1) \big]_{J - 1}}{\big[ 2 \eta (J - j) \big]_{J - j} \big[ 2 \eta (j - 1) \big]_{j - 1}} \displaystyle\frac{ \big[ \eta \Lambda - 2 \eta (i + j - 1) - v \big]_{J - j} }{\big[ \eta \Lambda - v \big]_{J - 1}} \\
	& \qquad \times \displaystyle\frac{\big[ v + \lambda + 2 \eta (i + 2j - 2) - \eta \Lambda \big]_{j - 1} \big[ \lambda + 2 \eta (i + j) \big]_{J - j} }{\big[ \lambda + 2 \eta j \big]_{J - j} \big[ \lambda + 2 \eta (2j - J - 1) \big]_{j - 1} }.
	\end{aligned}
	\end{flalign}
	
	\noindent Inserting \eqref{wj1j201} and \eqref{wj1j202} into \eqref{wjj20recursion} (using the second identity from \eqref{elliptichypergeometricidentities1}) yields 
	\begin{flalign}
	\label{wjj20recursion2}
	\begin{aligned}
	W_J \big( i, j; i + j, 0 \b| v, \lambda \big) & = f(2 \eta)^j \displaystyle\frac{S}{f (\lambda)} \displaystyle\frac{\big[ 2 \eta (J - 1) \big]_{J - 1}}{\big[ 2 \eta (J - j) \big]_{J - j} \big[ 2 \eta j \big]_{j}} \displaystyle\frac{ \big[ \eta \Lambda - 2 \eta (i + j) - v \big]_{J - j - 1} }{\big[ \eta \Lambda - v \big]_J } \\
	& \quad \times \displaystyle\frac{\big[ v + \lambda + 2 \eta (i + 2j - 2) - \eta \Lambda \big]_{j - 1} \big[ \lambda + 2 \eta (i + j) \big]_{J - j} }{\big[ \lambda + 2 \eta j \big]_{J - j} \big[ \lambda + 2 \eta (2j - J - 1) \big]_j} ,
	\end{aligned}
	\end{flalign}
	
	\noindent where 
	\begin{flalign*}
	S & = f \big( 2 \eta (j - J) \big) f (\lambda + 2 \eta j) f \big( v + 2 \eta (i + j + J - 1) - \eta \Lambda \big) f \big( v + \lambda + 2 \eta (i + j - 1) - \eta \Lambda \big) \\
	& \quad + f (2 \eta j) f \big( \lambda + 2\eta (j - J) \big) f \big( v + \lambda - \eta \Lambda + 2 \eta (i + j - 1 + J) \big) f \big( \eta \Lambda - 2 \eta (i + j - 1) - v \big).  
	\end{flalign*}
	
	\noindent To evaluate $S$, we apply \eqref{quarticrelationf} with 
	\begin{flalign*}
	w = \displaystyle\frac{\lambda}{2} - \eta J; \quad x = \displaystyle\frac{\lambda}{2} + \eta (2j - J); \quad y = \displaystyle\frac{\lambda}{2} + \eta J; \quad z = \displaystyle\frac{\lambda}{2} + v - \eta \Lambda + \eta (2i + 2j + J - 2), 
	\end{flalign*}
	
	\noindent to find that 
	\begin{flalign}
	\label{fsumj20} 
	S = f (2 \eta J) f (\lambda) f \big( \lambda + v - \eta \Lambda + 2 \eta (i + 2j - 1) \big) f \big( \eta \Lambda - v - 2 \eta (i + J - 1) \big). 
	\end{flalign} 
	
	\noindent Inserting \eqref{fsumj20} into \eqref{wjj20recursion2} (and using the second identity from \eqref{elliptichypergeometricidentities1}) yields \eqref{wjj20identity}; this establishes the proposition.  
\end{proof}

\begin{prop}

\label{wjj2j}

Let $i \in \mathbb{Z}_{\ge 0}$, $J \in \mathbb{Z}_{> 0}$, and $j \in \mathbb{Z}$. If $j < 0$ or $j > J$, then $\widehat{W}_J (i, j; i + j - J, J) = 0$. Otherwise, 
\begin{flalign}
\label{wj2jidentity}
\begin{aligned}
& W_J (i, j; i + j - J, J) \\
& \quad = f(2 \eta)^{j - J} \displaystyle\frac{[2 \eta J]_J}{\big[ 2 \eta (J - j) \big]_{J - j} [2 \eta j]_j} \displaystyle\frac{[2 \eta i]_{J - j} \big[ 2 \eta ( i + j - J) - \eta \Lambda - v \big]_j \big[ 2 \eta (\Lambda - i - j + J) \big]_{J - j}}{ [\eta \Lambda - v]_J } \\
& \quad \times \displaystyle\frac{\big[ \lambda + 2 \eta (i + j - J) - \eta \Lambda - v \big]_{J - j} \big[ \lambda + 2 \eta (i + 2j - J - \Lambda - 1) \big]_j}{[\lambda + 2 \eta j]_{J - j} \big[ \lambda + 2 \eta (2j - J - 1) \big]_j} .
\end{aligned} 
\end{flalign} 
\end{prop}

\begin{proof}

The proof is similar to that of Proposition \ref{wjj20}. Again, by \eqref{wjlargesmallj}, the proposition holds when $j < 0$ or $j > J$, so it remains to establish \eqref{wj2jidentity}. To that end, we induct on $J$; the cases $J \in \{ 0, 1 \}$ are quickly verified so suppose that \eqref{wjj20identity} holds when $J$ is replaced by $J - 1$, for all $0 \le j \le J - 1$. 

Let us now establish \eqref{wjj20identity}. In an entirely analogous way to how we deduced \eqref{wjj20recursion2} (by applying the inductive hypothesis and inserting into Lemma \ref{fusedweightrecursion}, where now the first two terms on the right side of that identity are equal to $0$), we find that 
\begin{flalign}
\label{wj2jrecursion2}
\begin{aligned}
W_J \big( & i, j; i + j - J, J \b| v, \lambda \big) \\
& = f (2 \eta)^{j - J} \displaystyle\frac{T}{f (\lambda)} \displaystyle\frac{\big[ 2 \eta (J - 1) \big]_{J - 1} }{\big[ 2 \eta j]_j \big[ 2 \eta (J - j) \big]_{J - j}}\displaystyle\frac{\big[ 2 \eta i \big]_{J - j}\big[ 2 \eta (i + j - J) - \eta \Lambda - v \big]_{j - 1} }{  \big[ \eta \Lambda - v \big]_J } \\
& \quad \times \displaystyle\frac{\big[ 2 \eta (\Lambda - i - j+ J) \big]_{J - j}  \big[ \lambda + 2 \eta (i + 2j - \Lambda - J - 1) \big]_j \big[ \lambda + 2 \eta (i + j - J) - \eta \Lambda - v \big]_{J - j - 1}}{\big[ \lambda + 2 \eta j \big]_{J - j} \big[ \lambda + 2 \eta (2j - J - 1) \big]_j } , 
\end{aligned} 
\end{flalign}

\noindent where 	
\begin{flalign*}
T & = f \big( 2 \eta (J - j) \big) f (\lambda + 2 \eta j) f \big( \lambda - v - \eta \Lambda + 2 \eta (i + j - 2 J + 1) \big) f \big( 2 \eta (i + j - J + 1) - \eta \Lambda - v \big) \\
& \qquad f ( 2 \eta j) f \big( \lambda + 2 \eta (j - J) \big) f \big( 2 \eta (i + j - 2 J + 1) - \eta \Lambda - v \big)f \big( \lambda + 2 \eta (i + j - J + 1) - \eta \Lambda - v \big).  
\end{flalign*}

\noindent To evaluate $T$, we apply \eqref{quarticrelationf} with
\begin{flalign*}
w = \displaystyle\frac{\lambda}{2} - \eta J; \quad x = \displaystyle\frac{\lambda}{2} + \eta (2i + 2j - 3 J + 2) - \eta \Lambda - v; \quad y = \displaystyle\frac{\lambda}{2} + \eta J; \quad z = \displaystyle\frac{\lambda}{2} + \eta (2j - J), 
\end{flalign*}

\noindent which yields 
\begin{flalign*}
T = f (\lambda) f (2 \eta J) f \big( \lambda + 2 \eta (i + 2j - 2 J + 1) - \eta \Lambda - v \big) f \big( 2 \eta (i - J + 1) - \eta \Lambda - v \big). 
\end{flalign*}

\noindent Inserting the above identity into \eqref{wj2jrecursion2} yields \eqref{wj2jidentity}. 
\end{proof}

\subsection{General Fused Weights}

\label{FusedWeightsGeneral}

The purpose of this section is to establish Theorem \ref{fusedweighthypergeometric}, which provides an explicit form for the general $W_J \big( i_1, j_1; i_2, j_2 \b| v, \lambda \big)$ fused weight, and Theorem \ref{vlambdaspecialization}, which shows how this weight simplifies when $v = - \eta \Lambda$. Let us begin with the former.

\begin{proof}[Proof of Theorem \ref{fusedweighthypergeometric}]

Using Proposition \ref{wjj20} and Proposition \ref{wjj2j}, we obtain that 
\begin{flalign*}
W_{j_2} & \big( i_1, k; i_1 + k - j_2, j_2 \b| v, \lambda + 2 \eta (2j_1 + j_2 - 2k - J) \big) \\
& = f (2 \eta)^{j_2 - k} \displaystyle\frac{[2 \eta \Lambda]_{i_1}}{[2 \eta \Lambda]_{i_1 + k - j_2}} \displaystyle\frac{\big[ 2 \eta j_2 \big]_{j_2}}{\big[ 2 \eta (j_2 - k) \big]_{j_2 - k} \big[ 2 \eta k \big]_k} \displaystyle\frac{\big[ 2 \eta i_1 \big]_{j_2 - k} \big[ 2 \eta (i_1 + k - j_2) - \eta \Lambda - v \big]_k}{\big[ \eta \Lambda - v \big]_{j_2} } \\
& \quad \quad \times \displaystyle\frac{\big[ \lambda + 2 \eta (i_1 + 2j_1 - k - J) - \eta \Lambda - v \big]_{j_2 - k} \big[ \lambda + 2 \eta (i_1 + 2j_1 - J - 1 - \Lambda) \big]_k}{\big[ \lambda + 2 \eta (2j_1 + j_2 - k - J) \big]_{j_2 - k} \big[ \lambda + 2 \eta (2j_1 - J - 1) \big]_k },  
\end{flalign*}

\noindent and
\begin{flalign*}
W_{J - j_2} \big( & i_1 + k - j_2, j_1 - k; i_2, 0 | v+ 2 \eta j_2, \lambda \big) \\
& = f (2 \eta)^{k - j_1} \displaystyle\frac{[2 \eta \Lambda]_{i_1 + k - j_2}}{[2 \eta \Lambda]_{i_2}} \displaystyle\frac{\big[ 2 \eta (J - j_2) \big]_{J - j_2}}{\big[ 2 \eta (j_1 - k) \big]_{j_1 - k} \big[ 2 \eta (J - j_1 - j_2 + k) \big]_{J - j_1 - j_2 + k}} \\
& \qquad \times  \displaystyle\frac{\big[ \eta \Lambda - v - 2 \eta (i_1 + j_1) \big]_{J - j_1 - j_2 + k} \big[ 2 \eta (\Lambda - i_1 - k + j_2) \big]_{j_1 - k}}{\big[ \eta \Lambda - v - 2 \eta j_2 \big]_{J - j_2}} \\
& \qquad \times \displaystyle\frac{\big[ \lambda + 2 \eta i_2 \big]_{J - j_1 - j_2 + k}  \big[ \lambda + v + 2 \eta (i_1 + 2j_1 - k - 1) - \eta \Lambda \big]_{j_1 - k} }{ \big[ \lambda + 2 \eta (j_1 - k) \big]_{J - j_1 - j_2 + k} \big[ \lambda + 2 \eta (2j_1 + j_2 - 2k - J - 1) \big]_{j_1 - k}}.
\end{flalign*}

\noindent Using Lemma \ref{fusedweightrecursion2} and the above two identities, one can express $W_J (i_1, j_1; i_2, j_2)$ as a sum of $J$ terms; this sum can be seen to coincide with	the multiple of the hypergeometric series appearing in the statement of the theorem, by using the six identities 
\begin{flalign*}
& \displaystyle\frac{\big[ 2 \eta j_2 \big]_{j_2}}{\big[ 2 \eta (j_2 - k) \big]_{j_2 - k} \big[ 2 \eta k \big]_k} \displaystyle\frac{\big[ 2 \eta (J - j_2) \big]_{J - j_2}}{\big[ 2 \eta (j_1 - k) \big]_{j_1 - k} \big[ 2 \eta (J - j_1 - j_2 + k) \big]_{J - j_1 - j_2 + k}}  \\
& \qquad \qquad \qquad \qquad \qquad \qquad \qquad \qquad \qquad = \displaystyle\frac{\big[ 2 \eta j_2 \big]_k}{\big[ -2 \eta \big]_k} \displaystyle\frac{\big[ 2 \eta (J - j_2) \big]_{j_1}}{\big[ 2 \eta (j_1 + j_2 - J - 1) \big]_k} \displaystyle\frac{\big[ 2 \eta j_1 \big]_k}{\big[ 2 \eta j_1 \big]_{j_1}}; 
\end{flalign*}

\begin{flalign*}
\displaystyle\frac{\big[ 2 \eta i_1 \big]_{j_2 - k} \big[ 2 \eta (i_1 + k - j_2) - \eta \Lambda - v \big]_k}{ \big[ \eta \Lambda - v \big]_{j_2} \big[ \eta \Lambda - v - 2 \eta j_2 \big]_{J - j_2}} = \displaystyle\frac{ \big[ 2 \eta i_1 \big]_{j_2}}{\big[ 2 \eta (j_2 - i_1 - 1) \big]_k} \displaystyle\frac{\big[ \eta \Lambda + v + 2 \eta (j_2 - i_1 - 1 ) \big]_k}{\big[ \eta \Lambda - v \big]_J}; 
\end{flalign*}

\begin{flalign*}
& \displaystyle\frac{\big[ \lambda + 2 \eta (i_1 + 2 j_1 - k - J) - \eta \Lambda - v \big]_{j_2 - k}}{\big[ \lambda + 2 \eta (2j_1 + j_2 - k - J) \big]_{j_2 - k}} \\
& \qquad \qquad = \displaystyle\frac{\big[ \lambda + 2 \eta (i_1 + 2j_1 - J) - \eta \Lambda - v \big]_{j_2} }{\big[ \lambda + 2 \eta (i_1 + 2 j_1 - J) - \eta \Lambda - v \big]_k} \displaystyle\frac{\big[ \lambda + 2 \eta (2j_1 + j_2 - J) \big]_k}{\big[ \lambda + 2 \eta (2j_1 + j_2 - J) \big]_{j_2} }; 
\end{flalign*} 

\begin{flalign*}
& \big[ \eta \Lambda - v - 2 \eta (i_1 + j_1)  \big]_{J - j_1 - j_2 + k} \big[ 2 \eta (\Lambda - i_1 - k + j_2) \big]_{j_1 - k} \\
& \qquad \qquad = \big[ \eta \Lambda - v - 2 \eta (i_1 + j_1) \big]_{J - j_1 - j_2} \big[ \eta \Lambda - v - 2 \eta (i_1 - j_2 + J) \big]_k \displaystyle\frac{\big[ 2 \eta (\Lambda - i_1 + j_2) \big]_{j_1}}{ \big[ 2 \eta (\Lambda - i_1 + j_2) \big]_k};
\end{flalign*}

\begin{flalign*}
& \big[ \lambda + v + 2 \eta (i_1 + 2 j_1 - k - 1) - \eta \Lambda  \big]_{j_1 - k} \big[ \lambda + 2 \eta i_2 \big]_{J - j_1 - j_2 + k} \\
& \qquad \qquad = \displaystyle\frac{\big[ \lambda + v + 2 \eta (i_1 + 2 j_1 - 1) - \eta \Lambda  \big]_{j_1} }{\big[ \lambda + v + 2 \eta (i_1 + 2 j_1 - 1) - \eta \Lambda  \big]_k} \big[ \lambda + 2 \eta i_2 \big]_{J - j_1 - j_2} \big[ \lambda + 2 \eta (i_2 + j_1 + j_2 - J) \big]_k;
\end{flalign*}

\begin{flalign*}
\big[ & \lambda + 2 \eta (j_1 - k) \big]_{J - j_1 - j_2 + k} \big[ \lambda + 2 \eta (2j_1 + j_2 - 2k - J - 1) \big]_{j_1 - k} \\
&  \qquad = \big[ \lambda + 2 \eta j_1 \big]_{J - j_1 - j_2} \big[ \lambda + 2 \eta (2 j_1 + j_2 - J - 1)  \big]_{j_1} \displaystyle\frac{\big[ \lambda + 2 \eta (j_1 + j_2 - J - 1) \big]_k}{\big[ \lambda + 2 \eta j_1 \big]_k} \\
& \qquad \qquad \times \displaystyle\frac{f \big( \lambda + 2 \eta (2 j_1 + j_2 - J) \big)}{f \big( \lambda + 2 \eta (2j_1 + j_2 - 2k - J) \big)}, 
\end{flalign*}

\noindent each of which following from the three identities listed in \eqref{elliptichypergeometricidentities1}. 
\end{proof}

Now let us outline a proof of Theorem \ref{vlambdaspecialization}.

\begin{proof}[Proof of Theorem \ref{vlambdaspecialization} (Outline)]
	
We know of two ways to establish Theorem \ref{vlambdaspecialization}. 

The first is to observe that in the special case $v = - \eta \Lambda$, the $_{12} v_{11}$ elliptic hypergeometric series for the $W_J (i_1, j_1; i_2, j_2)$ fused weight (from Theorem \ref{fusedweighthypergeometric}) in fact becomes a very well-poised, balanced $_{10} v_9$ hypergeometric series. One is therefore able to apply the elliptic Jackson identity \eqref{hypergeometric109sumterminating} to deduce Theorem \ref{vlambdaspecialization} from Theorem \ref{fusedweighthypergeometric}. 

The second way is to induct as in the proofs of Proposition \ref{wjj20} and Proposition \ref{wjj2j} in Section \ref{FusedWeights0J}. Let us briefly outline this. To that end, denote the right side of the identity in Theorem \ref{vlambdaspecialization} by $R = R_J$ and induct on $J$; the theorem is quickly verified in the case $J = 1$, so let us assume that it holds with $J$ replaced by $J - 1$. 

Applying Lemma \ref{fusedweightrecursion} in the case $v = - \eta \Lambda$ yields  
\begin{flalign}
\label{svlambda4}
\widehat{W}_J \big( i_1, j_1; i_2, j_2 \b| v, \lambda \big) = S + T, 
\end{flalign} 

\noindent where 
\begin{flalign*}
S & = \widehat{W}_{J - 1} \big( i_1, j_1; i_2, j_2 \b| v, \lambda - 2 \eta \big) \displaystyle\frac{f \big( 2 \eta (\Lambda + 1 - i_2 - J ) \big) f \big( \lambda + 2 \eta i_2 \big) }{f \big( 2 \eta (\Lambda + 1 - J) \big) f \big( \lambda \big)} \\
& \quad + \widehat{W}_{J - 1} \big( i_1, j_1 - 1; i_2 - 1, j_2 \b| v, \lambda + 2 \eta \big) \displaystyle\frac{f \big( \lambda + 2 \eta (i_2 - 1 + J - \Lambda) \big) f \big( 2 \eta (  \Lambda + 1 - i_2) \big) }{f \big( 2 \eta (\Lambda + 1 - J)  \big) f \big( \lambda \big)},
\end{flalign*}

\noindent and 
\begin{flalign*}
T & = \widehat{W}_{J - 1} \big( i_1, j_1; i_2 + 1, j_2 - 1 \b| v, \lambda - 2 \eta \big) \displaystyle\frac{f \big( \lambda + 2 \eta (i_2 + 1 - J) \big) f \big( 2 \eta (i_2 + 1) \big) }{f \big(  2 \eta (\Lambda + 1 - J)  \big) f \big( \lambda \big)} \\
& \quad + \widehat{W}_{J - 1} \big( i_1, j_1 - 1; i_2, j_2 - 1 \b| v, \lambda + 2 \eta \big) \displaystyle\frac{f \big( 2 \eta (i_2 + 1 - J)  \big) f \big( \lambda  + 2 \eta (i_2 - \Lambda) \big) }{f \big( 2 \eta (\Lambda + 1 - J) \big) f \big( \lambda \big)}. 
\end{flalign*}

\noindent Using the inductive hypothesis and the identity \eqref{quarticrelationf} one finds (in a way similar to in the proofs of Proposition \ref{wjj20} and Proposition \ref{wjj2j}) that 
\begin{flalign}
\label{svlambda3}
S = \binom{J - 1}{j_2} \binom{J}{j_2}^{-1} R; \qquad T = \binom{J - 1}{j_2 - 1} \binom{J}{j_2}^{-1} R. 
\end{flalign}

\noindent Now Theorem \ref{vlambdaspecialization} follows from \eqref{hatwwj}, \eqref{svlambda4}, and \eqref{svlambda3}.
\end{proof}

\section{Stochasticity}

\label{Stochastic}

In this section we explain how to derive the dynamical stochastic higher spin vertex model and dynamical $q$-Hahn boson model from the fused $B$-symmetric functions defined in Section \ref{EllipticSymmetric}. To that end, we first in Section \ref{StochasticFactor} explain how to use the Pieri identity \eqref{stochasticbranching} and Cauchy identity \eqref{stochasticsum} to ``stochastically correct'' the fused $B_{\mu / \nu}$ symmetric functions so that they sum to one. Then in Section \ref{DynamicalParticles} we reparameterize these stochastic weights in terms of the parameters used in Section \ref{ProbabilityMeasures} and Section \ref{LambdaSpecializationModelGeneral}; in Section \ref{PsiStochasticity} we verify that these weights are indeed stochastic. Next, in Section \ref{Transition}, we explicitly explain how the stochastically corrected, fused elliptic weight functions are related to the measures $\textbf{P}_n$ and $\textbf{M}_n$ from Definition \ref{measuresignaturesvertexmodel}.

\subsection{Stochastic Corrections} 

\label{StochasticFactor}

Recall that the $B_{\mu / \nu}$ symmetric functions from Definition \ref{definitionb} are defined as partition functions for a dynamical vertex model with specific vertex weights. Since we will eventually be interested in stochastic processes, we would like for these vertex models to be stochastic, that is, for appropriate sums of the local weights at each vertex to equal one. 

Since this is not the case for the $B_{\mu / \nu}$ symmetric functions, the purpose of this section is to introduce altered versions of these fused functions, called \textit{stochastically corrected fused weight functions}, which we denote by $B_{\mu / \nu}^{(s)}$. To that end, we proceed in a way similar to what was done by Borodin in Section 8 of \cite{ERSF}, in which he produced stochastically corrected versions of the original (unfused) weight functions. The observation there (and from previous works \cite{DP,HSVMRSF} in the non-dynamical case) was one could accomplish this using the Pieri identity \eqref{stochasticbranching}. 

To explain further, let us recall some notation. Let $\eta, \tau, \lambda \in \mathbb{C}$ with $\Im \tau > 0$; $W = (w_1, w_2, \ldots ) \subset \mathbb{C}$; $Z = (z_0, z_1, \ldots ) \in \mathbb{C}$; and $L = (\Lambda_0, \Lambda_1, \ldots) \subset \mathbb{C}$. Further denote $v_{i, j} = w_j - z_i - \eta$ for each $i, j$. Recall that $\tau$ is used to define the elliptic function $f(z)$ \eqref{theta1}; $\lambda$ is the additive dynamical parameter associated with the symmetric functions $B$ and $D$; $W$ consists of the spectral parameters associated with rows of the vertex model; and $Z$ and $L$ are associated with columns of the vertex model. 

Now let $U = (u_0, u_1, \ldots , u_r) \subset \mathbb{C}$ be a set of complex numbers. In view of the branching identity \eqref{branching} and the Pieri identity \eqref{stochasticbranching}, if one defines
\begin{flalign}
\label{bstochasticu}
B_{\mu / \nu}^{(s; U)} \big( w_1, w_2, \ldots , w_k \b| \lambda \big) = C B_{\mu / \nu} \big( w_1, w_2, \ldots , w_k \b| \lambda \big) \displaystyle\frac{D_{\mu}^{(\text{n})} \big( U \b| \lambda - 2 \eta r \big) }{D_{\nu}^{(\text{n})} \big( U \b| \lambda - 2 \eta (r - k)\big)}, 
\end{flalign}

\noindent for some suitable constant $C$, dependent on all complex parameters but independent of the signatures $\mu$ and $\nu$, then $\sum_{\mu \in \Sign^+} B_{\mu / \nu}^{(s; U)} \big( W \b| \lambda \big) = 1$ (if all parameters satisfy the assumptions of Proposition \ref{bdidentities}). Thus, replacing $B_{\mu / \nu}$ with $B_{\mu / \nu}^{(s; U)}$ provides some way of producing a ``globally stochastic'' weight function. 

However, we would also like to ensure some form of local stochasticity. Specifically, we would like for $B_{\mu / \nu}^{(s; U)}$ to have an interpretation as a partition function for a vertex model with local weights that are stochastic. At the moment, this is not guaranteed since in general the $D$ symmetric functions are often quite ``non-local,'' in the sense that they do not factor.

In Section 8 of \cite{ERSF}, an exception to this phenomenon was found in the \textit{trigonometric degeneration} of the original elliptic setting, obtained by letting $\tau$ tend to $\textbf{i} \infty$. Thus, for the remainder of this section, we assume that we are in the trigonometric setting and that the limit $\tau \rightarrow \textbf{i} \infty$ has been taken. 

Under this assumption, a specialization $\rho = (u, u + 2 \eta, \ldots )$ was defined in \cite{ERSF} (based on a non-dynamical analog from Section 6.3 of \cite{HSVMRSF}), for some specific value of $u \in \mathbb{C}$, at which $D_{\mu} \big( \rho \b| \lambda \big)$ factors completely and is particularly manageable. We do not require the explicit definition of $\rho$, but the value of $D_{\mu}^{(\text{n})}$ at that specialization is given as follows. 

\begin{lem}[{\cite[Proposition 8.2]{ERSF}}]

\label{drho}

Fix a positive integer $M$, and let $\mu = (\mu_1, \mu_2, \ldots , \mu_M) \in \Sign_M^+$. If $\mu_M = 0$, then $D_{\mu}^{(\text{n})} \big( \rho \b| \lambda) = 0$. Otherwise, 
\begin{flalign*}
D_{\mu}^{(\text{\emph{n}})} \big( \rho \b| \lambda \big) & = \displaystyle\frac{\big( - f (2 \eta) \big)^M }{\pi^M c_{\mu} (\lambda)} \displaystyle\prod_{k = 0}^{M - 1} \displaystyle\frac{f (\lambda + 2 \eta \big( j + 1 - \Lambda_0) \big)}{f (\lambda + 2 \eta j)}, 
\end{flalign*}

\noindent where 
\begin{flalign}
\label{cmu} 
\begin{aligned}
c_{\mu} (\lambda) & = \displaystyle\prod_{i = 0}^{\infty} \displaystyle\prod_{j = 0}^{m_i - 1} \displaystyle\frac{f \big( \lambda + 2 \eta ( 2 m_{[0, i)} + m_i + j - \Lambda_{[0, i]} ) \big) f \big( \lambda + 2 \eta (2 m_{[0, i)} + j + 1 - \Lambda_{[0, i)} \big) }{f \big( 2 \eta (\Lambda_i - j) \big)} \\
& \qquad \times \pi^{-M} f (2 \eta)^M \displaystyle\prod_{j = 0}^{M - 1} f (\lambda + 2 \eta j)^{-1},
\end{aligned}
\end{flalign} 

\noindent and we have denoted $m_i = m_i (\mu)$; $m_{\mathcal{I}} = \sum_{k \in \mathcal{I}} m_k$; and $\Lambda_{\mathcal{I}} = \sum_{k \in \mathcal{I}} \Lambda_k$ for any finite $\mathcal{I} \subset \mathbb{Z}_{\ge 0}$. 
\end{lem}

Replacing $U$ in \eqref{bstochasticu} by the specialization $\rho$ gives rise to the a stochastic modification $B_{\mu / \nu}^{(s)} = B_{\mu / \nu}^{(s; \rho)}$ of $B_{\mu / \nu}$ that appeared as Definition 8.4 of \cite{ERSF}, given explicitly by 
\begin{flalign}
\label{bsymmetrics}
\begin{aligned}	
B_{\mu / \nu}^{(s)} \big( u_1, u_2, \ldots , u_r \b| \lambda \big) & =  \left( - \displaystyle\frac{1}{f (2 \eta)} \right)^r \left( \displaystyle\prod_{j = 1}^r \displaystyle\frac{f (u_j - q_0) f \big( \lambda + 2 \eta (j - 1) \big) }{f (u_j - p_0)} \right)  \displaystyle\frac{D_{\mu}^{(\text{n})} \big( \rho \b| \lambda \big)}{D_{\nu}^{(\text{n})} \big( \rho \b| \lambda + 2 \eta r \big)} \\
& \qquad \times B_{\mu / \nu} \big( u_1, u_2, \ldots , u_r \b| \lambda \big),
\end{aligned}
\end{flalign}

\noindent where we have denoted 
\begin{flalign}
\label{piqi}
	 p_0 = z + \eta (1 - \Lambda_0); \qquad q_0 = z + \eta (1 + \Lambda_0). 
	 \end{flalign}

Now, we would like to express $B_{\mu / \nu}^{(s)}$ as the partition function of a dynamical model with explicit, stochastic vertex weights. These weights, denoted $\sigma_J$ below, will be (a gauge transformation of) the product of the original $W_J$ fused weights from Definition \ref{fusedweightdefinition} and a stochastic correction. Definition \ref{cstochastic} below introduces these stochastic corrections and stochastic vertex weights, and Proposition \ref{stochasticbsum} establishes that the $B_{\mu / \nu}^{(s)}$ can be expressed as the partition function of a vertex model whose weights are given by the $\sigma_J$. The fact that these weights are indeed stochastic follows from Lemma \ref{lpsiweights} and Proposition \ref{psisumstochastic} below.

\begin{definition}

\label{cstochastic}

For any $i_1, j_1, i_2, j_2 \in \mathbb{Z}_{\ge 0}$, define the \emph{stochastic correction} $C \big( i_1, j_1; i_2, j_2 \b| \lambda \big) = C_J \big( i_1, j_1; i_2, j_2 \b| \lambda, \Lambda \big)$ through 
\begin{flalign*}
C \big( i_1, j_1; i_2, j_2 | \lambda \big) & = \displaystyle\frac{ f (2 \eta)^{j_2 - j_1} \big[ 2 \eta \Lambda \big]_{i_2}}{\big[ 2 \eta \Lambda]_{i_1}} \displaystyle\frac{ \big[ \lambda + 2 \eta (i_1 + 2 j_1 - J) \big]_{j_2} \big[ \lambda + 2 \eta (i_1 + 2 j_1 - j_2 - 1 - \Lambda)\big]_{J - j_2}}{\big[ \lambda + 2 \eta (i_1 + j_1 - j_2) \big]_{J - j_1} \big[ \lambda + 2 \eta (i_1 + 2 j_1 - j_2 - 1 - \Lambda) \big]_{j_1}} \\
& \quad \times \displaystyle\frac{ \big[ \lambda + 2 \eta j_1 \big]_{J - j_1} \big[ \lambda + 2 \eta (2 j_1 - J - 1) \big]_{j_1} }{\big[ \lambda + 2 \eta (2i_1 + 2j_1 - j_2 - \Lambda) \big]_{j_2} \big[ \lambda + 2 \eta (2 i_1 + 2 j_1 - 2 j_2 - 1 - \Lambda) \big]_{J - j_2}}. 
\end{flalign*}

\noindent Furthermore, recalling the fused vertex weights $W_J (i_1, j_1; i_2, j_2)$ from Definition \ref{fusedweightdefinition}, define the \emph{stochastic vertex weights} $\sigma_J ( i_1, j_1; i_2, j_2 ) = \sigma_J \big( i_1, j_1; i_2, j_2 \b| v, \lambda \big) = \sigma_J \big( i_1, j_1; i_2, j_2 \b| v, \lambda, \Lambda \big)$ through 
\begin{flalign}
\label{sigmajdefinition}
\begin{aligned}
\sigma_J ( i_1, j_1; i_2, j_2) & = C \big( i_1, j_1; i_2, j_2 \b| \lambda \big) W_J \big( i_1, j_1; i_2, j_2 \b| v, \lambda \big) \\
& \qquad \times \displaystyle\frac{\big[ 2 \eta J \big]_J}{\big[ 2 \eta j_2 \big]_{j_2} \big[ 2 \eta (J - j_2) \big]_{J - j_2}} \left( \displaystyle\frac{\big[ 2 \eta J \big]_J}{\big[ 2 \eta j_1 \big]_{j_1} \big[ 2 \eta (J - j_1) \big]_{J - j_1}}\right)^{-1}.  
\end{aligned} 
\end{flalign}

\end{definition}

\begin{prop}

\label{stochasticbsum}

Let $M$, $N$, and $r$ be positive integers, with $M \ge N + r$. Let $\mu \in \Sign_M^+$ and $\nu \in \Sign_N^+$, and let $J_1, J_2, \ldots , J_r$ be positive integers such that $\sum_{i = 1}^r J_i = M - N$. Assume that $\nu_j \le \mu_j$ for each $j \in [1, N]$ and that all parts of $\mu$ and $\nu$ are greater than $0$. Let $w_1, w_2, \ldots , w_r \in \mathbb{C}$, and let $\varpi_k = \big( w_k, w_k + 2 \eta, \ldots , w_k + 2 \eta (J_k - 1) \big)$ for each $k \in [1, r]$. 

Then $B_{\mu / \nu}^{(s)} (\{ \varpi_1 \}, \{ \varpi_2 \}, \ldots , \{ \varpi_r \} \b| \lambda)$ can be evaluated as the partition function of the (not necessarily unfused) $r$-row vertex model, as defined in part 1 of Proposition \ref{bdfunctionsfused}, but whose vertex weights $W_J$ are altered as follows. At every vertex with positive $x$-coordinate in the $i$th row (for each $i \in [1, r]$), replace the $W_{J_i}$ weights from Proposition \ref{bdfunctionsfused} with the stochastic weights $\sigma_{J_i}$ with the same arguments; moreover, set the weight at each coordinate $(0, i)$ on the $y$-axis equal to $1$. 
	\end{prop}

\begin{proof}
	
The proof of this proposition is similar to that of Theorem 8.6 of \cite{ERSF}. 

In view of the branching property \eqref{branching}, it suffices to assume $r = 1$. Therefore, we set $r = 1$, $J_1 = J$, $w_1 = u$, and $\varpi = (u_1, u_2, \ldots , u_J) = \big( u, u + 2 \eta, \ldots , u + 2 \eta (J - 1) \big)$; further denote $v_k = u_k - z_0 - \eta$ for each $k \in [1, J]$.

In particular, this implies that the vertex model under consideration has a single row and that there is only one path ensemble satisfying the properties listed in the first part of Proposition \ref{bdfunctionsfused}. Hence, $B_{\mu / \nu} \big( \varpi \b| \lambda \big)$ is expressible as a product of the $W_J$ weights from Definition \ref{fusedweightdefinition}, and the proposition claims that it is possible to express $B_{\mu / \nu}^{(s)} \big( \varpi \b| \lambda \big)$ as the same product but using the altered weights $\sigma_J (i_1, j_1; i_2, j_2)$ in the positive quadrant and weight equal to $1$ at $(0, 1)$. 

Thus, to establish the proposition, it suffices to show (using \eqref{bsymmetrics}) that the ratio 
\begin{flalign}
\label{bsymmetrics2}
\displaystyle\frac{B_{\mu / \nu}^{(s)} \big( \varpi \b| \lambda \big)}{B_{\mu / \nu} \big( \varpi \b| \lambda \big)} = \left( - \displaystyle\frac{1}{f (2 \eta)} \right)^J \displaystyle\frac{D_{\mu}^{(\text{n})} \big( \rho \b| \lambda \big)}{D_{\nu}^{(\text{n})} \big( \rho \b| \lambda + 2 \eta J \big)} \displaystyle\prod_{k = 1}^J \displaystyle\frac{f(u_k - p_0) f \big( \lambda + 2 \eta (k - 1) \big)}{f (u_k - q_0)} 
\end{flalign}

\noindent is equal to again an analogous product of weights, but where we now replace the $\sigma_J$ weights with the stochastic corrections $C \big( i_1, j_1; i_2, j_2 \b| \lambda \big)$, and also use the weight $W_J \big( 0, J; 0, J \b| v_1, \lambda \big)^{-1} = \prod_{k = 0}^{J - 1} W_1 \big( 0, 1; 0, 1 \b| v_{J - k}, \lambda + 2 \eta k \big)^{-1}$ at the vertex $(0, 1)$; observe here that the elliptic binomial factors appearing in the definition of $\sigma_J$ (see Definition \ref{cstochastic}) multiply to $1$ and can thus be ignored. 

We establish this statement by induction on $|\mu|$ (fixing $J$). We must first verify the initial case, which is when $\mu = (1, 1, \ldots , 1)$ and $\nu$ is empty. By \eqref{bsymmetrics2} and Lemma \ref{drho}, we have that 
\begin{flalign*}
\displaystyle\frac{B_{\mu / \nu}^{(s)} \big( \varpi \b| \lambda \big)}{B_{\mu / \nu} \big( \varpi \b| \lambda \big)} & = \left( \displaystyle\prod_{k = 1}^J \displaystyle\frac{f (u_k - q_0) f \big( \lambda + 2 \eta (k - 1) \big)}{f (u_k - p_0)}  \right) \displaystyle\frac{1}{f (2 \eta)^J} \displaystyle\prod_{k = 0}^{J - 1} \displaystyle\frac{f \big( 2 \eta (\Lambda_1 - k) \big)}{f \big( \lambda + 2 \eta (J + j - \Lambda_0 - \Lambda_1) \big) }. 
\end{flalign*}

\noindent We must verify that this is equal to 
\begin{flalign*}
C \big( 0, J; J, 0 & \b| \lambda, \Lambda_0 \big) \displaystyle\prod_{k = 0}^{J - 1} \displaystyle\frac{1}{W_1 \big( 0, 1; 0, 1 \b| v_{J - k}, \lambda + 2 \eta k \big)} \\
& = \displaystyle\frac{ \big[ 2 \eta \Lambda_1 \big]_J \big[ \lambda + 2 \eta (J - 1 - \Lambda_0) \big]_J }{f (2 \eta)^J  \big[ \lambda + 2 \eta (2 J - 1 - \Lambda_0 - \Lambda_1) \big]_J} \left( \displaystyle\prod_{k = 1}^J \displaystyle\frac{f(u_k - q_0)}{f (u_k - p_0)} \displaystyle\frac{f \big( \lambda + 2 \eta (k - 1) \big) }{f \big( \lambda + 2 \eta (k - 1 - \Lambda_0) \big)} \right), 
\end{flalign*}

\noindent and it is quickly seen that the two are indeed equal. 

To proceed, we now assume that the statement of the proposition is valid for all signatures $\mu \in \Sign_M^+$ and $\nu \in \Sign_N^+$ such that $|\mu| \le S - 1$, for some integer $S > J$. It remains to show that the proposition is valid for any $\mu \in \Sign_M^+$ and $\nu \in \Sign_N^+$ such that $|\mu| = S$. 

There are two cases to consider. The first is when $\mu_1 = \nu_1$, that is the largest part of $\mu$ is equal to the largest element in $\nu$; the second is when $\mu_1 > \nu_1$.	

We will define two signatures $\widetilde{\mu}$ and $\widetilde{\nu}$ depending on these two cases. In the first case, we set $\widetilde{\mu} \in \Sign_{M - 1}^+$ equal to $\mu$ with $\mu_1$ removed, and we set $\widetilde{\nu} \in \Sign_{N - 1}^+$ equal to $\nu$ with $\nu_1$ removed. In the second case, we set $\widetilde{\mu} \in \Sign_M^+$ equal to $\mu$, with $\mu_1$ removed and $\mu_1 - 1$ appended; we also set $\widetilde{\nu} \in \Sign_N^+$ equal to $\widetilde{\nu}$. These two cases are depicted in Figure \ref{twoarrows}. 

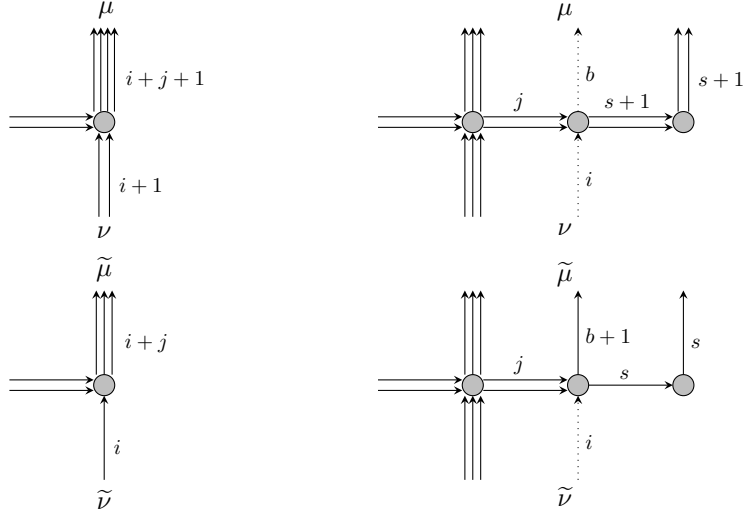
\begin{figure}

\begin{center}

\begin{tikzpicture}[
      >=stealth,
			scale = .7	
			]

			\draw[->, black] (-.195, .2) -- (-.195, 1.8);
			\draw[->, black] (-.065, .2) -- (-.065, 1.8) node[right = 2, above = 0]{$\mu$};
			\draw[->, black] (.065, .2) -- (.065, 1.8) node[scale = .8, below = 25, right = 5]{$i + j + 1$};
			\draw[->, black] (.195, .2) -- (.195, 1.8);
			
			\draw[->, black] (-1.8, -.1) -- (-.2, -.1) node[right =4, below = 34]{$\nu$};
			\draw[->, black] (-1.8, .1) -- (-.2	, .1);
			\draw[->, black] (-.1, -1.8) -- (-.1, -.2) node[scale = .8, below = 25, right = 6]{$i + 1$};
			\draw[->, black] (.1, -1.8) -- (.1, -.2);

			\draw[->, black] (-.15, -4.8) -- (-.15, -3.2) node[scale = .8, below = 25, right = 10]{$i + j$};
			\draw[->, black] (0, -4.8) -- (0, -3.2) node[above = 0]{$\widetilde{\mu}$};
			\draw[->, black] (.15, -4.8) -- (.15, -3.2);
			
			\draw[->, black] (-1.8, -5.1) -- (-.2, -5.1) node[right = 4, below = 34]{$\widetilde{\nu}$};
			\draw[->, black] (-1.8, -4.9) -- (-.2, -4.9) node[scale = .8, below = 32, right = 6]{$i$};
			
			\draw[->, black] (0, -6.8) -- (0, -5.2);

			\draw[->, black] (5.2, -.1) -- (6.8, -.1);
			\draw[->, black] (5.2, .1) -- (6.8, .1);
			\draw[->, black] (7, .2) -- (7, 1.8);
			\draw[->, black] (6.85, .2) -- (6.85, 1.8);
			\draw[->, black] (7.15, .2) -- (7.15, 1.8);
			\draw[->, black] (7, -1.8) -- (7, -.2);
			\draw[->, black] (6.85, -1.8) -- (6.85, -.2);
			\draw[->, black] (7.15, -1.8) -- (7.15, -.2);
			
			\draw[->, black] (7.2, -.1) -- (8.8, -.1); 
			\draw[->, black] (9.2, -.1) -- (10.8, -.1); 
			\draw[->, black] (11.1, .2) -- (11.1, 1.8);
			\draw[->, black, dotted] (9, .2) -- (9, 1.8) node[scale = .8, below = 22, right = 0]{$b$}; 
			\draw[->, black, dotted] (9, -1.8) -- (9, -.2) node[scale = .8, below = 22, right = 0]{$i$}; 
			\draw[->, black] (7.2, .1) -- (8.8, .1) node[scale = .8, above = 7, left = 17]{$j$}; 
			\draw[->, black] (9.2, .1) -- (10.8, .1) node[scale = .8, above = 7, left = 7]{$s + 1$}; 
			\draw[->, black] (10.9, .2) -- (10.9, 1.8) node[scale = .8, below = 25, right = 6]{$s + 1$}; 
			
			\draw[->, black] (5.2, -5.1) -- (6.8, -5.1);
			\draw[->, black] (5.2, -4.9) -- (6.8, -4.9);
			\draw[->, black] (7, -4.8) -- (7, -3.2);
			\draw[->, black] (6.85, -4.8) -- (6.85, -3.2);
			\draw[->, black] (7.15, -4.8) -- (7.15, -3.2);
			\draw[->, black] (7, -6.8) -- (7, -5.2);
			\draw[->, black] (6.85, -6.8) -- (6.85, -5.2);
			\draw[->, black] (7.15, -6.8) -- (7.15, -5.2);
			
			\draw[->, black] (7.2, -5.1) -- (8.8, -5.1);  
			\draw[->, black] (9, -4.8) -- (9, -3.2) node[left = 5, above = -2]{$\widetilde{\mu}$} node[scale = .8, below = 22, right = 0]{$b + 1$}; 
			\draw[->, black] (7.2, -4.9) -- (8.8, -4.9) node[scale = .8, above = 7, left = 17]{$j$}; 
			\draw[->, black, dotted] (9, -6.8) -- (9, -5.2) node[scale = .8, below = 22, right = 0]{$i$}; 
			\draw[->, black] (9.2, -5) -- (10.8, -5) node[scale = .8, above = 5, left = 17]{$s$}; 
			\draw[->, black] (11, -4.8) -- (11, -3.2) node[scale = .8, below = 25, right = 0]{$s$};

			\filldraw[fill=gray!50!white, draw=black] (0, 0) circle [radius=.2];
			\filldraw[fill=gray!50!white, draw=black] (0, -5) circle [radius=.2];
			
			\filldraw[fill=gray!50!white, draw=black] (7, 0) circle [radius=.2];
			\filldraw[fill=gray!50!white, draw=black] (9, 0) circle [radius=.2]  node[left = 5, above = 35] {$\mu$} node[left = 5, below = 35]{$\nu$};
			\filldraw[fill=gray!50!white, draw=black] (11, 0) circle [radius=.2];
			
			\filldraw[fill=gray!50!white, draw=black] (7, -5) circle [radius=.2];
			\filldraw[fill=gray!50!white, draw=black] (9, -5) circle [radius=.2] node[left = 5, below = 35]{$\widetilde{\nu}$};
			\filldraw[fill=gray!50!white, draw=black] (11, -5) circle [radius=.2];

\end{tikzpicture}

\end{center}

\caption{ \label{twoarrows} Depicted above are the arrow configurations of the rightmost vertices in the path ensembles corresponding to $\mu / \nu$ and $\widetilde{\mu} / \widetilde{\nu}$. The top left and bottom left diagrams correspond to the case $\mu_1 = \nu_1$, and the top right and bottom right diagrams correspond to the case $\mu_1 > \nu_1$. }
\end{figure}

Observe in either case that $|\widetilde{\mu}| < |\mu| = S$; thus, the statement of the proposition applies to the pair $(\widetilde{\mu}, \widetilde{\nu})$. This implies that we can establish the proposition for $(\mu, \nu)$ by comparing to $(\widetilde{\mu}, \widetilde{\nu})$. 

In particular, it suffices to show that the ratio of the products of weights (given by the stochastic corrections $C_J$) corresponding to the vertex models for $\mu / \nu$ and $\widetilde{\mu} / \widetilde{\nu}$ is equal to 
\begin{flalign}
\label{bstochasticb}
\begin{aligned}
\displaystyle\frac{B_{\mu / \nu}^{(s)} \big( \varpi \b| \lambda \big)}{B_{\mu / \nu} \big( \varpi \b| \lambda \big)} \displaystyle\frac{B_{\widetilde{\mu} / \widetilde{\nu} } \big( \varpi \b| \lambda \big)}{B_{\widetilde{\mu} / \widetilde{\nu}}^{(s)} \big( \varpi \b| \lambda \big)} = \displaystyle\frac{D_{\mu}^{(\text{n})} \big( \rho \b| \lambda \big)}{D_{\nu}^{(\text{n})} \big( \rho \b| \lambda + 2 \eta J \big)} \displaystyle\frac{D_{\widetilde{\nu}}^{(\text{n})} \big( \rho \b| \lambda + 2 \eta J \big)}{D_{\widetilde{\mu}}^{(\text{n})} \big( \rho \b| \lambda \big)} = \displaystyle\frac{c_{\nu} \big( \rho \b| \lambda + 2 \eta J \big)}{c_{\mu} \big( \rho \b| \lambda \big)} \displaystyle\frac{c_{\widetilde{\mu}} \big( \rho \b| \lambda \big)}{c_{\widetilde{\nu}} \big( \rho \b| \lambda + 2 \eta J \big)},
\end{aligned} 
\end{flalign}

\noindent where we have used \eqref{bsymmetrics2} and Lemma \ref{drho}. 

To do this, we analyze each of the two cases individually. Let us address the former first. 

Thus assume that $\mu_1 = \nu_1 = R$, and denote $m_R (\mu) = i + j + 1$ and $m_R (\nu) = i + 1$; then, $m_R (\widetilde{\mu}) = i + j$, and $m_R (\widetilde{\nu}) = i$. On the left side of Figure \ref{twoarrows}, we have that $i = 1$ and $j = 2$. Denote $\widehat{\lambda} = \lambda + 2 \eta \big( 2 m_{[0, R)} - \Lambda_{[0, R)} \big)$ to be the value of the dynamical parameter at the vertex $(R, 1)$. 

Since $\mu$ and $\widetilde{\mu}$, and $\nu$ and $\widetilde{\nu}$ only differ in their value of $m_R$, the partition function (with respect to the weights $C_J$) of the vertex model associated with the shape $\mu / \nu$ divided by analogous partition function associated with $\widetilde{\mu} / \widetilde{\nu}$ is equal to 
\begin{flalign}
\label{cij1}
\begin{aligned}
& \displaystyle\frac{C \big( i + 1, j; i + j + 1, 0 \b| \widehat{\lambda} \big)}{C \big( i, j; i + j, 0 \b| \widehat{\lambda} \big)} \\
& \qquad \quad = \displaystyle\frac{f \big( \widehat{\lambda} + 2 \eta (2i + 2j - J + 1 - \Lambda_R) \big) f \big( \widehat{\lambda} + 2 \eta (2i + 2j - J - \Lambda_R) \big) }{f \big( \widehat{\lambda} + 2 \eta (2i + 2j + 1 - \Lambda_R) \big) f \big( \widehat{\lambda} + 2 \eta (2i + 2j - \Lambda_R) \big)} \displaystyle\frac{f \big( 2 \eta (\Lambda_R - i - j) \big) }{ f \big( 2 \eta (\Lambda_R - i) \big)} \\
& \qquad \qquad \quad \times \displaystyle\frac{f \big( \widehat{\lambda} + 2 \eta (i + 2j - J + 1) \big) f \big( \widehat{\lambda} + 2 \eta (i + j - \Lambda_R) \big)}{f \big( \widehat{\lambda} + 2 \eta (i + 2j - J - \Lambda_R) \big) f \big( \widehat{\lambda} + 2 \eta (i + j + 1) \big) }, 
\end{aligned}
\end{flalign}

\noindent which quickly follows from Definition \ref{cstochastic} for the stochastic corrections $C_J$ and the identities \eqref{basichypergeometric1}. 

We must equate \eqref{cij1} with the right side of \eqref{bstochasticb}. To that end, observe that $m_{[0, R)} (\mu) = m_{[0, R)} (\widetilde{\mu}) = m_{[0, R)} (\nu) + J - j = m_{[0, R)} (\widetilde{\nu}) + J - j$. Inserting this into the definition \eqref{cmu}, we find  
\begin{flalign*}
\displaystyle\frac{c_{\widetilde{\mu}} (\lambda) c_{\nu} (\lambda + 2 \eta J) }{c_{\mu} (\lambda) c_{\widetilde{\nu}} (\lambda + 2 \eta J) } & = \displaystyle\prod_{k = 0}^{i + j - 1} \displaystyle\frac{f \big( \widehat{\lambda} + 2 \eta ( i + j + k - \Lambda_R ) \big) f \big( \widehat{\lambda} + 2 \eta ( k + 1 ) \big)}{f \big( 2 \eta (\Lambda_R - k) \big)} \\
& \qquad \times \displaystyle\prod_{k = 0}^i \displaystyle\frac{f \big( \widehat{\lambda} + 2 \eta ( 2j - J + i + k + 1 - \Lambda_R ) \big) f \big( \widehat{\lambda} + 2 \eta ( 2j - J + k + 1) \big)}{f \big( 2 \eta (\Lambda_R - k) \big)}  	\\
& \qquad \times \left( \displaystyle\prod_{k = 0}^{i + j} \displaystyle\frac{f \big( \widehat{\lambda} + 2 \eta ( i + j + k + 1 - \Lambda_R ) \big) f \big( \widehat{\lambda} + 2 \eta ( k + 1 ) \big)}{f \big( 2 \eta (\Lambda_R - k) \big)} \right)^{-1} \\
& \qquad \times \left( \displaystyle\prod_{k = 0}^{i - 1} \displaystyle\frac{f \big( \widehat{\lambda} + 2 \eta ( 2j - J + i + k - \Lambda_R ) \big) f \big( \widehat{\lambda} + 2 \eta ( 2j - J + k + 1) \big)}{f \big( 2 \eta (\Lambda_R - k) \big)}\right)^{-1}, 
\end{flalign*}

\noindent which is quickly verified to be equal to the right side of \eqref{cij1}. This concludes the proof of the proposition in the first case. 

We now proceed to the second case. Let $\mu_1 = R + 1$, and assume that $m_{R + 1} (\nu) = 0$. Denote $s + 1 = m_{R + 1} (\mu)$; $b = m_R (\mu)$; $i = m_R (\nu)$; $\Lambda = \Lambda_R$; and $\widetilde{\Lambda} = \Lambda_{R + 1}$. Further let $j$ denote the number of horizontal incoming arrows at the vertex $(1, R)$ in the diagram corresponding to $\mu / \nu$. In Figure \ref{twoarrows}, we have that $s = 1$, $i = b = 0$, and $j = 2$. 

Observe that $m_{R + 1} (\widetilde{\mu}) = s$, that $m_R (\widetilde{\mu}) = b + 1$, and that $i + j = s + b + 1$. Let $\widehat{\lambda} = \lambda + 2 m_{[0, R)} (\mu) - \Lambda_{[0, R)}$ denote the value of the dynamical parameter at the vertex $(R, 1)$ in the single-row vertex model corresponding to the signature $\mu$ (or $\widetilde{\mu}$).

In view of the fact that $m_R (\mu) = b$ and $m_R (\widetilde{\mu}) = b + 1$, the value of the dynamical parameter at $R + 1$ in the single-row vertex models corresponding to $\mu$ and $\widetilde{\mu}$ are $\widehat{\lambda} + 2 \eta (2b - \Lambda_R)$ and $\widehat{\lambda} + 2 \eta (2b + 2 - \Lambda_R)$, respectively. Since $\mu / \nu$ and $\widetilde{\mu} / \widetilde{\nu}$ also coincide at all vertices to the left of $R$, the partition function (with respect to the stochastic corrections $C_J$) of the single-row vertex model corresponding to $\mu / \nu$ divided by the analogous partition function of the vertex model corresponding to $\widetilde{\mu} / \widetilde{\nu}$ is equal to 
\begin{flalign} 
\label{cijcmu2}
\displaystyle\frac{C \big( a, j; b, s + 1 \b| \widehat{\lambda} \big)}{C \big( a, j; b + 1, s \b| \widehat{\lambda} \big) } \displaystyle\frac{C \big( 0, s + 1; s + 1, 0 \b| \widehat{\lambda} + 2 \eta (2b - \Lambda_R) \big)}{C \big( 0, s; s, 0 \b| \widehat{\lambda} + 2 \eta (2b + 2 - \Lambda_R) \big)};
\end{flalign}

\noindent it suffices to equate \eqref{cijcmu2} with the right side of \eqref{bstochasticb}. Similar to before, we establish this fact by explicitly evaluating \eqref{cijcmu2} and \eqref{bstochasticb}. Let us begin with the former. 

To that end, one quickly verifies using Defintion \ref{cstochastic} for the $C_J$ that 
\begin{flalign}
\label{cij2}
\displaystyle\frac{C \big( 0, s + 1; s + 1, 0 \b| \widehat{\lambda} + 2 \eta (2b - \Lambda_R) \big)}{C \big( 0, s; s, 0 \b| \widehat{\lambda} + 2 \eta (2b + 2 - \Lambda_R) \big)} & = \displaystyle\frac{f \big( 2 \eta (\Lambda_{R + 1} - s) \big)}{f (2 \eta)} \displaystyle\frac{ f \big( \widehat{\lambda} + 2 \eta (2b + s + 1 - J - \Lambda_R) \big)}{f \big( \widehat{\lambda} + 2 \eta (2b + s + 1 - \Lambda_R - \Lambda_{R + 1}) \big)},
\end{flalign}

\noindent and
\begin{flalign}
\label{cij3}
\begin{aligned}
\displaystyle\frac{C \big( a, j; b, s + 1 \b| \widehat{\lambda} \big)}{C \big( a, j; b + 1, s \b| \widehat{\lambda} \big)} & = \displaystyle\frac{f \big( \widehat{\lambda} + 2 \eta (2a + 2j - 2s - 2 - \Lambda_R) \big) f \big( \widehat{\lambda} + 2 \eta (2a + 2j - s - \Lambda_R) \big) }{ f \big( \widehat{\lambda} + 2 \eta (2a + 2j - s - J - 1 - \Lambda_R) \big) f \big( \widehat{\lambda} + 2 \eta (2a + 2j - 2s - \Lambda_R) \big)} \\
& \quad \times \displaystyle\frac{f (2 \eta)}{f \big( 2 \eta (\Lambda_R - b) \big)} \displaystyle\frac{ f \big( \widehat{\lambda} + 2 \eta (a + j - s) \big) }{ f \big( \widehat{\lambda} + 2 \eta (a + j - s - 1 - \Lambda_R) \big) }. 
\end{aligned}
\end{flalign}

\noindent Multiplying \eqref{cij2} and \eqref{cij3} yields \eqref{cijcmu2}. 

To evaluate the right side of \eqref{bstochasticb}, we recall that $\nu = \widetilde{\nu}$ and the $\mu$ and $\widetilde{\mu}$ coincide everywhere except $R$ and $R + 1$. Thus, applying the definition \eqref{cmu} of $c_{\mu}$ yields 
\begin{flalign*}
& \displaystyle\frac{c_{\widetilde{\mu}} (\lambda) c_{\nu} (\lambda + 2 \eta J) }{c_{\mu} (\lambda) c_{\widetilde{\nu}} (\lambda + 2 \eta J) } \\
& \qquad = \displaystyle\prod_{k = 0}^b \displaystyle\frac{f \big( \widehat{\lambda} + 2 \eta (b + k + 1 - \Lambda_R \big) f \big( \widehat{\lambda} + 2 \eta (k + 1) \big)}{f \big( 2 \eta (\Lambda_R - k) \big)} \\ 
& \qquad \qquad \times \displaystyle\prod_{k = 0}^{s - 1} \displaystyle\frac{f \big( \widehat{\lambda} + 2 \eta (2b + k + s + 2 - \Lambda_R - \Lambda_{R + 1}) \big) f \big( \widehat{\lambda} + 2 \eta (2b + k + 3 - \Lambda_R \big)}{f \big( 2 \eta (\Lambda_{R + 1} - k ) \big)}\\
& \qquad \qquad \times \left( \displaystyle\prod_{k = 0}^{b - 1} \displaystyle\frac{f \big( \widehat{\lambda} + 2 \eta (b + k - \Lambda_R \big) f \big( \widehat{\lambda} + 2 \eta (k + 1) \big)}{f \big( 2 \eta (\Lambda_R - k) \big)}  \right)^{-1} \\ 
& \qquad \qquad \times \left( \displaystyle\prod_{k = 0}^s \displaystyle\frac{f \big( \widehat{\lambda} + 2 \eta (2b + k + s + 1 - \Lambda_R - \Lambda_{R + 1}) \big) f \big( \widehat{\lambda} + 2 \eta (2b + k + 1 - \Lambda_R \big)}{f \big( 2 \eta (\Lambda_{R + 1} - k) \big)} \right)^{-1}. 
\end{flalign*}

\noindent Using \eqref{cij2}, \eqref{cij3}, \eqref{basichypergeometric1}, and the fact that $i + j = s + b + 1$, it is quickly seen that the above quantity is equal to \eqref{cijcmu2}. This verifies the second case and thus completes the proof of the proposition. 
\end{proof}

\subsection{Reparameterization of Weights} 

\label{DynamicalParticles}	

In Definition \ref{cstochastic}, the stochastic weights $\sigma_J$ are written in terms of the parameters $\lambda, \eta \in \mathbb{C}$ and $W, Z, L \subset \mathbb{C}$. It will be useful for us to reparameterize these weights in terms of the parameters $q, \kappa \in \mathbb{C}$ and $A, B, U, S \subset \mathbb{C}$ from Section \ref{ModelGeneral} and Section \ref{LambdaSpecializationModelGeneral}. Lemma \ref{dynamicalparticlesvlambda} does this in the $v = - \eta \Lambda$ setting, in which case they reparameterize to the $\varphi \big( j \b| i \big)$ stochastic weights for the dynamical $q$-boson model given by Definition \ref{stochasticparticlesjumps}; Lemma \ref{lpsiweights} does this for the general weights, in which case they reparameterize to the $\psi$ weights for the dynamical stochastic higher spin vertex model given by Definition \ref{psiweightdefinition}; and Lemma \ref{wjjju} does this for weights of the form $\sigma_J (i, J; i', j')$, which will be useful in Section \ref{PsiStochasticity}.

We begin with the $v = - \eta \Lambda$ case.

\begin{lem}

\label{dynamicalparticlesvlambda}

Let $i_1, j_1 \in \mathbb{Z}_{\ge 0}$, $J \in \mathbb{Z}_{\ge 1}$, and $\eta, w, z, \Lambda, \lambda, \tau \in \mathbb{C}$; denote $v = z + \eta - w$. Under the trigonometric degeneration $\tau \rightarrow \textbf{\emph{i}} \infty$ and the further specialization
\begin{flalign*}
q = e^{- 4 \pi \textbf{\emph{i}} \eta}; \quad v = w - z - \eta = - \eta \Lambda; \quad s = e^{2 \pi \textbf{\emph{i}} \eta \Lambda}; \quad \varkappa = e^{2 \pi \textbf{\emph{i}} \lambda}; \quad \widetilde{\varkappa} = q^{J - 2j_1} \varkappa,
\end{flalign*}

\noindent we have that $\sigma_J \big( i_1, j_1; i_2, j_2) = \varphi_{q, a, b, \kappa} \big( j_1 \b| i_2 \big) \textbf{1}_{i_1 + j_1 = i_2 + j_2}$, where $a = s^2 q^J$, $b = s^2$, and $\kappa = \widetilde{\varkappa}^{-1} = q^{2j_1 - J} \varkappa^{-1}$. 
\end{lem}

\begin{proof}

Denote $\widetilde{\lambda} = \lambda + 2 \eta (2 j_1 - J)$; observe that $e^{2 \pi \textbf{i} \widetilde{\lambda}} = \widetilde{\varkappa}$. Then, one quickly verifies using Theorem \ref{vlambdaspecialization} and Definition \ref{cstochastic} that 
\begin{flalign}
\label{ljetalambdav}
\begin{aligned}
\sigma_J \big( i_1, j_1; i_2, j_2 \b| -\eta \Lambda, \lambda \big) & = \displaystyle\frac{\big[ 2 \eta J \big]_J}{\big[ 2 \eta j_2 \big]_{j_2} \big[ 2 \eta (J - j_2) \big]_{J - j_2} } \displaystyle\frac{\big[ 2 \eta i_1 \big]_{j_2} \big[ 2 \eta (\Lambda - i_1) \big]_{J - j_2}}{\big[ 2 \eta \Lambda \big]_J} \\
& \qquad \times  \displaystyle\frac{\big[ \widetilde{\lambda} + 2 \eta i_1 \big]_{j_2} \big[ \widetilde{\lambda} + 2 \eta (i_1 + J - j_2 - 1 - \Lambda) \big]_{J - j_2} }{\big[ \widetilde{\lambda} + 2 \eta (2 i_1 + J - j_2 - \Lambda) \big]_{j_2} \big[ \widetilde{\lambda} + 2 \eta (2 i_1 + J - 2 j_2 - 1 - \Lambda) \big]_{J - j_2}}. 
\end{aligned}
\end{flalign}

\noindent Now, under the trigonometric degeneration $\tau \rightarrow \textbf{i} \infty$, we have that
\begin{flalign*}
f (a) = \displaystyle\frac{e^{\pi \textbf{i} a} - e^{- \pi \textbf{i} a}}{2 \textbf{i}} = \displaystyle\frac{\textbf{i}}{2 A^{1 / 2}} \big( 1 - A \big) = \displaystyle\frac{A^{1 / 2}}{2 \textbf{i}} \big( 1 - A^{-1} \big), \quad \text{where $A = e^{2 \pi \textbf{i} a}$,}
\end{flalign*}

\noindent for any $a \in \mathbb{C}$. This gives 
\begin{flalign}
\label{akhypergeometric}
[a]_k = \left( \displaystyle\frac{\textbf{i}}{2 A^{1 / 2}} \right)^k q^{- \binom{k}{2} / 2} \big( A; q \big)_k = \left( \displaystyle\frac{A^{1 / 2}}{2 \textbf{i}} \right)^k q^{\binom{k}{2} / 2} \big( q^{1 - k} A^{-1}; q \big)_k,
\end{flalign}

\noindent for any $a \in \mathbb{C}$ and $k \in \mathbb{Z}$. 

Applying \eqref{akhypergeometric} in the expression for $\sigma_J$ above yields
\begin{flalign*}
& \sigma_J \big( i_1, j_1; i_2, j_2 \b| - \eta \Lambda, v \big) \\
& = (s^2 q^J)^{j_2} \displaystyle\frac{(q^{-J}; q)_{j_2}}{(q; q)_{j_2}} \displaystyle\frac{(q^{i_1 - j_2 + 1}; q)_{j_2} (s^2 q^{i_1}; q)_{J - j_2}}{(s^2; q)_J} \displaystyle\frac{(q^{i_1 - j_2 + 1} / \widetilde{\varkappa}; q)_{j_2} (s^2 q^{i_1} / \widetilde{\varkappa}; q)_{J - j_2}}{(s^2 q^{2i_1 + J - 2j_2 + 1} / \widetilde{\varkappa}; q)_{j_2} (s^2 q^{2i_1 - j_2} / \widetilde{\varkappa}; q)_{J - j_2}} \\
& = (s^2 q^J)^{j_2} \displaystyle\frac{\textbf{1}_{j_2 \le i_1} (q; q)_{i_1}}{(q; q)_{i_1 - j_2} (q; q)_{j_2}}  \displaystyle\frac{ \ (q^{-J}; q)_{j_2} (s^2 q^J; q)_{i_1 - j_2} }{ (s^2; q)_{i_1} } \displaystyle\frac{(s^2 q^{i_1} \kappa; q)_{i_1 - j_2} (q^{i_1 - j_2 + 1} \kappa; q)_{j_2}}{(s^2 q^{i_1 + J - j_2} \kappa; q)_{i_1 - j_2} (s^2 q^{2i + J - 2j_2 + 1} \kappa; q)_{j_2}},
\end{flalign*}

\noindent where to deduce the second equality we used the first identity in \eqref{basichypergeometric1}. From Definition \ref{stochasticparticlesjumps} and the facts that $a = s^2 q^J$ and $b = s^2$, it follows that the above expression is equal to $\varphi_{q, a, b, \kappa} \big( j \b| i \big)$, from which we deduce the proposition. 
\end{proof} 

The following lemma provides the reparameterization for the general weights.

\begin{lem}
	
	\label{lpsiweights} 
	
	Let $i_1, j_1 \in \mathbb{Z}_{\ge 0}$, $J \in \mathbb{Z}_{\ge 1}$, and $\eta, w, z, \Lambda, \lambda, \tau \in \mathbb{C}$; denote $v = z + \eta - w$. Under the trigonometric degeneration $\tau \rightarrow \textbf{\emph{i}} \infty$ and the further specialization
	\begin{flalign}
	\label{qxiuskappa}
	q = e^{- 4 \pi \textbf{\emph{i}} \eta}; \quad \xi = e^{2 \pi \textbf{\emph{i}} z}; \quad u = e^{2 \pi \textbf{\emph{i}} (\eta - w)}; \quad s = e^{2 \pi \textbf{\emph{i}} \eta \Lambda}; \quad \varkappa = e^{2 \pi \textbf{\emph{i}} \lambda}; \quad \widetilde{\varkappa} = q^{J - 2j_1} \varkappa; \quad \kappa = \widetilde{\varkappa}^{-1},
	\end{flalign}
	
	\noindent we have that $\sigma_J \big( i_1, j_1; i_2, j_2 \b| v, \lambda \big) = \psi_{u; s; q; J} \big( i_1, j_1; i_2, j_2 \b| \kappa \big)$. 
\end{lem} 

\begin{proof} 
	
	Combining the explicit form of the $W_J$ weights (given by Theorem \ref{fusedweighthypergeometric}) with Definition \ref{cstochastic} for the $\sigma_J$ weights, we deduce that $\sigma_J (i_1, j_1; i_2, j_2 \b| v, \lambda)$ is equal to
	\begin{flalign*}
	& \displaystyle\frac{\big[ 2 \eta i_1\big]_{j_2}  \big[ 2 \eta (J - j_2) \big]_{j_1} \big[ \eta \Lambda - v - 2 \eta (i_1 + j_1) \big]_{J - j_1 - j_2} \big[ 2 \eta (\Lambda - i_1 + j_2)  \big]_{j_1} \big[ 2 \eta (J - j_1)  \big]_{J - j_1}}{\big[ \eta \Lambda - v \big]_J \big[ 2 \eta j_2 \big]_{j_2} \big[ 2 \eta (J - j_2) \big]_{J - j_2}} \\
	& \quad \times \displaystyle\frac{\big[ \lambda + 2 \eta i_2 \big]_{J - j_1 - j_2} \big[ \lambda - \eta \Lambda - v + 2 \eta (i_1 + 2 j_1 - J) \big]_{j_2} \big[ \lambda + v - \eta \Lambda + 2 \eta (i_1 + 2 j_1 - 1) \big]_{j_1}}{\big[ \lambda + 2 \eta (2 j_1 + j_2 - J - 1) \big]_{j_1} \big[ \lambda + 2 \eta (i_1 + j_1 - j_2) \big]_{J - j_1} \big[ \lambda + 2 \eta (i_1 + 2j_1 - j_2 - \Lambda - 1) \big]_{j_1} } \\
	& \quad \times \displaystyle\frac{\big[ \lambda + 2 \eta (i_1 + 2 j_1 - J) \big]_{j_2} \big[ \lambda + 2 \eta (i_1 + 2j_1 - j_2 - \Lambda - 1) \big]_{J - j_2} \big[ \lambda + 2 \eta (2j_1 - J - 1) \big]_{j_1}}{\big[ \lambda + 2 \eta (2i_1 + 2j_1 - j_2 - \Lambda) \big]_{j_2} \big[ \lambda + 2 \eta (2i_1 + 2j_1 - 2j_2 - \Lambda - 1)\big]_{J - j_2}} \\
	& \quad \times {_{12} v_{11}} \big( \lambda + 2 \eta (2j_1 + j_2 - J); 2 \eta j_1, 2 \eta j_2, \lambda + 2 \eta j_1, \lambda + 2 \eta (i_1 + 2j_1 - J - \Lambda - 1), \\
	& \qquad \qquad \quad \eta \Lambda + v + 2 \eta (j_2 - i_1 - 1), \eta \Lambda - v - 2 \eta (i_1 - j_2 + J), \lambda + 2 \eta (i_2 + j_1 + j_2 - J) \big). 
	\end{flalign*}
	
	\noindent Now the proof of the lemma follows from repeated application of \eqref{akhypergeometric}; further details are similar those provided in the proof of Lemma \ref{dynamicalparticlesvlambda} are therefore omitted. 
\end{proof} 

The following lemma provides a specialization of the previous result, which will be useful for us when verifying stochasticity of the $\psi$ weights in the next section. 

\begin{lem}
	
	\label{wjjju} 
	
	Let $i \in \mathbb{Z}_{\ge 0}$, $J \in \mathbb{Z}_{\ge 1}$, and $\eta, w, z, \Lambda, \lambda, \tau \in \mathbb{C}$; denote $v = z + \eta - w$. Under the trigonometric degeneration $\tau \rightarrow \textbf{\emph{i}} \infty$ and the further specialization \eqref{qxiuskappa} (in particular, $\widetilde{\varkappa} = q^{-J} \varkappa$), we have that 
	\begin{flalign}
	\label{wjjq}
	\begin{aligned}
	\sigma_J \big( i, J; i + J - j, j \b| v, \lambda \big) & = (s u \xi q^J)^j \displaystyle\frac{(q^{-J}; q)_j (s^2 q^i; q)_{J - j} (s q^{i - j + 1} / u \xi; q)_j}{(q; q)_j (s u \xi; q)_J} \\
	& \qquad \times \displaystyle\frac{(q^{i - j + 1} \kappa; q)_j (s u \xi q^i \kappa; q)_{J - j}}{(q^{2i + J - 2j + 1} s^2 \kappa; q)_j (q^{2i + J - j} s^2 \kappa; q)_{J - j}}, 
	\end{aligned}
	\end{flalign}
	
	\noindent for any $j \in \mathbb{Z}_{\ge 0}$. Furthermore, abbreviating $\sigma (j) = \sigma_J \big( i_1, J; i + J - j, j \b| v, \lambda \big)$, we have that $\sum_{j = 0}^{\infty} \sigma (j) = 1$.  
\end{lem}

\begin{proof}
	Let us begin with the first statement of the lemma. To that end, we apply Lemma \ref{wjjj} and Definition \ref{cstochastic} to find that	
	\begin{flalign*}
	\sigma_J (i, J; i + J - j, j) & = \displaystyle\frac{[2 \eta J]_J}{[2 \eta j]_j \big[ 2 \eta (J - j) \big]_{J - j}} \displaystyle\frac{\big[2 \eta (\Lambda - i) \big]_{J - j}  [2 \eta i - \eta \Lambda - v]_j}{[\eta \Lambda - v]_J } \\
	& \qquad \times \displaystyle\frac{\big[ \lambda + 2 \eta (i + J) \big]_j \big[ v + \lambda - \eta \Lambda + 2 \eta (i + 2 J - j - 1) \big]_{J - j}}{\big[ \lambda + 2 \eta (2i + 2J - j - \Lambda)]_j \big[ \lambda + 2 \eta (2i + 2 J - 2j - \Lambda - 1) \big]_{J - j}}.
	\end{flalign*}
	
	\noindent Now \eqref{wjjq} follows from repeated use of \eqref{akhypergeometric} and the first four of the identities listed in \eqref{basichypergeometric1}. 
	
	To establish the second part of the lemma, observe that combining \eqref{wjjq} with repeated use of the first six of the identities listed in \eqref{basichypergeometric1} yields that 
	\begin{flalign*}
	\sigma (j) & = \left( \displaystyle\frac{q}{u \xi s} \right)^j \displaystyle\frac{(q^{-J}; q)_j (u \xi / s q^i; q)_j}{(q; q)_j (q^{1 - i - J} / s^2; q)_j} \displaystyle\frac{(\widetilde{\varkappa} q^{-i}; q)_j (\widetilde{\varkappa} q^{-2i} / s^2 q^J; q)_j}{(\widetilde{\varkappa} / s u \xi q^{i + J - 1}; q)_j (\widetilde{\varkappa} / s^2 q^{2i - 1}; q)_j} \displaystyle\frac{1 - \widetilde{\varkappa} / s^2 q^{2i + J - 2j}}{1 - \widetilde{\varkappa} / s^2 q^{2i + J}} \\
	& \qquad \times \displaystyle\frac{(q^{1 - i - J} / s^2; q)_J (\widetilde{\varkappa} q^{1 - i - J} / u \xi s; q)_J}{(q^{1 - 2i - J} \widetilde{\varkappa} / s^2; q)_J (q^{1 - J} / u \xi s; q)_J}.
	\end{flalign*}
	
	\noindent Now the fact that $\sum_{j = 0}^{\infty} \sigma (j) = 1$ follows from the terminating Rogers identity \eqref{hypergeometric65sumterminating}, applied with $a = \widetilde{\varkappa} / s^2 q^{2i + J} $; $b = u \xi / s q^i$; $c = \widetilde{\varkappa} q^{-i}$; and $n = J$. 
\end{proof}

\subsection{Stochasticity of the \texorpdfstring{$\psi$}{} Weights}

\label{PsiStochasticity} 	

The purpose of this section is to establish the following proposition, which shows that the $\psi$ weights from Definition \ref{psiweightdefinition} are stochastic.  

\begin{prop}
	
	\label{psisumstochastic}
	
	Fix $i, j \in \mathbb{Z}_{\ge 0}$, $J \in \mathbb{Z}_{\ge 1}$, and $q, u, s , \kappa \in \mathbb{C}$. Then, 
	\begin{flalign*}
	\displaystyle\sum_{j' = 0}^{\infty} \psi_{u; s; q; J} \big( i, j; i + j - j', j' \b| \kappa \big) = 1.
	\end{flalign*}
\end{prop}

Proposition \ref{psisumstochastic} can be viewed as an identity for basic hypergeometric series, but we unfortunately do not know of a quick way to verify it directly (although it seems to resemble some of the identities given in \cite{STCBRF}). Instead, we establish Proposition \ref{psisumstochastic} through properties of the stochastic elliptic weight functions $B_{\mu / \nu}^{(s)}$. 

To begin, we require the following consequence of Lemma \ref{wjjju}, which states the the $\{ \sigma (j) \}_{j \ge 1}$ (defined in that lemma) are in some sense linearly independent as functions of $\xi$. 

\begin{cor}
	\label{linearlyindependentbeta}
	
	Fix $J \in \mathbb{Z}_{\ge 1}$, and adopt the notation of Lemma \ref{wjjju}, with $i_1$ set to $0$; in particular, recall that $\sigma (j)$ depends on complex parameters $q, \kappa, u, s, \xi$. If $A_0, A_1, A_2, \ldots , A_m \in \mathbb{C} (q, \kappa, u, s)$ satisfy $A_0 + \sum_{j = 1}^m A_j \sigma (j)= 0$ for all $q, \kappa, u, s, \xi \in \mathbb{C}$, then $A_0 = A_1 = \cdots = A_m = 0$. 
\end{cor}

\begin{proof}
	
	First, observe from \eqref{wjjq} that setting $\xi = s / u$ yields $\sigma (j) = 0$ if $j > 0$. Inserting this into the identity $A_0 + \sum_{j = 1}^m A_j \sigma (j)= 0$ gives $A_0 = 0$.   
	
	Now, assume that at least one of the $A_0, A_1, \ldots , A_m$ is nonzero, and let $k \in [1, m]$ be the smallest index such that $A_k \ne 0$. Setting $\xi = s / q^k u$, we find again from \eqref{wjjq} that $\sigma (j) = 0$ if $j > k$ and that $\sigma (k)$ is generically (in the parameters $q$, $\kappa$, $u$, and $s$) nonzero. Inserting this into the identity $A_0 + \sum_{j = 1}^m A_j \sigma (j)= 0$ yields that $A_k = 0$. This is a contradiction, which implies the corollary. 
\end{proof}

Now we can establish Proposition \ref{psisumstochastic}.

\begin{proof}[Proof of Proposition \ref{psisumstochastic}]
	
	First recall from \eqref{wjlargesmallj} that, if $\max \{ j_1, j_2 \} > J$, then $W_J (i_1, j_1; i_2, j_2) = 0$ and therefore $\sigma_J (i_1, j_1; i_2, j_2) = 0$. In what follows, we will often use that fact implicitly. 
	
	Now, use \eqref{qxiuskappa} to define parameters $\eta, w, z, \Lambda, \lambda \in \mathbb{C}$ in terms of the given parameters $q, u, \xi, s, \kappa \in \mathbb{C}$, and denote $v = z + \eta - w$. In view of Lemma \ref{lpsiweights}, it suffices to show that $\sum_{j' = 0}^J \sigma_J \big( i, j; i + j - j', j' \b| v, \lambda \big) = 1$. 
	
	To that end, we use Proposition 8.5 of \cite{ERSF}, which states that 
	\begin{flalign}
	\label{sbsum1}
	\sum_{\mu \in \Sign_{N + k}^+} B_{\mu / \nu}^{(s)} \big( W \b| \lambda + 2 \eta (\Lambda_1 + 2j - 2J ), Z, L \big) = 1, 
	\end{flalign}
	
	\noindent for any integers $j \ge 0$ and $J, N, k \ge 1$; signature $\nu \in \Sign_N^+$; and complex parameters $\lambda \in \mathbb{C}$, $W = (w_1, w_2, \ldots , w_k) \subset \mathbb{C}$, $Z = (z_0, z_1, \ldots ) \subset \mathbb{C}$, and $L = (\Lambda_0, \Lambda_1, \ldots ) \subset \mathbb{C}$, provided that the left side of \eqref{sbsum1} converges absolutely. Observe here that \eqref{sbsum1} would remain true if $\lambda + 2 \eta (\Lambda_1 + 2j - 2 J)$ were replaced by $\lambda$; our specific choice of the dynamical parameter in \eqref{sbsum1} will be useful to us later. Also denote $v_k = w - z_k - \eta$, for each $k \ge 0$. 
	
	The proof of the identity $\sum_{j' = 0}^J \sigma_J \big( i, j; i + j - j', j' \b| v, \lambda \big) = 1$ will follow from a suitable specialization of \eqref{sbsum1}. In particular, recall that Proposition \ref{stochasticbsum} implies that, for any $\mu \in \Sign_{i + J}^+$, the function $B_{\mu / \nu}^{(s)}$ can be expressed as a product of weights (each of the form $\sigma_J \big( i_1, j_1; i_2, j_2)$) corresponding to the single-row path ensemble associated with the shape $\mu / \nu$; see Figure \ref{singlerowdouble} for an example. We will pick the parameters $W$, $Z$, and $L$ such that the weight at a vertex of the form $(k, 1)$ for any $k > 2$ is either equal to $0$ or $1$; in particular, $B_{\mu / \nu}$ can be written as the product of the weights at $(1, 1)$ and $(2, 1)$. We will further choose the signature $\nu$ such that the vertex weight at $(1, 1)$ is of the form $\sigma (j)$ (from Corollary \ref{linearlyindependentbeta}) and such that the weight at $(2, 1)$ is of the form $\sigma_j (i_1, \widetilde{j_1}; i_2, \widetilde{j}_2)$ for some $\widetilde{j}_1$ and $\widetilde{j}_2$. Then, \eqref{sbsum1} will imply a that a suitable linear combination of the $\sigma (j)$ weights is equal to $0$, and we can apply Corollary \ref{linearlyindependentbeta} to conclude. 
	
	Let us implement this in more detail. To that end, let $\widetilde{z}, \widetilde{\Lambda} \in \mathbb{C}$ be arbitrary, and set $\widetilde{s} = e^{2 \pi \textbf{i} \eta \widetilde{\Lambda}}$ and $\widetilde{\xi} = e^{2 \pi \textbf{i} z}$. Next apply \eqref{sbsum1} with the following choice of parameters. 
	\begin{itemize} 
		\item Set $k = J$ and $W = (w_1, w_2, \ldots , w_J) = \big( w, w + 2 \eta, \ldots , w + 2 \eta (J - 1) \big)$. 
		
		\item Set $Z = (z_0, z_1, \ldots ) = (0, \widetilde{z}, z, 0, 0, \ldots )$. 
		
		\item Set $L = (\Lambda_0, \Lambda_1, \ldots ) = (0, \widetilde{\Lambda}, \Lambda, 1 - \eta^{-1} w, 0, 0, \ldots )$, so in particular $\eta \Lambda_3 = \eta - w = - v_3$. 
		
		\item Set $N = i$ and $\nu = (2, 2, \ldots , 2)$, where $2$ appears with multiplicity $i$.
		 
	\end{itemize} 
	
	Now let us evaluate the vertex weights associated with the diagram $\mu / \nu$.

	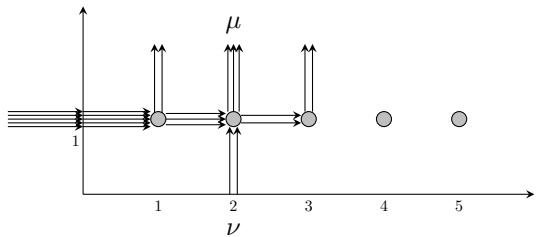
\begin{figure}
		
		\begin{center}
			
			\begin{tikzpicture}[
			>=stealth,
			scale=.5
			]
			
			\draw[->,black] (0, 0) --(12, 0);
			\draw[->,black] (0, 0) --(0, 5);
			
			\draw[->, black] (3.9, 0) -- (3.9, 1.8); 
			\draw[->, black] (4.1, 0) -- (4.1, 1.8); 
			
			\draw[->,black] (-2, 1.8) -- (0, 1.8);
			\draw[->,black] (-2, 1.9) -- (0, 1.9);
			\draw[->,black] (-2, 2) -- (0, 2);
			\draw[->,black] (-2, 2.1) -- (0, 2.1);
			\draw[->,black] (-2, 2.2) -- (0, 2.2);
			
			\draw[->,black] (0, 1.8) -- (1.8, 1.8);
			\draw[->,black] (0, 1.9) -- (1.8, 1.9);
			\draw[->,black] (0, 2) -- (1.8, 2);
			\draw[->,black] (0, 2.1) -- (1.8, 2.1);
			\draw[->,black] (0, 2.2) -- (1.8, 2.2);
			
			\draw[->,black] (2.2, 1.85) -- (3.8, 1.85);
			\draw[->,black] (2.2, 2) -- (3.8, 2);
			\draw[->,black] (2.2, 2.15) -- (3.8, 2.15);
			
			\draw[->,black] (1.9, 2) -- (1.9, 4);
			\draw[->,black] (2.1, 2) -- (2.1, 4);
			
			\draw[->, black] (3.85, 2.2) -- (3.85, 4);
			\draw[->, black] (4, 2.2) -- (4, 4);
			\draw[->, black] (4.15, 2.2) -- (4.15, 4);
			
			\draw[->,black] (4.2, 1.9) -- (5.8, 1.9);
			\draw[->,black] (4.2, 2.1) -- (5.8, 2.1);
			
			\draw[->,black] (5.9, 2.2) -- (5.9, 4);
			\draw[->,black] (6.1, 2.2) -- (6.1, 4);

			\filldraw[fill=gray!50!white, draw=black] (2, 2) circle [radius=.2] node[scale = .6, below = 48]{$1$} node[scale = .6, left = 52, below = 7]{$1$};
			\filldraw[fill=gray!50!white, draw=black] (4, 2) circle [radius=.2] node[scale = .6, below = 48]{$2$} node[below = 36] {$\nu$} node[above = 29]{$\mu$};
			\filldraw[fill=gray!50!white, draw=black] (6, 2) circle [radius=.2] node[scale = .6, below = 48]{$3$}; 
			\filldraw[fill=gray!50!white, draw=black] (8, 2) circle [radius=.2] node[scale = .6, below = 48]{$4$};
			\filldraw[fill=gray!50!white, draw=black] (10, 2) circle [radius=.2] node[scale = .6, below = 48]{$5$};

			\end{tikzpicture}

		\end{center}

		\caption{\label{singlerowdouble} Depicted above is the single-row, fused, directed path ensemble corresponding to the shape $\mu / \nu$, where $\mu = (3, 3, 2, 2, 2, 1, 1)$ and $\nu = (2, 2)$. }
	\end{figure}
	
	First, consider the weight at the vertex $(1, 3)$. Since $\eta \Lambda_3 + v_3 = 0$, the stochastic weight $\sigma_J (i_1, j_1; i_2, j_2)$ is given explicitly by \eqref{ljetalambdav}. Due to the factor of $[2 \eta i_1]_{j_2}$ in that identity, and the fact that $i_1 = 0$ at $(3, 1)$ (since no part of $\nu$ is equal to $3$), we have that $\sigma_J (i_1, j_1; i_2, j_2) = 0$ unless $j_2 = 0$. This implies that, if $\mu_1 > 3$, then $B_{\mu / \nu}^{(s)} (W)$ is equal to $0$ and thus does not contribute to the sum \eqref{sbsum1}. Hence, we will assume in what follows that $\mu_1 \le 3$. 
	
	Moreover, again using \eqref{ljetalambdav}, we find that $\sigma_J (0, j_1; j_1, 0) = 1$. Therefore, $B_{\mu / \nu}^{(s)} \big( W \big)$ is equal to the product of the $\sigma_J$ weights at $(1, 1)$ and $(1, 2)$. Let $m_1 = m_1 (\mu)$ and $m_2 = m_2 (\mu)$. Then, the dynamical parameter at $(1, 1)$ is equal to $\lambda + 2 \eta (\widetilde{\Lambda} + 2j - 2J)$ (this comes from the choice of dynamical parameter in \eqref{sbsum1} and the fact that $\Lambda_0 = 0$), and the dynamical parameter at $(2, 1)$ is equal to $\lambda + 4 \eta (m_1 + j - J)$. Hence, using the fact that $v_1 = w - z_1 - \eta = w - \widetilde{z}- \eta$, we obtain that 
	\begin{flalign}
	\label{bljlj}
	\begin{aligned}
	B_{\mu / \nu}^{(s)} \big( & W \b| \lambda + 2 \eta (\widetilde{\Lambda} + 2 j - 2J), Z, L \big) \\ 
	& = \sigma_J \big( 0, J; m_1, J - m_1 \b| w - \widetilde{z} - \eta, \lambda + 2 \eta (\widetilde{\Lambda} + 2j - 2J), \widetilde{\Lambda} \big) \\
	& \quad \times \sigma_J \big( i, J - m_1; m_2, i + J - m_1 - m_2 \b| v, \lambda + 4 \eta (m_1 + j - J), \Lambda \big). 
	\end{aligned}
	\end{flalign} 
	
	Similar to in the statement of Lemma \ref{wjjju}, let us abbreviate $\sigma (J - m_1) = \sigma_J \big( 0, J; m_1, J - m_1 \b| w - \widetilde{z} - \eta, \lambda + 2 \eta (\widetilde{\Lambda} + 2j - 2J), \widetilde{\Lambda} \big)$. Now, \eqref{sbsum1} becomes 
	\begin{flalign*}
	\displaystyle\sum_{m_1 = 0}^J \sigma (J - m_1) \displaystyle\sum_{m_2 = 0}^J \sigma_J \big( i, J - m_1; m_2, i + J - m_1 - m_2 \b| v, \lambda + 4 \eta (m_1 + j - J), \Lambda \big) = 1. 
	\end{flalign*}
	
	\noindent Changing variables $k = J - m_1$ and using the fact that $\sum_{j = 0}^{\infty} \sigma (j) = 1$ (due to the second statement in Lemma \ref{wjjju}), this can be equivalently rewritten as $A_0 + \sum_{k = 1}^J A_k \sigma (k) = 0$, where 
	\begin{flalign*}
	A_0 = \displaystyle\sum_{m_2 = 0}^J \sigma_J \big( i, 0; m_2, i - m_2 \b| v, \lambda + 4 \eta j, \Lambda \big) - 1,
	\end{flalign*}
	
	\noindent and 
	\begin{flalign*}
	A_k = \displaystyle\sum_{m_2 = 0}^J \sigma_J \big( i, k; m_2, i + k - m_2 \b| v, \lambda + 4 \eta (j - k), \Lambda \big) - A_0 - 1. 
	\end{flalign*}
	
	\noindent Observe that, for fixed $\Lambda$ and $z$, we find from Lemma \ref{lpsiweights} that $A_k \in \mathbb{C} (q, \kappa, u, \widetilde{s})$ for each $k \ge 0$. Therefore, since  $A_0 + \sum_{k = 1}^J A_k \sigma (k) = 0$ holds for all $\widetilde{\xi} = e^{2 \pi \textbf{i} \widetilde{z}}$, Corollary \ref{linearlyindependentbeta} gives that $A_k = 0$ for all $k \ge 0$. In particular, combining $A_0 = 0$ and $A_j = 0$ (and relabeling $j' = i + j - m_2$) yields $\sum_{j' = 0}^J \sigma_J (i, j; i + j - j', j' \b| v, \lambda) = 1$, which implies the proposition. 
 	\end{proof}

\subsection{Transition Functions} 

\label{Transition}

Now let us explain the connection between the $B_{\mu / \nu}^{(s)}$ stochastically corrected fused symmetric functions and the measures $\textbf{P}_n$ and $\textbf{M}_n$ (from Definition \ref{measuresignaturesvertexmodel}) prescribing the probability distribution for the dynamical stochastic higher spin vertex model and dynamical $q$-Hahn boson model.

In what follows, we assume that $q, \delta \in \mathbb{C}$; $U = (u_1, u_2, \ldots ) \subset \mathbb{C}$; $\Xi = (\xi_1, \xi_2, \ldots ) \subset \mathbb{C}$, and $S = (s_1, s_2, \ldots ) \subset \mathbb{C}$ are chosen as in Section \ref{ProbabilityMeasures}, and that $B = (b_1, b_2, \ldots ) \subset \mathbb{C}$; $C = (c_1, c_2, \ldots ) \subset \mathbb{C}$; and $a_{x, y} = b_x c_y \in \mathbb{C}$ are chosen as in Section \ref{LambdaSpecializationModelGeneral}. 

The following proposition follows quickly from the definition of the dynamical stochastic higher spin vertex model $\mathcal{P}$ from Section \ref{ProbabilityMeasures}; Definition \ref{measuresignaturesvertexmodel} for $\textbf{P}_n$; Proposition \ref{stochasticbsum}; and Lemma \ref{lpsiweights}. Observe that the fact that $\kappa = \widetilde{\varkappa}^{-1}$ (instead of $\kappa = \varkappa^{-1}$) corresponds to the fact that the $\psi$ weights on the right side of \eqref{pndynamicalstochasticspin} is taken at the dynamical parameter $\kappa_{x, y - 1}$ (instead of at $\kappa_{x, y}$). 

\begin{prop}
	
	\label{measuremmodel} 
	
	Let $\eta, \tau, \lambda \in \mathbb{C}$, $W = (w_1, w_2, \ldots ) \subset \mathbb{C}$; $Z = (z_0, z_1, \ldots ) \subset \mathbb{C} $; $L = (\Lambda_0, \Lambda_1, \ldots) \subset \mathbb{C}$, with $\Lambda_0 = 0$; and $\mathscr{J} = (J_1, J_2, \ldots ) \subset \mathbb{Z}_{\ge 0}$. Also denote $\varpi_i = \big( w_i, w_i + 2 \eta, \ldots , w_i + 2 \eta (J_i - 1) \big)$ for each $i \ge 1$. Assume (for all integers $j \ge 1$) that
	\begin{flalign}
	\label{qusdelta}
	q = e^{- 4 \pi \textbf{\emph{i}} \eta}; \quad u_j = e^{	2 \pi \textbf{\emph{i}} (\eta - w_j)}; \quad \xi_j = e^{2 \pi \textbf{\emph{i}} z_j}; \quad s_j = e^{2 \pi \textbf{\emph{i}} \Lambda_j}; \quad \delta = e^{-2 \pi \textbf{\emph{i}} \lambda}. 
	\end{flalign}
	
	\noindent Then, under the trigonometric degeneration $\tau \rightarrow \textbf{\emph{i}} \infty$, we have that 
	\begin{flalign*}
	\textbf{\emph{P}}_n (\mu) = B_{\mu}^{(s)} \left( \varpi_1, \varpi_2, \ldots , \varpi_n \b| \lambda - 2 \eta \displaystyle\sum_{k = 1}^n J_k, Z, L \right). 
	\end{flalign*}

\end{prop}

By setting $w_j - z_i - \eta = - \eta \Lambda_i$ in Proposition \ref{measuremmodel}, using Lemma \ref{dynamicalparticlesvlambda}, and recalling the definition of the dynamical $q$-Hahn boson model from Section \ref{LambdaSpecializationModelGeneral}, we obtain the following corollary, which is the analog of the previous proposition for the measure $\textbf{M}_n$ (from Definition \ref{measuresignaturesvertexmodel}).  

\begin{cor}
	
	\label{measuremmodelb} 
	
	Let $\eta, \tau, \lambda \in \mathbb{C}$, $W = (w_1, w_2, \ldots ) \subset \mathbb{C}$; $Z = (z_0, z_1, \ldots ) \subset \mathbb{C} $; $L = (\Lambda_0, \Lambda_1, \ldots) \subset \mathbb{C}$, with $\Lambda_0 = 0$; and $\mathscr{J} = (J_1, J_2, \ldots ) \subset \mathbb{Z}_{\ge 0}$. Also denote $\varpi_i = \big( w_i, w_i + 2 \eta, \ldots , w_i + 2 \eta (J_i - 1) \big)$ for each $i \ge 1$. Assume (for all integers $i, j \ge 1$) that
	\begin{flalign}
	\label{qcbdelta}
	q = e^{- 4 \pi \textbf{\emph{i}} \eta}; \quad c_i = q^{J_i} = e^{-4 J_i \pi \textbf{\emph{i}} \eta }; \quad w_j - z_i - \eta = - \eta \Lambda_i; \quad b_i = e^{4 \pi \textbf{\emph{i}} \eta \Lambda_i}; \quad \delta = e^{-2 \pi \textbf{\emph{i}} \lambda}, 
	\end{flalign} 
	
	\noindent so in particular we must have that $w_1 = w_2 = \cdots $. Then, under the trigonometric degeneration $\tau \rightarrow \textbf{\emph{i}} \infty$, we have that 
	\begin{flalign*}
	\textbf{\emph{M}}_n (\mu) = B_{\mu}^{(s)} \left( \varpi_1, \varpi_2, \ldots , \varpi_n \b| \lambda - 2 \eta \displaystyle\sum_{k = 1}^n J_k, Z, L \right).
	\end{flalign*}

\end{cor}

\section{Some Degenerations of the Dynamics} 

\label{JumpRandomWalk}

The dynamical stochastic higher spin vertex model can be degenerated to a number of other models of interest. For example, we mentioned in Section \ref{DegenerationsParticle} that when one sets $\delta = 0$, it becomes the stochastic higher spin vertex model of Corwin-Petrov \cite{SHSVML}, and when one sets $\mathcal{J} = (1, 1, \ldots )$ it becomes the stochastic IRF model of Borodin \cite{ERSF}. We also mentioned in Section \ref{LambdaSpecializationModelGeneral} that setting $u = s$ and reparameterizing $s^2 = b$ and $q^J = c$ produces the dynamical $q$-Hahn boson model. 

In this section we focus on the latter by listing a few degenerations of the dynamical $q$-Hahn boson model $\mathcal{M} (B, C; q; \delta)$; this list is by no means exhaustive. In what follows, we assume that all of the parameters $\{ b_ i \}$ and $\{ c_i \}$ defining $\mathcal{M}$ (from Section \ref{LambdaSpecializationModelGeneral}) are equal, that is, there exist $b, c \in \mathbb{R}$ such that all $b_i = b$ and all $c_i = c$; we also denote $a = bc$.

\begin{example}[$q$-Hahn boson model]
	
	Setting $\delta = 0$, $a = \mu$, and $b = \nu$, the weights from Definition \ref{stochasticparticlesjumps} become 
	\begin{flalign*}
	\mu^j \displaystyle\frac{\textbf{1}_{j \le i} (q; q)_i}{(q; q)_j (q; q)_{i - j}} \displaystyle\frac{(\nu / \mu; q)_j (\mu; q)_{i - j}}{(\nu; q)_i},
	\end{flalign*}
	
	\noindent which coincide with the weights for the $q$-Hahn boson model defined in equation (8) of \cite{IZRCMFSS}. 
	
\end{example} 

\begin{example}[$q$-Hahn boson model with inverted parameters]
	
	Letting $\delta$ tend to $\infty$, the weights from Definition \ref{stochasticparticlesjumps} become
	\begin{flalign*}
	\left(\displaystyle\frac{b}{a} \right)^{i - j} \displaystyle\frac{\textbf{1}_{j \le i} (q; q)_i}{(q; q)_j (q; q)_{i - j}} \displaystyle\frac{(b / a; q)_j (a; q)_{i - j}}{(b; q)_i}, 
	\end{flalign*}
	
	\noindent which coincide with the stochastic weights for the $q$-Hahn boson model with the parameters $q$, $\mu$, and $\nu$ equal to $q^{-1}$, $a^{-1}$, and $b^{-1}$, respectively. Thus, one can view the dynamical $q$-Hahn boson model as an interpolation between the $q$-Hahn boson model and the same model but with inverted parameters.
\end{example}

\begin{example}[Discrete-time AIP]
	
	Fixing positive real numbers $A, B > 0$, let us set $a = 1 - A \varepsilon$, $b = (1 - B \varepsilon) a = (1 - A \varepsilon) (1 - B \varepsilon)$, and send $\varepsilon$ to $0$. Then, the $\varphi \big( i \b| j \big)$ weight from Definition \ref{stochasticparticlesjumps} tends to $0$ unless $j = 0$ or $j = i$. In those cases, one obtains 
	\begin{flalign*}
	\varphi \big( 0 \b| i \big) = \displaystyle\frac{A}{A + B}; \qquad \varphi \big( i \b| i \big) = \displaystyle\frac{B}{A + B}. 
	\end{flalign*}
	
	\noindent In particular, if $i$ particles are at a site, then either all of them jump to the right with probability $B / (A + B)$ or none of them jump; thus, stacks of particles ``coalesce'' as time evolves. Scaling time by $B^{-1}$ and sending $B$ to $0$, this becomes the continuous-time \textit{asymmetric inclusion process} (AIP) introduced in \cite{AIP}. 
	
	Interestingly, these (discrete-time) dynamics are independent of both $q$ and the dynamical parameter $\kappa$ (they could have been derived from the original $q$-Hahn boson model but we are unaware of any previous mention to this), which leads us to believe that they are free-fermionic. Indeed, they can in fact be obtained by combining a limit transition with a zero-range transformation of the geometric discrete-time TASEP discussed in Section 2.6 of \cite{AGRSD}; we are again unaware of any previous mention to this fact in the literature.  
	
\end{example}

\begin{example}[$q$-Boson model with few particles] 
	
	\label{modelfewparticles}
	
	Set $a = qb$ in the weights from Definition \ref{stochasticparticlesjumps}, and consider the associated dynamical $q$-Hahn boson model with \textit{step initial data}. Thus one path enters through each vertex of the $y$-axis or, in terms of the interacting particle interpretation from Section \ref{InteractingParticles}, one particle enters the system at each time step. Under the specialization $a = qb$, the stochastic weights become 
	\begin{flalign*}
	\varphi \big( 0 \b| i \big) = \displaystyle\frac{(1 - q^i b)(1 - b \omega)}{(1 - b)(1 - q^i b \omega)}; \qquad \varphi \big( 1 \b| i \big) = \displaystyle\frac{b (q^i - 1)(1 - \omega)}{(1 - b)(1 - q^i b \omega)},
	\end{flalign*}
	
	\noindent and $\varphi \big( j \b| i \big) = 0$ for $j \ge 2$; here, we have set $\omega = \omega (x, y) = q^{i_1 (x, y)} \kappa_{x, y - 1}$. The multiplicative dynamical rules for $\kappa$ listed in \eqref{kappamultiplicative} can be used to provide similar dynamical rules for $\omega$. In particular, $\omega$ is defined by setting $\omega (1, 0) = \delta$; $\omega (x, y) = b q^{i_1 (x, y)} \omega( x - 1, y)$ for each $x \ge 2$ and $y \ge 0$; and $\omega (x, y) = q^{1 - j_1 (x, y - 1) - j_2 (x, y - 1)} \omega (x, y - 1)$ for each $x, y \ge 1$.
	
	Now set $b = - \varepsilon$ and $y = \varepsilon^{-1} t$, and then take the limit as $\varepsilon$ tends to $0$. Then, denoting $\widehat{\omega} (x, t) = \omega (x, \lfloor \varepsilon^{-1} t \rfloor)$, we find that $\widehat{\omega} (x, t) = 0$ for $x > 1$. Thus, particles jump from site $x$ to $x + 1$ (for $x > 1$) according to independent exponential clocks of rate $1 - q^i$, where $i$ denotes the number of particles at site $x$. In particular, the parameter $\omega$ does not affect the dynamics in the \textit{bulk} (to the right of the first site) of the model; there, the dynamics coincide with those of the $q$-boson model (which is also the zero-range transformation of the $q$-TASEP \cite{MP}), introduced in \cite{DIGM,RTAEM}. 
	
	At the first site, particles jump one space to the right according to an exponential clock of rate $(1 - q^i) (1 - \widehat{\omega})$, where $\omega = \widehat{\omega} (1, t)$. In particular, $\widehat{\omega} (0) = \delta$ and, due to the step initial data, $\widehat{\omega} (1, t)$ remains constant unless a particle jumps away from site $1$ at time $t$; at that moment, $\widehat{\omega} (1, t)$ is multiplied by $q^{-1}$. Hence, $\widehat{\omega} (t) = \delta q^{-r}$, where $r = r_t$ denotes the total number of particles that have jumped to the right from site $1$ before (or at) time $t$. 
	
	Now, further set $\omega = q^d$, for some integer $d > 0$. Then, $\widehat{\omega} = 1$ after $d$ particles have jumped out of site $1$; due to the factor of $1 - \widehat{\omega}$ in the jump rate out of the first site, this implies that at most $d$ particles will be in the bulk of the system at any given time. To summarize, setting $a = qb$, $b = \varepsilon$, scaling time by $\varepsilon^{-1}$, and setting $\delta = q^d$ for some positive integer $d$ in the stochastic weights of the dynamical $q$-Hahn boson model gives rise to a $q$-boson model in which particles ``enter'' the system according to an exponential clock of rate $1 - q^{d - r}$, where $r$ denotes the number of particles already in the system. 
	
	It is plausible that taking the limit as $q$ tends to $1$ might give rise to a polymer model with an inhomogeneous boundary condition, but we do not pursue that direction further here. 
	
\end{example}

\begin{example}[Dynamical asymmetric midpoint corner growth model and partial exclusion process]
	
	\label{dynamicalasymmetricmidpointcorner}
	
	Fix $q \in (0, 1)$ and $\delta \le 0$, and set $p = 1 / (q + 1)$. Consider the dynamical $q$-Hahn boson model with $a = q^{-1}$, and $b = q^{-2}$. Then, the weights from Definition \ref{stochasticparticlesjumps} become
	\begin{flalign*}
	\varphi \big( 0 \b| 0 \big) = 1; \quad \varphi \big( 1 \b| 2 \big) = 1; \quad \varphi \big( 0 \b| 1 \big) = \displaystyle\frac{q - \kappa}{(q + 1)(1 - \kappa) }; \quad \varphi \big( 1 \b| 1 \big) = \displaystyle\frac{1 - q \kappa}{(q + 1) (1 - \kappa)}. 
	\end{flalign*}
	
	\noindent and $\varphi \big( i \b| j \big) = 0$ for all other $i, j$. 
	
	If $\delta = 0$, then all $\kappa_{x, y} = 0$; in this case, $\varphi \big( 0 \b| 1 \big) = 1 - p$ and $\varphi \big( 1 \b| 1 \big) = p$. When the corresponding $q$-Hahn TASEP is run with $(1, 1, \ldots )$-initial data, this gives rise to the following discrete-time partial exclusion process on $\mathbb{Z}_{> 0}$, in which each site can accommodate at most two particles. Initially (at time $0$), no site is occupied by any particles. At each (discrete) time step $t > 0$, one particle enters the leftmost site (site 1), and then particles move to the right as follows. If a site $i$ has two particles at time $t - 1$, then it deterministically transfers one of its particles to site $i + 1$. If site $i$ has one particle at time $t - 1$, then it transfers one particle to site $i + 1$ with probability $p$ and no particles with probability $1 - p$. We refer to this (non-dynamical) system as the \textit{$(q; 1)$-asymmetric PEP}; its dynamical analog is called the \textit{dynamical $(q; 1; \delta)$-asymmetric PEP}. 
	
	It is quickly verified that, denoting the current of this model by $\mathfrak{h}_t (x)$, the height function $\zeta_t (x)$ of the asymmetric midpoint corner growth model with parameter $p$ (see Section \ref{DynamicalCornerGrowth}) is related to $\mathfrak{h}_t (x)$ through $2 \mathfrak{h}_t (x) + 2 (x - 1) - t = \zeta_t \big( x - \frac{t}{2} - 1 \big)$. In the symmetric setting, this was explained as the $J = 1$ case of \eqref{dynamicalparameterpartialexclusion} in Remark \ref{cornerpartialexclusion}, and the asymmetric case is entirely analogous. 
	
	If $\delta < 0$, then this model becomes a dynamical partial exclusion process. It can again be mapped to an asymmetric dynamical midpoint corner growth model, as in Section \ref{DynamicalCornerGrowth}, in which the probabilities $p (t, \zeta_t , z)$ of a midpoint $z = (z_1 + z_2) / 2$ (of two consecutive vertices $z_1, z_2 \in V_t$ forming a horizontal segment) going up and down are $\frac{1 - q^{1 - \zeta_t (z)} \delta}{(q + 1)(1 - q^{- \zeta_t (z)} \delta)}$ and $\frac{q - q^{- \zeta_t (z)} \delta}{(q + 1)(1 - q^{- \zeta_t (z)} \delta)}$, respectively. We refer to this model as the \emph{dynamical asymmetric midpoint corner growth model}. Observe that, when $\delta = 0$ and $\delta = - \infty$, it becomes the (non-dynamical) asymmetric midpoint corner growth models with parameters $p$ and $1 - p$, respectively.
	\end{example}

\begin{example}[Dynamical $(A, J; \gamma)$-PEP with Hahn weights]
	
	\label{dynamicjagamma}
	
	Fix $A, J \in \mathbb{Z}_{\ge 0}$, $\gamma \in \mathbb{R}$, and set $c = q^J$, $a = q^{-A}$, $\delta = q^{-\gamma}$. Next let $q$ tend to $1$. Using the seventh identity in \eqref{basichypergeometric1}, we find that 
	\begin{flalign}
	\label{dynamicalq1}
	\varphi \big( j \b| i \big) = \binom{J}{j} \binom{A}{i - j} \binom{A+ J}{i}^{-1} \displaystyle\frac{(i - A - J - \widehat{\kappa})_{i - j} (i - j - \widehat{\kappa} + 1)_j}{(i - j - A - \widehat{\kappa})_{i - j} (2i - 2j + 1 - A - \widehat{\kappa})_j}, 
	\end{flalign}
	
	\noindent where we have denoted $\kappa_{x, y} = q^{-\widehat{\kappa}}$. 
	
	Run this model with \textit{$J$-step initial data}, meaning that $J$ paths enter through each site of the $y$-axis, and no paths enter through the $x$-axis. In view of the weights \eqref{dynamicalq1}, we have that $\varphi \big( j \b| i \big) = 0$ if $j > J$, meaning that at most $J$ particles can jump from some site $x \in \mathbb{Z}_{> 0}$ to site $x + 1$ at any given time. Furthermore, using the $J$-step initial data and \eqref{dynamicalq1}, it can be quickly shown that $i_1 (x, y) \le A + J$ for any vertex $(x, y)$ in a path ensemble supported by this dynamical $q$-Hahn boson model. Stated alternatively, at most $A + J$ particles can occupy any given site. 
	
	Now let us see when the $\varphi \big( j \b| i \big)$ given by \eqref{dynamicalq1} are nonnegative. In view of the description \eqref{dynamicalparametercurrent} of the dynamical parameter in terms of the current, we find that $\widehat{\kappa} = \gamma - A - J + 2 \mathfrak{h} (x, y) + Ax + J (x - y) $. Due to the $J$-step initial data, we have that $\mathfrak{h} (x, y) \ge \max \big\{ 0, J (y - x + 1) \big\}$. This implies that $\widehat{\kappa} (x, y) \ge \gamma$ for each $x \ge 1$ and $y \ge 0$. Using the fact that $j \le i \le A + J$ and $j \le J$, we deduce that the weights $\varphi$ above are nonnegative if $\gamma > 2 A + 2 J$. 
	
	Thus, setting $\mu = q^{-J}$, $a = q^{-A}$, $\delta = q^{-\gamma}$ (with $\gamma> 2 A + 2 J$), and letting $q$ tend to $1$ gives rise to a totally asymmetric, dynamical \textit{partial exclusion process} (PEP), in which at most $A + J$ particles may occupy any site and at most $J$ particles can jump from any site. When $\gamma$ tends to $\infty$, the $\varphi$ \eqref{dynamicalq1} become the weights for the Hahn orthogonal polynomials, so we refer to this process as the \textit{dynamical $(A, J; \gamma)$-PEP with Hahn weights}. 
	
\end{example}

\begin{example}[Dynamical $(J; \gamma)$-PEP]
	
	\label{dynamicjgammaexclusion} 
	
	Consider the dynamical $(A, J; \gamma)$-PEP with Hahn weights (run with $J$-step initial data), as defined above, with $A = 1$ and $\gamma > J + 1$. Further setting $\Upsilon = \widehat{\kappa}$, the weights \eqref{dynamicalq1} become 
	\begin{flalign}
	\label{dynamicalq2}
	\varphi \big( i - 1 \b| i \big) = \displaystyle\frac{i (\Upsilon + J - i + 1) }{(J + 1) \Upsilon }; \qquad \varphi \big( i \b| i \big) = \displaystyle\frac{(J - i + 1) (\Upsilon - i)}{ (J + 1) \Upsilon},
	\end{flalign}
	
	\noindent and $\varphi \big( j \b| i \big) = 0$ for all other $j$, where we have again denoted $\kappa_{x, y} = q^{-\Upsilon}$. As before, $\Upsilon$ can be defined in terms of the current of the model through $\Upsilon = \gamma + 2 \mathfrak{h} (x, y) + (J + 1) (x - 1) - J y$. Thus, due to \eqref{dynamicalparameterpartialexclusion}, this model coincides with the dynamical $(J; \gamma)$-PEP from Definition \ref{dynamicjgammaexclusion1}. 
	
	Observe that the $J = 1$ special case can also be obtained by taking the limit as $q$ tends to $1$ in Example \ref{dynamicalasymmetricmidpointcorner}. One could also consider an asymmetric deformation of the general $J$ case by considering the dynamical $q$-Hahn boson model with $a = q^{-1}$ and $b = q^{-J - 1}$, which again coincides with the degeneration from Example \ref{dynamicalasymmetricmidpointcorner} when $J = 1$. 
	
	Let us also mention that it might be possible to recover some dynamical analog of the Gaussian limit of the $q$-Whittaker process studied in \cite{AGLP}, by taking the limit of this dynamical $(J; \gamma)$-exclusion process as $J$ tends to $\infty$ (and suitably rescaling mass and time); however, we will not pursue this further in the present paper. 
\end{example}

\begin{example}[Dynamic MADM]
	
	Denote $b / a = \mu$, set $\mu = 1 - \varepsilon$, and set $a = q$. Let us consider what happens to the jump probabilities $\varphi = \varphi_{q, a, b, \kappa}$ when $\varepsilon$ tends to $0$. Denoting $\widehat{\kappa} = \widehat{\kappa} (x, t) = \kappa_{x, \lfloor \varepsilon^{-1} t \rfloor}$, we find that 
	\begin{flalign}
	\label{jumpratemu0}
	\widetilde{\varphi}_{\widehat{\kappa}} \big( j \b| i \big) = \displaystyle\lim_{\varepsilon \rightarrow 0} \varepsilon^{-1} \varphi \big(j \b| i \big) = \textbf{1}_{j \le i} \displaystyle\frac{q^j (1 - q^{2i - 2j + 1} \widehat{\kappa}) }{(1 - q^j)(1 - q^{2i - j + 1} \widehat{\kappa})}. 
	\end{flalign} 
	
	\noindent Hence, at any time $t$, $j$ particles move from site $x$ to site $x + 1$ according to a (time-inhomogeneous) exponential clock of rate $\widetilde{\varphi}_{\widehat{\kappa} (x, t)} \big( j \b| i \big)$ defined above, where $i$ denotes the number of particles at site $x$ and time $t^-$. 
	
	Observe here that the dynamical parameter $\widehat{\kappa} (x, t)$ now changes continuously over time. In particular, we set $\widehat{\kappa} (1, 0) = \delta$ and in general have that $\widehat{\kappa} (x, T) = q^{x - 1 - 2 \mathfrak{h}_T (x)} e^T \delta$ due to \eqref{dynamicalparametercurrent} where, as in Section \ref{InteractingParticles}, $\mathfrak{h}_T (x)$ denotes the current of the model through $x$ after time $T$. 
	
	In particular, if $\delta \le 0 < q < 1$, then $\widehat{\kappa} (x, t) \le 0$ for all $x$ and $t$, which implies that the rates $\widetilde{\varphi}$ from \eqref{jumpratemu0} are positive. When $\delta = 0$, this process degenerates to a totally asymmetric version of the \textit{multi-particle asymmetric diffusion model} (MADM) introduced by Sasamoto and Wadati \cite{DWE} and later analyzed by Corwin and Barraquand \cite{AEP}. Therefore, we refer to the general ($\delta < 0$) version of this process as a \textit{dynamical totally asymmetric MADM}. 
	
\end{example}

\section{Observables and Asymptotics}

\label{ObservablesDegenerations}

The purpose of this section is to provide explicit contour integral identities for a certain family of observables of the dynamical $q$-Hahn boson model, under some restrictions of the underlying parameters; this is given by Theorem \ref{fusedhigherspin}, which will follow from suitable analytic continuation of a similar identity produced in \cite{ERSF} for the (unfused) stochastic IRF models. 

As an application of Theorem \ref{fusedhigherspin}, we establish Corollary \ref{kexclusionidentity} in Section \ref{MomentIntegrals}, which is a contour integral identity for moments of the dynamical $(J; \gamma)$-PEP introduced in Definition \ref{dynamicjgammaexclusion1}. In Section \ref{momentm} and Section \ref{momentmdynamic}, we use this identity to analyze asymptotics of the current for this model and in particular establish Proposition \ref{momentm} and Theorem \ref{momentmdynamic}. 

\subsection{Observables of the Dynamical \texorpdfstring{$q$}{}-Hahn Boson Model}

\label{MomentIntegrals}

The goal of this section is to provide contour integral identities for certain observables of the dynamical $q$-Hahn boson model, given by Theorem \ref{fusedhigherspin}. We then use this theorem to prove Corollary \ref{kexclusionidentity}, which gives observables for the dynamical $(J; \gamma)$-PEP. 

To establish that theorem we require the proposition below, established in \cite{ERSF}, that evaluates averages of specific quantities against the \textit{formal measure} on $\Sign^+$ given by $B_{\mu}^{(s)}$ (recall its definition from \eqref{bsymmetrics}). Here, we use the term ``formal measure'' since $\sum_{\mu \in \Sign_N^+} B_{\mu}^{(s)} \big( w_1, w_2, \ldots , w_M \b| \lambda \big) = 1$ under suitable constraints on $W = (w_1, w_2, \ldots, w_M) \subset \mathbb{C}$, but $B_{\mu}^{(s)} (W \b| \lambda)$ is not always guaranteed to be nonnegative (or real). 

In what follows, for any $\mu \in \Sign^+$ and $k \in \mathbb{Z}$, we denote by $\mathfrak{h}_{\mu} (k)$ the number of indices $j$ for which $\mu_j \ge k$.

\begin{prop}[{\cite[Theorem 10.1 and Equation (10.8)]{ERSF}}]
	
\label{higherspin}

Fix parameters $M \in \mathbb{Z}_{> 0}$; $\eta, \lambda, \tau \in \mathbb{C}$; $W = (w_1, w_2, \ldots ) \subset \mathbb{C}$; $Z = (z_0, z_1, \ldots ) \subset \mathbb{C} $; and $L = (\Lambda_0, \Lambda_1, \ldots) \subset \mathbb{C}$. Also set $v_{i, j} = w_j - z_i - \eta$ for each $i, j$. Denote 
\begin{flalign}
\label{uxis}
q = e^{- 4 \pi \textbf{\emph{i}} \eta}; \qquad s_j = e^{2 \pi \textbf{\emph{i}} \eta \Lambda_j}; \qquad \xi_j = e^{2 \pi \textbf{\emph{i}} z_j}; \qquad u_k = e^{2 \pi \textbf{\emph{i}} (\eta - w_k)}, 
\end{flalign}

\noindent and $\Lambda_{[i, j)} = \sum_{k = i}^{j - 1} \Lambda_k$ for any integers $i \le j$; assume that $|q| < 1$; and suppose we have taken the trigonometric limit $\tau \rightarrow \textbf{\emph{i}} \infty$. We additionally impose the following restrictions on parameters. 

\begin{itemize}
	
	\item The $\{ s_j \}$ and $\{ \xi_j \}$ are sufficiently close so that $\min_j |s_j^{-1} \xi_j| < q \max_j |s_j^{-1} \xi_j|$ 
	
	\item The $\{ s_j \}$ are sufficiently large so that $\min_j |q^M s_j \xi_j| > \max_j |s_j^{-1} \xi_j|$. 

\end{itemize}

\noindent Then, there exists an open subset $\mathcal{D} \in \mathbb{C}$ containing $s_1^{-1} \xi_1$ such that if all $s_j^{-1} \xi_j \in \mathcal{D}$ and all $u_j^{-1} \in \mathcal{D}$, then for any positive integers $x_1, x_2, \ldots, x_k$ we have that 
\begin{flalign} 
\label{integralbunfused}
\begin{aligned}
& \displaystyle\frac{1}{(e^{2 \pi \textbf{\emph{i}} \lambda}; q)_k} \displaystyle\sum_{\mu \in \Sign_M^+} B_{\mu}^{(s)} \big( w_1, w_2, \ldots , w_M \b| \lambda + 2 \eta (\Lambda_0 - M) \big) \\
& \qquad \qquad \qquad \qquad \qquad \times  \displaystyle\prod_{j = 0}^{k - 1} \Bigg( \big( 1 - q^{k - \mathfrak{h}_{\mu} (x_{j + 1})} \big) \bigg( q^M \prod_{i = 1}^{x_{j + 1} - 1} s_i^2 - e^{2 \pi \textbf{\emph{i}} \lambda} q^{k + \mathfrak{h}_{\mu} (x_{j + 1})} \bigg) \Bigg)	 \\
& \qquad \qquad = \displaystyle\frac{q^{\binom{k}{2}}}{(2 \pi \textbf{\emph{i}})^k} \displaystyle\oint \cdots \displaystyle\oint \displaystyle\prod_{1 \le i < j \le k} \displaystyle\frac{y_i - y_j}{y_i - q y_j} \displaystyle\prod_{j = 1}^k \left( \displaystyle\prod_{i = 1}^{x_j - 1} \displaystyle\frac{ \xi_i - s_i y_j}{\xi_i - s_i^{-1} y_j} \displaystyle\prod_{i = 1}^N \displaystyle\frac{1 - q u_i y_j}{1 - u_i y_j} \right) \displaystyle\frac{d y_j}{y_j},
\end{aligned}
\end{flalign}

\noindent where the $y_j$ contours are positively oriented and contain the $u_i^{-1}$ but not the $\{ \xi_i s_i \}$ or $0$.  
\end{prop}

\begin{rem}
	
	\label{contouryunfused}
	
	Observe that poles of the integrand on the right side of \eqref{integralbunfused} also arise when $y_j = q^{-k} u_i^{-1}$ for some $i, j, k$. As explained in the proof of the second part of Theorem 8.13 in \cite{HSVMRSF}, the residues at these poles are zero, so the integral does not depend on whether the $y_j$ contours contain enclose numbers of the form $q^{-k} u_i^{-1}$. 
\end{rem}

Using Proposition \ref{measuremmodel} and Proposition \ref{higherspin}, we can establish the Theorem \ref{fusedhigherspin}, which evaluates observables (similar to the ones on the left side of \eqref{integralbunfused}) for the dynamical $q$-Hahn boson model.

\begin{proof}[Proof of Theorem \ref{fusedhigherspin}]
	
	Define parameters $\eta, \lambda \in \mathbb{C}$; $W = (w_1, w_2, \ldots ) = (w, w, \ldots) \subset \mathbb{C}$; $Z = (z_0, z_1, \ldots) \subset \mathbb{C}$; and $L = (\Lambda_0, \Lambda_1, \ldots ) \subset \mathbb{C}$ from the parameters $q, \delta \in \mathbb{C}$, $B, C \subset \mathbb{C}$, and $\mathcal{J} \subset \mathbb{Z}_{> 0}$ through \eqref{qcbdelta}. Further define the parameter sets $\Xi = (\xi_1, \xi_2, \ldots ), S = (s_1, s_2, \ldots ) \subset \mathbb{C}$ in terms of $\eta$, $W$, $Z$, and $L$ through \eqref{uxis}; we also set $U = \varpi_1 \cup \varpi_2 \cup \cdots \varpi_N$, where $\varpi_k = (u, qu, \ldots , q^{J_k - 1} u)$ and $u = e^{2 \pi \textbf{i} (\eta - w)}$, as in \eqref{uxis}.
	
	Denoting by $\mu$ the signature associated with the particle configuration of the dynamical $q$-Hahn boson model after $N$ steps (recall the end of Section \ref{InteractingParticles}), we find that the observable on the left side of \eqref{integralbunfused} is equal to 
	\begin{flalign*}
	& \displaystyle\prod_{j = 0}^{k - 1} \Bigg( \big( 1 - q^{k - \mathfrak{h}_{\mu} (x)} \big) \bigg( q^{\sum_{i = 1}^N J_i} \displaystyle\prod_{i = 1}^{x - 1} s_i^2 - e^{2 \pi \textbf{i} \lambda} q^{k + \mathfrak{h}_{\mu}	 (x)} \bigg) \Bigg)\\
	& \qquad  =	\displaystyle\prod_{j = 0}^{k - 1} \Bigg( \big( 1 - q^{k - \mathfrak{h}_N (x)} \big) \bigg( \displaystyle\prod_{i = 1}^N c_i \displaystyle\prod_{i = 1}^{x - 1} b_i - \delta^{-1} q^{k + \mathfrak{h}_N (x)} \bigg) \Bigg),
	\end{flalign*}
	
	\noindent since $M = |U| = \sum_{i = 1}^N J_i$. Then, due to Lemma \ref{measuremmodel}, it suffices to show that the left side of \eqref{integralbunfused} is equal to the right side of \eqref{fusedhigherspinintegral}. 
	
	Denote the former by $E = E (\eta, \mathcal{J}, \lambda, W, Z, L)$ and the latter by $I = I (\eta, \mathcal{J}, \lambda, W, Z, L)$; also denote the right side of \eqref{integralbunfused} by $\widetilde{I} = \widetilde{I} (\eta, \mathcal{J}, \lambda, W, Z, L) = \widehat{I} (q, U, \Xi, S)$. Since $u = s_j \xi_j$ and $s_j^2 = b_j$ for each $j \ge 1$, we deduce from the change of variables $z_j = u y_j$ (and Remark \ref{contouryunfused}) that $I = \widetilde{I}$. 
	
	Hence, it suffices to show that $E = \widetilde{I}$. Under specific constraints of the parameters, this is the content of Proposition \ref{higherspin}. However, those constraints do not apply in our situation; for example, the different $u_j^{-1}$ will differ by powers of $q$, meaning that they are not sufficiently close as in Proposition \ref{higherspin}. To remedy this issue, we proceed through analytic continuation. 
	
	To that end, observe that since $1 - q^{k - \mathfrak{h}_{\mu} (x)}$ and $q^M \prod_{i = 1}^{x_j - 1} s_i^2 - e^{2 \pi \textbf{i} \lambda} q^{k + \mathfrak{h}_{\mu} (x)}$ are uniformly bounded, the expression $E(\eta, \mathcal{J}, \lambda, W, Z, L)$ is analytic in all parameters in the (connected) region $\mathcal{R}_1$ where 
	\begin{flalign}
	\label{bsumconvergence}
	\displaystyle\sum_{\mu \in \Sign_M^+} \Big| B_{\mu}^{(s)} \big( w_1, w_2, \ldots , w_M \b| \lambda + 2 \eta (\Lambda_0 - M) \big) \Big| 
	\end{flalign}
	
	\noindent converges absolutely. 
	
	Furthermore, $\widetilde{I} (\eta, \mathcal{J}, \lambda, W, Z, L)$ is also analytic in all parameters in the (connected) region $\mathcal{R}_2$ where all $\xi_i s_i \notin \big[ \min_j |u_j|^{-1}, \max_j |u_j| \big]$; indeed, this is the domain where the $y_j$ contours satisfying the properties in Proposition \ref{higherspin} (listed directly after \eqref{integralbunfused}) exist. 
	
	By Proposition \ref{higherspin}, $E (\eta, \mathcal{J}, W, Z, L) = \widetilde{I} (\eta, \mathcal{J}, \lambda, W, Z, L)$ when the parameters $\eta, \mathcal{J}, W, Z, L$ are in some open subset $\mathcal{R}_3 \subset \mathcal{R}_1 \cap \mathcal{R}_2$. Thus, by analytic continuation, we have that $E (\eta, \mathcal{J}, W, Z, L) = \widetilde{I} (\eta, \mathcal{J}, \lambda, W, Z, L)$ when $(\eta, \mathcal{J}, W, Z, L) \in \mathcal{R}_1 \cap \mathcal{R}_2$. 
	
	Now, observe that since the parameters $(\eta, \mathcal{J}, \lambda, W, Z, L)$ were chosen in the statement of the theorem so that $\varphi_{q, b_x, a_{x, y}, \kappa_{x, y - 1}} \big( j \b| i \big) \ge 0$, the sum \eqref{bsumconvergence} converges and is fact equal to $1$; hence, in our given situation, we have that $(\eta, \mathcal{J}, W, Z, L) \in \mathcal{R}_1$. Furthermore, since we stipulated that $b_j \notin \big[ 1, \max_i q^{1 - J_i} | \big]$, we also have that $(\eta, \mathcal{J}, W, Z, L) \in \mathcal{R}_2$.
	
	Thus, in view of the above, it follows that $E = \widetilde{I}$, from which we deduce the theorem.
\end{proof}

\begin{rem}
	
	\label{generalfusedcontour}
	
	A contour integral identity, similar to \eqref{fusedhigherspinintegral}, exists for observables of the measure $\textbf{P}_N$. In this case, the identity becomes 
	\begin{flalign}
	\label{fusedhigherspinintegralgeneral}
	\begin{aligned}
	& \displaystyle\frac{1}{\big( \delta^{-1}; q \big)_k} \mathbb{E}_{\textbf{P}_N} \left[ \displaystyle\prod_{j = 0}^{k - 1} \bigg( \delta^{-1} q^j - q^{-\mathfrak{h}_N (x_{j + 1})} \displaystyle\prod_{i = 1}^N c_i \displaystyle\prod_{i = 1}^{x_j - 1} b_i \bigg) \big(q^j - q^{\mathfrak{h}_N (x_{j + 1})} \big) \right] \\
	& \qquad \qquad \quad = \displaystyle\frac{q^{\binom{k}{2}}}{(2 \pi \textbf{i})^k} \displaystyle\oint \cdots \displaystyle\oint \displaystyle\prod_{1 \le i < j \le k} \displaystyle\frac{z_i - z_j}{z_i - q z_j} \displaystyle\prod_{j = 1}^k \left( \displaystyle\prod_{i = 1}^{x_j - 1} \displaystyle\frac{\xi_i s_i - s_i^2 z_j}{\xi_i s_i - z_j} \displaystyle\prod_{i = 1}^N \displaystyle\frac{1 - q^{J_i} u_i z_j}{1 - u_i z_j} \right) \displaystyle\frac{d z_j}{z_j},
	\end{aligned}
	\end{flalign}
	
	\noindent where the $z_j$ contours are positively oriented and contain $\bigcup_{i = 1}^N \{ u_i^{-1}, q^{-1} u_i^{-1}, \ldots , q^{1 - J_i} u_i^{-1} \}$ but no other poles of the integrand. Here, we assume that such contours exist and further that all parameters $U = (u_1, u_2, \ldots ) \subset \mathbb{C}$, $\Xi = (\xi_1, \xi_2, \ldots ) \subset \mathbb{C}$, $S = (s_1, s_2, \ldots ) \subset \mathbb{C}$, $J_1, J_2, \ldots \in \mathbb{Z}_{\ge 1}$, and $\kappa \in \mathbb{C}$ are chosen the $\psi$ stochastic weights are nonnegative (so that $\textbf{P}_N$ is indeed a probability measure). The proof of \eqref{fusedhigherspinintegralgeneral} is similar to that of Theorem \ref{fusedhigherspin} and is therefore omitted. 
\end{rem}

Through a suitable degeneration of Theorem \ref{fusedhigherspin}, we can establish Corollary \ref{kexclusionidentity}.

\begin{proof}[Proof of Corollary \ref{kexclusionidentity}]
	
Recall from Example \ref{dynamicjgammaexclusion} that the dynamical $(J, \gamma)$-PEP is obtained from the dynamical $q$-Hahn boson model by setting $\delta = q^{-\gamma}$; $b_i = q^{-J - 1}$ for each $i \ge 1$; $c_i = q^J$ for each $i \ge 1$; and taking the limit as $q$ tends to $1$. Applying this degeneration to Theorem \ref{fusedhigherspinintegral} (with both sides divided by $(1 - q)^k$) yields 
\begin{flalign}
\label{mexclusionintegral1} 
\begin{aligned}
& \displaystyle\lim_{q \rightarrow 1} \displaystyle\frac{1}{(1 - q)^k (q^{\gamma}; q)_k} \mathbb{E} \left[ \displaystyle\prod_{j = 0}^{k - 1} \Big( q^{\gamma + j} -  q^{NJ - (J + 1) x_j - \mathfrak{h}_N (x_j + 1)} \Big) \big( q^j - q^{\mathfrak{h}_N (x_j + 1)} \big) \right] \\
& \quad = \displaystyle\lim_{q \rightarrow 1} \displaystyle\frac{q^{\binom{k}{2}} (1 - q)^k}{(2 \pi \textbf{i})^k} \displaystyle\oint \cdots \displaystyle\oint \displaystyle\prod_{1 \le i < j \le k} \displaystyle\frac{z_i - z_j}{z_i - q z_j} \displaystyle\prod_{j = 1}^k \left( \displaystyle\frac{1 - z_j}{1 - q^{J + 1} z_j}  \right)^{x_j} \left( \displaystyle\frac{1 - q^J z_j}{1 - z_j} \right)^N \displaystyle\frac{d z_j}{z_j},
\end{aligned}
\end{flalign}

\noindent where the $z_j$ contours are positively oriented and enclose $\{ 1, q^{-1}, \ldots , q^{1 - J} \}$ but not $0$ or $q^{-J - 1}$. 

Now change variables on the right side of the above equality by denoting $z_j = q^{y_j - J - 1}$. Since $z_j^{-1} d z_j = (\log q) d y_j$, the right side becomes 
\begin{flalign}
\label{mexclusionintegral2}
\displaystyle\lim_{q \rightarrow 1} \displaystyle\frac{q^{\binom{k}{2}} (\log q)^k}{(2 \pi \textbf{i})^k (1 - q)^k} \displaystyle\oint \cdots \displaystyle\oint \displaystyle\prod_{1 \le i < j \le k} \displaystyle\frac{q^{y_i} - q^{y_j}}{q^{y_i} - q^{y_j + 1}} \displaystyle\prod_{j = 1}^k \left( \displaystyle\frac{1 - q^{y_j - J - 1}}{1 - q^{y_j}} \right)^{x_j} \left( \displaystyle\frac{1 - q^{ y_j - 1}}{1 - q^{y_j  -J - 1}} \right)^N d y_j,
\end{flalign}

\noindent where the $y_j$ are negatively oriented and enclose $2, 3, \ldots , J + 1$ but not $0$; observe here that the contours for the $y_j$ are negatively oriented since the contours for $z_j = q^{y_j}$ were positively oriented and $|q| < 1$. 

Observing that $\displaystyle\lim_{q \rightarrow 1} (1 - q)^{-1} \log q = -1$ and inserting \eqref{mexclusionintegral2} into \eqref{mexclusionintegral1} yields the theorem.  	
\end{proof}

\subsection{Asymptotics for the \texorpdfstring{$(J, \infty)$}{}-Partial Exclusion Process} 

\label{NonDynamickAsymptotic}

In this section we use Corollary \ref{kexclusionidentity} to outline a proof of Proposition \ref{momentm}, which accesses asymptotics for the current of the non-dynamical limit ($\gamma \rightarrow \infty$) of the $(J; \gamma)$-PEP, given in Section \ref{DynamicalPartialExclusion}. This proposition is plausibly accessible from the hydrodynamical limit methods described in the book \cite{SLIS}; however, let us briefly outline how to derive it through the identities given by Corollary \ref{kexclusionidentity}. 

We will see that the current is governed by the heat equation. In particular, for any $s \in \mathbb{R}$ and $r \in \mathbb{R}_{> 0}$, define 
\begin{flalign}
\label{hsr}
\mathcal{H} (s, r) & = \displaystyle\frac{1}{2 \pi \textbf{i}} \displaystyle\oint_{1 + \textbf{i} \mathbb{R}} \exp \left(\displaystyle\frac{r J z^2}{2} - s (J + 1) z \right)\displaystyle\frac{dz}{z^2}, 
\end{flalign}

\noindent where $z$ is integrated from top to bottom. It can be quickly verified that $\mathcal{H} (s, r)$ is the solution to the normalized heat equation $J \partial_r \mathcal{H} = 2 (J + 1)^2 \partial_s^2 \mathcal{H}$, with initial data $\mathcal{H} (s, 0) = (J + 1) |s| \textbf{1}_{s < 0}$. In particular, $\mathcal{H} (0, r) = \sqrt{r J / 2 \pi}$.

The following proposition shows the convergence of the normalized current of the $(J, \infty)$-PEP $\mathfrak{H}_T (s, r)$ (from \eqref{htsr}) to $\mathcal{H} (s, r)$ in the sense of moments; it is possible to establish stronger forms of convergence through a more refined analysis, but we do not pursue that further in this paper. 

\begin{prop} 
	
	\label{momentm}
	
	Let $m, J \in \mathbb{Z}_{\ge 1}$ and $s, r \in \mathbb{R}$ with $r > 0$. Then, 
	\begin{flalign*}
	\displaystyle\lim_{T \rightarrow \infty} \mathbb{E} \bigg[ \mathfrak{H}_T (s, r)^m \bigg] = \displaystyle\lim_{T \rightarrow \infty} T^{- m / 2} \mathbb{E} \bigg[ \mathfrak{h}_{\lfloor r T \rfloor} \Big( \Big\lfloor \displaystyle\frac{J r T}{J + 1} + s \sqrt{T} \Big\rfloor \Big)^m \bigg] = \mathcal{H} (s, r)^m, 
	\end{flalign*}
	
	\noindent where the expectations are taken with respect to the $(J, \infty)$-PEP. 
\end{prop}

\begin{proof}[Proof (Outline)]
	
	The first equality of the proposition is from the definition \eqref{htsr} of $\mathfrak{H}_T (s, r)$, so it suffices to establish the second equality. We only do this in the case $m = 1$; the general case $m \ge 1$ is entirely analogous. 
	
	To that end, let $x + 1 = \lfloor J r T / (J + 1) + s \sqrt{T} \rfloor$. Applying Corollary \ref{kexclusionidentity} with $k = m = 1$, $x_1 = x$, and $\gamma = \infty$ yields 
	\begin{flalign*}
	\mathbb{E}\big[ \mathfrak{h}_T (x + 1) \big] = \displaystyle\frac{- 1}{2 \pi \textbf{i}} \displaystyle\oint \left( \displaystyle\frac{y - J - 1}{y} \right)^{x + 1} \left( \displaystyle\frac{y - 1}{y - J - 1} \right)^T dy,
	\end{flalign*}

	\noindent where the $y$ contour is positively oriented, contains $J + 1$, and avoids $0$. 
	
	Changing variables $y = z^{-1} \sqrt{T}$, and observing that $dy = - z^{-2} d z \sqrt{T}$, we obtain that 
	\begin{flalign*}
	\mathbb{E}\big[ T^{-1 / 2} \mathfrak{h}_T (x + 1) \big] = \displaystyle\frac{1}{2 \pi \textbf{i}} \displaystyle\oint \left( 1 - \displaystyle\frac{(J + 1) z }{\sqrt{T}} \right)^{x + 1} \left( 1 + \displaystyle\frac{zJ}{\sqrt{T} - (J + 1) z} \right)^T \displaystyle\frac{dz}{z^2} ,
	\end{flalign*}
	
	\noindent where the contour for $z$ is a positively oriented circle containing $T^{1 / 2} (J + 1)^{-1}$ and avoiding $0$. This contour can be deformed (without passing any poles of the integrand) to the union of two contours. The first $\mathcal{C}_T$ is the (top to bottom) interval $[1 + \textbf{i} T^{1 / 2} / (J + 1), 1 - \textbf{i} T^{1 / 2} / (J + 1)]$; the second $\mathcal{D}_T$ is the positively oriented union $[1 - \textbf{i} T^{1 / 2} / (J + 1), T^{1 / 2} - \textbf{i} T^{1 / 2} / (J + 1)] \cup [T^{1 / 2} - \textbf{i} T^{1 / 2} / (J + 1), T^{1 / 2} + \textbf{i} T^{1 / 2} / (J + 1)] \cup [T^{1 / 2} + \textbf{i} T^{1 / 2} / (J + 1), 1 + \textbf{i} T^{1 / 2} / (J + 1)]$. Thus, $\mathcal{C}_T \cup \mathcal{D}_T$ is a rectangle of height $2 T^{1 / 2} / (J + 1)$ and width $T^{1 / 2} - 1$. 
	
	Further observe that 
	\begin{flalign}
	\label{exponential1}
	\begin{aligned}
	1 - \displaystyle\frac{(J + 1) z}{\sqrt{T}} & = \exp \left( -\displaystyle\frac{(J + 1) z}{\sqrt{T}} \right) \left( 1 - \displaystyle\frac{(J + 1)^2 z^2}{2 T} + \mathcal{O} (T^{-3 / 2})  \right); \\
	 1 + \displaystyle\frac{zJ}{\sqrt{T} - (J + 1)z} & = \exp \left( \displaystyle\frac{zJ}{\sqrt{T}} \right) \left( 1 + \displaystyle\frac{(J^2 + 2J)z^2}{2 T} + \mathcal{O} (T^{-3 / 2}) \right), 
	 \end{aligned}
	\end{flalign}
	
	\noindent uniformly over $z$ in any fixed compact subset of $\mathbb{C}$. Moreover, using the fact that $x = \lfloor J r T / (J + 1) + s \sqrt{T} \rfloor$, it can be shown (through a proof we omit) that 
	\begin{flalign}
	\label{contourdc} 
	\displaystyle\sup_{T \ge 1}  \displaystyle\sup_{z \in \mathcal{C}_T \cup \mathcal{D}_T} \Bigg| \left( 1 - \displaystyle\frac{(J + 1) z }{\sqrt{T}} \right)^{x + 1} \left( 1 + \displaystyle\frac{zJ}{\sqrt{T} - (J + 1) z} \right)^T \Bigg| < \infty. 
	\end{flalign}
	
	\noindent Thus, again using the fact that $x = J r T / (J + 1) + s \sqrt{T}$ and also \eqref{exponential1} and \eqref{contourdc}, we obtain 
	\begin{flalign*}
	\mathbb{E}\big[ T^{-1 / 2} \mathfrak{h}_T (x + 1) \big] & = \displaystyle\frac{1}{2 \pi \textbf{i}} \displaystyle\oint_{\mathcal{C}_T \cup \mathcal{D}_T} \exp \big( -s (J + 1) z \big) \left( 1 + 	\displaystyle\frac{r J z^2}{2 T} + o (T^{-1}) \right)^{r T} \displaystyle\frac{dz}{z^2} + o (1) \\
	& = \displaystyle\frac{1}{2 \pi \textbf{i}} \displaystyle\oint_{\mathcal{C}_T \cup \mathcal{D}_T} \exp \left( \displaystyle\frac{r J z^2}{2} - s (J + 1) z + o (1	) \right) \displaystyle\frac{dz}{z^2} + o (1). 
	\end{flalign*}
	
	\noindent Now, we can let $T$ tend to $\infty$; in this limit, the integral over $\mathcal{D}_T$ tends to $0$, due to \eqref{contourdc} and the fact that $|z| > \sqrt{T} / (J + 1)$ when $z \in \mathcal{D}_T$. Thus, the contour for $z$ becomes $\mathcal{C}_{\infty} = 1 + \textbf{i} \mathbb{R}$ (from top to bottom), from which the theorem follows from the definition \eqref{hsr} of $\mathcal{H} (s, r)$. 
\end{proof}

\subsection{Asymptotics for the Dynamical \texorpdfstring{$(J; \gamma)$}{}-Partial Exclusion Process} 

\label{DynamicPartialExclusionAsymptotic}

The goal of this section is to access asymptotics for the current of the dynamical $(J; \gamma)$-partial exclusion process. In this setting the scaling changes, and we are instead interested in asymptotics for the current $\mathfrak{H}_T (s, r; \gamma)$ from \eqref{htsrgamma}. In particular, we would like to show the convergence \eqref{convergencedynamical1}, or equivalently (since $\mathfrak{H}$ is nonnegative) the convergence 
\begin{flalign}
\label{convergencedynamical}
\mathfrak{H}_T (s, r; \gamma) \big( \mathfrak{H}_T (s, r; \gamma) + s (J + 1) \big) \rightarrow \chi_{a, b}. 
\end{flalign}

The following theorem establishes \eqref{convergencedynamical} in the sense of moments; in what follows, recall the definition of the rational Pochhammer symbol $(a)_m$ from \eqref{productbasicelliptic}. 

\begin{thm}
	
	\label{momentmdynamic} 
	
	Let $m, J \in \mathbb{Z}_{\ge 1}$, and let $s, r, \gamma \in \mathbb{R}$, with $r > 0$ and $\gamma > J + 1$. Denoting $a = \gamma$ and $b = \sqrt{2 \pi / r J}$, we have that 
	\begin{flalign}
	\label{dynamicmomentsconverge}
	\displaystyle\lim_{T \rightarrow \infty} \mathbb{E} \Big[ \mathfrak{H}_T (s, r; \gamma)^m \big( \mathfrak{H}_T (s, r; \gamma) + s (J + 1) \big)^m  \Big] = \left( \displaystyle\frac{r J}{2 \pi}\right)^{m / 2} (\gamma)_m = \mathbb{E} [ \chi_{a, b}^m ],  
	\end{flalign}
	
	\noindent where the first expectation above is with respect to the dynamical $(J; \gamma)$-PEP and the second expectation is with respect to the $\chi_{a, b}$-distributed random variable. 
\end{thm}

\begin{proof}
	
	Recall that moments of a Gamma-distributed random variable $\chi_{a, b}$ are given by 
	\begin{flalign}
	\label{chiabmoments}
	\mathbb{E} \big[ \chi_{a, b}^m \big] = \displaystyle\frac{\Gamma (a + m)}{b^m \Gamma (a)} = b^{-m} (a)_m,  
	\end{flalign}
	
	\noindent for each integer $m$. Thus, the second equality in \eqref{dynamicmomentsconverge} follows from \eqref{chiabmoments}, so it suffices to establish the first. 
	
	To that end, we observe that the right side of the identity given by Corollary \ref{kexclusionidentity} is independent of $\gamma$. Thus, applying that lemma with $k = m$, $N = \lfloor r T \rfloor$, and $x_1 = x_2 = \cdots = x_m = x = \lfloor J r T / (J + 1) + s T^{1 / 4} \rfloor - 1$, and then equating its left sides with $\gamma$ equal to $\gamma$ and $\gamma$ equal to $\infty$, we obtain that 
	\begin{flalign*}
	& \displaystyle\frac{1}{(\gamma)_m} \mathbb{E} \left[ \displaystyle\prod_{j = 0}^{m - 1} \big( \mathfrak{h}_{\lfloor r T \rfloor} (x + 1) - j \big) \big( \mathfrak{h}_{\lfloor r T \rfloor} (x + 1) + s (J + 1) T^{1 / 4} + \gamma + j) \right] \\
	& \qquad \qquad = \mathbb{E} \left[ \displaystyle\prod_{j = 0}^{m - 1} \big( \widetilde{\mathfrak{h}_T} (x + 1) - j \big) \right].
	\end{flalign*} 
	
	Above, $\widetilde{\mathfrak{h}}$ denotes the height function of the non-dynamical exclusion process analyzed in Section \ref{NonDynamickAsymptotic}. Multiplying both sides by $T^{- m / 2} (\gamma)_m$, taking the limit as $T$ tends to $\infty$, and applying Theorem \ref{momentm} yields
	\begin{flalign}
	\label{dynamicmomentsconverge2}
	\begin{aligned}
	& \displaystyle\lim_{T \rightarrow \infty} \mathbb{E} \left[ \displaystyle\prod_{j = 0}^{m - 1} \big( \mathfrak{H} (s, r; \gamma) - j T^{-1 / 4} \big) \big( \mathfrak{H}_T (s, r; \gamma) + s (J + 1) + T^{-1 / 4} (\gamma + j) \big) \right] \\
	& \qquad = \mathcal{H} (0; r)^m (\gamma)_m = \left( \displaystyle\frac{rJ}{2 \pi} \right)^{m / 2} (\gamma)_m.
	\end{aligned}
	\end{flalign}
	
	\noindent In the second equality above, we used that $\mathcal{H} (0) = \sqrt{rJ / 2 \pi}$, which follows from the fact that $\mathcal{H} (s, r)$ solves the heat equation $J \partial_r \mathcal{H} = 2 (J + 1)^2 \partial_s^2 \mathcal{H}$ with initial data $\mathcal{H} (s, 0) = (J + 1) |s| \textbf{1}_{s < 0}$ (see the discussion directly after \eqref{hsr}). Now the first equality in \eqref{dynamicmomentsconverge} follows from \eqref{dynamicmomentsconverge2}; this establishes the theorem. 
\end{proof}

\appendix 

\section{Asymptotics for the Asymmetric Corner Growth Model} 

\label{AsymmetricCorner}

In this appendix we outline a proof of the following proposition, which determines the law of large numbers for the height function $\zeta$ of the non-dynamical, asymmetric midpoint corner growth model from Section \ref{DynamicalCornerGrowth}, and also establishes KPZ-type fluctuations of $\zeta$ around this limit shape.

\begin{prop}
	
	\label{asymmetriccornerfluctuations} 
	
	Consider the (non-dynamical) asymmetric midpoint corner growth model, as defined in Section \ref{DynamicalCornerGrowth}, in which the probability that a midpoint goes up is $p \in \big( \frac{1}{2}, 1 \big)$. Denote $q = (1 - p) / p$. Then, for any $s \in \big( \frac{1 - 2p}{2}, \frac{2p - 1}{2} \big)$, we have that
	\begin{flalign}
	\label{asymmetricasymptotic}
	\displaystyle\lim_{T \rightarrow \infty} \mathbb{P} \left[ \displaystyle\frac{M T - \zeta_T ( s T )}{\mathcal{F} T^{1 / 3}} \le r \right] = F_{\TW} (r), 
	\end{flalign}
	
	\noindent where $F_{\TW} (s)$ denotes the GUE Tracy-Widom distribution \cite{LSK} and
	\begin{flalign*}
	M = M (s) & = \displaystyle\frac{q + 1 - 2 \sqrt{q (1 - 4 s^2)}}{1 - q}; \\
	 \mathcal{F} = \mathcal{F} (s) & = \displaystyle\frac{2 q^{1 / 3} \Big(  \sqrt{\frac{1}{2} - s} - \sqrt{q \big( \frac{1}{2} + s \big)} \Big)^{2 / 3} \Big( \sqrt{\frac{1}{2} + s} - \sqrt{q \big( \frac{1}{2} - s \big)} \Big)^{2 / 3}}{ (1 - q)^{4 / 3} q^{1 / 6} \big( \frac{1}{2} + s \big)^{1 / 6} \big( \frac{1}{2} - s \big)^{1 / 6} }. 
	\end{flalign*}

\end{prop}

\begin{rem}
	\label{sq} 
	
	If $s \notin \big( \frac{1 - 2p}{2}, \frac{2p - 1}{2} \big)$, then it can be shown that $\zeta_T (s T) = 2 |s| T + o (T^{1 / 3})$ with probability $1 + o (1)$. 
\end{rem}

In view of \eqref{dynamicalparameterpartialexclusion} (specialized to the case $J = 1$) and Remark \ref{cornerpartialexclusion}, we can instead consider the current of the $(q; 1)$-asymmetric PEP as defined in Example \ref{dynamicalasymmetricmidpointcorner}. Specifically, to show \eqref{asymmetricasymptotic} it suffices to show that 
\begin{flalign}
\label{partialexclusionasymptotic}
\displaystyle\lim_{T \rightarrow \infty} \mathbb{P} \left[ \displaystyle\frac{m T - \mathfrak{h}_T (\lfloor \eta T \rfloor)}{f T^{1 / 3}} \le r \right] = F_{\TW} (r), 
\end{flalign} 

\noindent for any $\eta \in (1 - p, p)$, where 
\begin{flalign*}
m = m (\eta) = \displaystyle\frac{\big(\sqrt{1 - \eta} - \sqrt{q \eta} \big)^2}{1 - q}; \quad f = f (\eta) = \displaystyle\frac{q^{1 / 3} \big(  \sqrt{\eta} - \sqrt{q (1 - \eta)} \big)^{2 / 3} (\sqrt{1 - \eta} - \sqrt{q \eta})^{2 / 3}}{ (1 - q)^{4 / 3} q^{1 / 6} \eta^{1 / 6} (1 - \eta)^{1 / 6}}, 
\end{flalign*}

\noindent and the probability in \eqref{partialexclusionasymptotic} is taken with respect to the $(q; 1)$-asymmetric PEP. 

We explain a route to establishing \eqref{partialexclusionasymptotic} through a comparison between the stochastic higher spin vertex models \cite{HSVMRSF,SHSVML} and Schur measures \cite{IWRP}. This type of comparison was first observed in \cite{SHSVMM} to provide a new proof of the GUE Tracy-Widom current fluctuations in the stochastic six-vertex model with step initial data; later, it was used to establish similar asymptotics in a number of other works, including \cite{PTAEPSSVM,SHSVMAEP}. Here, we outline how to modify that proof so that it applies to asymmetric corner growth model. Since our work here will be similar to what was done in \cite{PTAEPSSVM,SHSVMM,SHSVMAEP}, we will not go through all details. 

First, we require some notation on symmetric functions; see Chapter 1 of Macdonald's book \cite{SFP} for a more thorough introduction. Denote the set of all partitions by $\mathbb{Y}$, and fix two (finite or infinite) sets of nonnegative real numbers $X = (x_1, x_2, \ldots ) \subset \mathbb{R}_{\ge 0}$ and $Y = (y_1, y_2, \ldots ) \subset \mathbb{R}_{\ge 0}$. For any $\lambda \in \mathbb{Y}$, let $h_{\lambda} (X), h_{\lambda} (Y)$ and $s_{\lambda} (X), s_{\lambda} (Y)$ denote the \emph{complete symmetric functions} and the \emph{Schur functions} associated with $\lambda$, respectively. Moreover, for any $z \in \mathbb{C}$, define the products 
\begin{flalign*}
F_X (z) = \displaystyle\sum_{n = 0}^{\infty} h_n (X) z^n = \displaystyle\prod_{j = 1}^{\infty} \displaystyle\frac{1}{1 - x_i z}; \qquad F_Y (z) = \displaystyle\sum_{n = 0}^{\infty} h_n (Y) z^n = \displaystyle\prod_{j = 1}^{\infty} \displaystyle\frac{1}{1 - y_i z},
\end{flalign*}

\noindent where we assume that the infinite sums and products converge.	

The \emph{Schur measure} $\textbf{SM}$ was introduced by Okounkov in \cite{IWRP} as a measure on $\mathbb{Y}$ weighted by products of Schur functions. Specifically, it is the probability measure that assigns weight 
\begin{flalign*}
\textbf{SM} (\lambda) = \textbf{SM}_{X; Y} \big( \{ \lambda \} \big) = Z^{-1} s_{\lambda} (X) s_{\lambda} (Y)
\end{flalign*}

\noindent to each $\lambda \in \mathbb{Y}$, where $Z = \prod_{j = 1}^n F_{x_j} (Y)$ is the normalization constant chosen to ensure that $\sum_{\lambda \in \mathbb{Y}} \textbf{SM} (\lambda) = 1$ (this follows from the \emph{Cauchy identity}). 

We will be interested in the Schur measure $\textbf{SM}_{X; Y}$ with parameters $X = (x_1, x_2, \ldots , x_T) = (1, 1, \ldots , 1)$ and $Y = (y_1, y_2, \ldots , y_{2 N}) = (1, q, 1, q, \ldots , 1, q,)$, where $N = \lfloor \eta T \rfloor$. In this case, 
\begin{flalign}
\label{xfyf}
F_X  (z) =  \left( \displaystyle\frac{1}{1 - z} \right)^T; \qquad F_Y (z) = \left( \displaystyle\frac{1}{1 - z} \right)^N \left( \displaystyle\frac{1}{1 - q z} \right)^N. 
\end{flalign}

Our interest in this measure is given by the following lemma. In what follows, $\ell (\lambda)$ denotes the \emph{length} (number of non-zero parts) of the partition $\lambda$, and we recall the notion of \emph{asymptotic equivalence} from Definition 4.2 of \cite{SHSVMM}.

\begin{lem}
	
	\label{lengthpartitionheight} 
	
	Sample a partition $\lambda$ randomly from the Schur measure $\textbf{\emph{SM}}$ defined above, and consider the current $\mathfrak{h}_t (x)$ of the $(q; 1)$-asymmetric PEP from Example \ref{dynamicalasymmetricmidpointcorner}. Then, the sequences of random variables $\big\{ \mathfrak{h}_T (N) \big\}_{T \ge 1}$ and $\big\{ T - \ell (\lambda) \big\}_{T \ge 1}$ are asymptotically equivalent. 
	
\end{lem}

\begin{proof}[Proof (Outline)]
	
	 In view of the discussion in Example \ref{dynamicalasymmetricmidpointcorner}, the $(q; 1)$-asymmetric PEP can be identified with the $q$-Hahn boson model (given in Section \ref{LambdaSpecializationModelGeneral}) with parameters $\delta = 0$, $b = q^{-2}$, and $a = q^{-1}$. Thus, this model can be identified with a (non-dynamical) stochastic higher spin vertex model (see, for example, Section 6.6 of \cite{HSVMRSF}) with spectral parameters $U = (1, 1, \ldots , 1)$ (where $1$ appears with multiplicity $T$); inhomogeneity parameters $\Xi = (q^{-1}, q^{-1}, \ldots )$; and spin parameters $S = (q^{-1}, q^{-1}, \ldots )$. Now, the lemma follows as a quick consequence of Corollary 4.11 of \cite{SHSVMM}; we omit further details. 
\end{proof}

The remainder of this section will be devoted a brief explanation of how to establish \eqref{lengthpartitionasymptotic}. What makes such a proof possible is that the configuration $\mathfrak{S} (\lambda) = \{ \lambda_k - k \}_{k \ge 1}$ forms a \emph{determinantal point process}; we refer to the survey \cite{DPP} for a precise definition of and more information about these processes. Its correlation kernel is known (see Theorem 2 of \cite{IWRP} or equations (21) and (22) of \cite{CPLGR}) to be given by a two-fold contour integral. Under our specialization, this kernel becomes 
\begin{flalign}
\label{determinantalkernel1}
\begin{aligned}
K (i, j) & = \displaystyle\frac{1}{(2 \pi \textbf{i})^2} \displaystyle\oint \displaystyle\oint \displaystyle\frac{F_X \big( w^{-1} \big) F_Y (v)}{F_Y (w) F_X \big( v^{-1} \big)} \displaystyle\frac{w^j dv dw}{v^{i + 1} (v - w)} \\ 
& = \displaystyle\frac{1}{(2 \pi \textbf{i})^2} \displaystyle\oint \displaystyle\oint \left( \displaystyle\frac{1 - v}{1 - w} \right)^T \left( \displaystyle\frac{(1 - w)(1 - wq)}{(1 - v)(1 - vq)} \right)^N \displaystyle\frac{w^{T + j} dv dw}{v^{T + i} (v - w)}, 
\end{aligned} 
\end{flalign} 

\noindent where $v$ and $w$ are integrated along positively oriented, simple, closed loops satisfying the following two properties. First, the contour for $w$ is contained inside the contour for $v$. Second, both contours enclose $0$ and $1$, but leave outside $q^{-1}$. 

In view of Lemma \ref{lengthpartitionheight}, to establish \eqref{partialexclusionasymptotic} it suffices to show that 
\begin{flalign}
\label{lengthpartitionasymptotic}
\displaystyle\lim_{T \rightarrow \infty} \mathbb{P} \left[ \displaystyle\frac{(m - 1) T + \ell (\lambda)}{f T^{1 / 3}} \le r \right] = F_{\TW} (r).  
\end{flalign}

Thus, we would like to understand the asymptotic behavior of $- \ell (\lambda)$. Since $- \ell (\lambda)$ is the smallest integer not contained in $\mathfrak{S} (\lambda)$, Kerov's complementation principle for determinantal point processes (see Proposition A.8 of \cite{AMSG}) implies that 
\begin{flalign}
\label{determinantalkernel3}
\mathbb{P} \big[ - \ell (\lambda) > R \big] = \det \big( \Id - \widetilde{K} \big)_{L^2 (\{ R, R + 1, \ldots \})},
\end{flalign} 

\noindent for any integer $R$, where $\widetilde{K} (i, j) = \textbf{1}_{i = j} - K (i, j)$. Since $\textbf{1}_{i = j}$ is the residue of the right side of \eqref{determinantalkernel1} at $w = v$, it follows that $- \widetilde{K} (i, j)$ is equal to the right side of \eqref{determinantalkernel1}, in which the contours for $v$ and $w$ are switched. More specifically,
\begin{flalign}
\label{determinantalkernel2} 
\begin{aligned}
\widetilde{K} (i + (m - 1) T & , j + (m - 1) T) = \displaystyle\frac{1}{(2 \pi \textbf{i})^2} \displaystyle\oint \displaystyle\oint \exp \Big( T \big( G(w) - G(v) \big) \Big) \displaystyle\frac{w^j d v d w}{\widetilde{v}^{i + 1} (v - w) }, 
\end{aligned}
\end{flalign}

\noindent where $w$ and $v$ are integrated along positively oriented, simple, closed loops satisfying the following two properties. First, the contour for $v$ is contained inside the contour for $w$. Second, both contours enclose $0$ and $1$, but leave outside $q^{-1}$. In \eqref{determinantalkernel2} we have denoted 
\begin{flalign*}
G(z) = (\eta - 1) \log (z - 1) + \eta \log (z q - 1) + m \log z.  	
\end{flalign*}

\noindent Now, using \eqref{determinantalkernel3}, \eqref{determinantalkernel2}, and the facts that 
\begin{flalign*}
G' (z) = \displaystyle\frac{q (\sqrt{\eta} - \sqrt{q (1 - \eta)})^2 (z - \vartheta)^2}{(1 - q) z (z - 1) (zq - 1)} \quad \text{and} \quad \displaystyle\frac{G''' (\vartheta)}{6} = \displaystyle\frac{1}{3} \left(\displaystyle\frac{f}{\vartheta} \right)^3, \quad \text{where} \quad \vartheta = \displaystyle\frac{\sqrt{1 - \eta} - \sqrt{q \eta}}{q \sqrt{1 - \eta} - \sqrt{q \eta}}, 
\end{flalign*}

\noindent one can show using a steepest descent analysis (by scaling around $\vartheta$ by a factor of $f T^{1 / 3} / \vartheta$) to establish \eqref{lengthpartitionasymptotic}. Asymptotic analyses of this type have by now become fairly standard (see, for example, \cite{PTAEPSSVM,RWBRE,AEP,SHSVMM,MP,SSVM,AGRSD,AMSG,SWLP,PTAP,IWRP,CPLGR,SHSVMAEP,LAEP}) and so we omit further details.

\end{document}